\documentclass{article}

\usepackage{amsmath}
\usepackage{amssymb}
\usepackage{amsthm}
\usepackage{amsfonts}
\usepackage{accents}
\usepackage[boxed,titlenumbered]{algorithm2e}
\usepackage[toc,page]{appendix}
\usepackage{authblk}
\usepackage[english]{babel}
\usepackage{braket}
\usepackage{booktabs}
\usepackage{bigdelim}
\usepackage{cancel}
\usepackage[shortlabels]{enumitem}
\usepackage{float}
\usepackage[utf8]{inputenc}
\usepackage[margin=1in]{geometry}
\usepackage{graphicx}
\usepackage[colorlinks=true,linktoc=all]{hyperref}
\usepackage{mathtools}
\usepackage{multirow}
\usepackage{physics}
\usepackage{tabularx}
\usepackage[dvipsnames]{xcolor}
\usepackage{tikz}
\usepackage[textsize=footnotesize]{todonotes}
\usepackage{qcircuit}
\usepackage{array}

\usepackage{color}
\usepackage[notext,nomath]{stix}
\usepackage{amsthm}
\usepackage{slashed}

\newcommand{\paren}[1]{\left( #1 \right)} 
\newcommand{\curly}[1]{\left\{#1\right\}} 
\newcommand{\like}[0]{\sim} 

\DeclareFontEncoding{LS2}          {}{\noaccents@}
\DeclareFontSubstitution{LS2}      {stix}{m}{n}
\DeclareSymbolFont{arrows3}        {LS2}{stixtt}   {m} {n}
\SetSymbolFont{arrows3}      {bold}{LS2}{stixtt}   {b} {n}
\DeclareMathSymbol{\ultriangle}    {\mathord}{arrows3}{"97}
\DeclareMathSymbol{\urtriangle}    {\mathord}{arrows3}{"98}
\DeclareMathSymbol{\lltriangle}    {\mathord}{arrows3}{"99}
\DeclareMathSymbol{\lrtriangle}    {\mathord}{arrows3}{"9E}
\DeclareMathSymbol{\ulblacktriangle}    {\mathord}{arrows3}{"F9}
\DeclareMathSymbol{\urblacktriangle}    {\mathord}{arrows3}{"FA}
\DeclareMathSymbol{\llblacktriangle}    {\mathord}{arrows3}{"FB}
\DeclareMathSymbol{\lrblacktriangle}    {\mathord}{arrows3}{"FC}

\usetikzlibrary{decorations.pathreplacing}
\usetikzlibrary{cd}
\usetikzlibrary{arrows.meta}
\usetikzlibrary{calc}

\newcommand{\gray}[1]{\textcolor{gray}{#1}}
\newcommand{\blue}[1]{\textcolor{blue}{#1}}
\newcommand{\orange}[1]{\textcolor{orange}{#1}}

\DeclarePairedDelimiter{\ceil}{\lceil}{\rceil}
\DeclarePairedDelimiter{\floor}{\lfloor}{\rfloor}

\newcolumntype{P}[1]{>{\centering\arraybackslash}p{#1}}

\makeatletter
\newtheorem*{rep@theorem}{\rep@title}
\newcommand{\newreptheorem}[2]{%
\newenvironment{rep#1}[1]{%
 \def\rep@title{#2 \ref{##1} (restated)}%
 \begin{rep@theorem}}%
 {\end{rep@theorem}}}
\makeatother

\newtheorem{lem}{Lemma}[section]
\newtheorem{thm}[lem]{Theorem}
\newtheorem{claim}[lem]{Claim}
\newtheorem{cor}[lem]{Corollary}
\newtheorem{defn}[lem]{Definition}
\newtheorem{note}[lem]{Note}
\theoremstyle{definition}
\newtheorem{alg}[lem]{Algorithm}
\newtheorem{obs}[lem]{Observation}
\newtheorem{fact}[lem]{Fact}

\newreptheorem{thm}{Theorem}
\newreptheorem{lem}{Lemma}
\newreptheorem{cor}{Corollary}
\newreptheorem{defn}{Definition}
\newreptheorem{claim}{Claim}
\newreptheorem{fact}{Fact}

\newcommand{\E}[1]{{\mathbb E}\left[#1\right]}
\newcommand{\psEop}{{\tilde{\mathbb E}}}
\newcommand{\psE}[1]{\psEop \left[#1\right]}

\newcommand{\p}[1]{{\mathbb P}\left\{#1\right\}}

\newcommand{\bbF}{{\mathbb F}}

\newcommand{\bbR}{{\mathbb R}}
\newcommand{\bbZ}{{\mathbb Z}}

\newcommand\cC{{\cal C}}

\newcommand\cE{{\cal E}}

\newcommand\cO{{\cal O}}

\newcommand\cQ{{\cal Q}}

\newcommand\cS{{\cal S}}
\newcommand\cT{{\cal T}}

\newcommand\cY{{\cal Y}}

\newcommand\eps{{\varepsilon}}

\newcommand{\all}{ \; \forall \;}

\newcommand{\transp}{T}

\DeclareMathOperator{\shuffle}{shuffle}

\DeclareMathOperator{\unpack}{unpack}
\DeclareMathOperator{\poly}{poly}
\DeclareMathOperator{\im}{im}

\DeclareMathOperator{\Ptime}{\mathsf{P}}
\DeclareMathOperator{\NP}{\mathsf{NP}}
\DeclareMathOperator{\UGC}{\mathsf{UGC}}

\newcommand{\nc}{\newcommand}
\nc{\rnc}{\renewcommand}
\nc\benum{\begin{enumerate}}
\nc\eenum{\end{enumerate}}

\nc\ot{\otimes}
\nc\bit{\begin{itemize}}
\nc\eit{\end{itemize}}

\newcommand{\lemref}[1]{Lemma~\ref{lem:#1}}
\newcommand{\thmref}[1]{Theorem~\ref{thm:#1}}

\newcommand{\defref}[1]{Definition~\ref{def:#1}}
\newcommand{\corref}[1]{Corollary~\ref{cor:#1}}

\newcommand{\tabref}[1]{Table~\ref{tab:#1}}
\newcommand{\claimref}[1]{Claim~\ref{claim:#1}}
\def\be#1\ee{\begin{equation}#1\end{equation}}
\def\ba#1\ea{\begin{align}#1\end{align}}
\def\bas#1\eas{\begin{align}#1\end{align}}
\def\bpm#1\epm{\begin{pmatrix}#1\end{pmatrix}}

\def\be#1\ee{\begin{equation}#1\end{equation}}
\def\bea#1\eea{\begin{eqnarray}#1\end{eqnarray}}
\def\beas#1\eeas{\begin{eqnarray*}#1\end{eqnarray*}}
\def\ba#1\ea{\begin{align}#1\end{align}}
\def\bas#1\eas{\begin{align}#1\end{align}}
\def\bpm#1\epm{\begin{pmatrix}#1\end{pmatrix}}
\def\bbm#1\ebm{\begin{bmatrix}#1\end{bmatrix}}
\def\bbml#1\ebml{\begin{bmatrix*}[l]#1\end{bmatrix*}}

\nc\ppsim{\stackrel{p}{\sim}}
\nc\psesim{\stackrel{\psEop}{\sim}}

\nc{\nli}{\sigma_2}
\nc{\vcstrat}{\hat{\eta}} 
\nc{\vqstrat}{\hat{\theta}} 
\nc{\cYtwo}{\cY_{2}}
\nc{\cYQ}{\cY_{Q}}
\nc{\MERP}{\text{MERP}}

\nc{\co}[0]{commuting-operator}
\nc{\Co}[0]{Commuting-operator}
\nc{\CO}[0]{Commuting-Operator}
\nc{\PR}[0]{PREF} 
\nc{\PREF}[0]{PREF}

\nc{\APDK}{\mathrm{APD}_{K}}
\nc{\CGn}{\mathrm{CG}_{n}}
\nc{\CG}{\mathrm{CG}}
\nc{\APD}{\mathrm{APD}}

\def\begsub#1#2\endsub{\begin{subequations}\label{eq:#1}\begin{align}#2\end{align}\end{subequations}}

\title{Algorithms, Bounds, and Strategies for Entangled XOR Games}

\author{Adam Bene Watts\thanks{abenewat@mit.edu}}
\author{Aram W.~Harrow\thanks{aram@mit.edu}}
\author{Gurtej Kanwar\thanks{gurtej@mit.edu}}
\author{Anand Natarajan\thanks{anandn@mit.edu}}
\affil{Center for Theoretical Physics, MIT}

\begin{document}
\maketitle

\begin{abstract}
  We study the complexity of computing the \co{} value $\omega^*$ of entangled XOR games
  with any number of players. We introduce necessary and sufficient criteria for
  an XOR game to have $\omega^* = 1$, and use these criteria to derive the
  following results:
  \begin{enumerate}
  \item
    An algorithm for symmetric games that  decides in polynomial time whether $\omega^* = 1$ or $\omega^* < 1$, a task that was not
  previously known to be decidable, together with a simple tensor-product
  strategy that achieves value 1 in the former case.  The only previous candidate algorithm for this
  problem was the Navascu{\'e}s-Pironio-Ac{\'i}n (also known as noncommutative Sum of Squares or ncSoS)
  hierarchy, but no convergence bounds were known. 
\item A family of games with three players and with $\omega^* < 1$, where it takes doubly exponential
  time for the ncSoS algorithm to witness this (in contrast with our algorithm which runs in
  polynomial time).
\item A family of games achieving a bias difference $2(\omega^* - \omega)$ arbitrarily
  close to the maximum possible value of $1$ (and as a consequence, achieving
  an unbounded bias ratio), answering an open question of Bri\"{e}t and Vidick.
\item
  Existence of an unsatisfiable phase for random
  (non-symmetric) XOR games: that is, we
  show that there exists a constant $C_k^{\text{unsat}}$ depending only on the number $k$
  of players, such that a random $k$-XOR game over an alphabet of size $n$ has
  $\omega^* < 1$ with high probability when the number of clauses is above
  $C_k^{\text{unsat}} n$. 
\item 
  A lower bound of  $\Omega(n \log(n)/\log\log(n))$ on the number of levels in the ncSoS hierarchy required to
  detect unsatisfiability for most random 3-XOR games.  This is in
  contrast with the classical case where the $n$-th level of the sum-of-squares
  hierarchy is equivalent to brute-force enumeration of all possible solutions.  
\end{enumerate}  

\end{abstract}
\vfill
\thispagestyle{empty}
\pagebreak

\tableofcontents
\vfill
\thispagestyle{empty}
\pagebreak

\setcounter{page}{1}

\section{Background}
\label{sec: Background}

Constraint satisfaction problems (CSPs) are a fundamental object of study in
theoretical computer science. 
In quantum information theory, there are
two natural analogues of CSPs, which both play important roles:
local Hamiltonians and (our focus) non-local games. 
Non-local games
originate from Bell's pioneering 1964 paper, which showed how to test for
quantum entanglement in a device with which we can interact only via classical
inputs and outputs. In modern language, the tests developed by Bell are games: a
referee presents two or more players with classical questions drawn from some
distribution and demands answers from them. Each
combination of question and answers receives some score and the players
cooperate (but do not communicate) in order to maximize their expected score.
These games are interesting because often the players can win the game with a higher
probability if they share an entangled quantum state, so a high average score can certify the presence of quantum
entanglement. Such tests are not only of scientific interest, but have had wide
application to proof
systems~\cite{CHTW04,ji2015classical}, quantum key distribution~\cite{E91,BHK-QKD05,VV14}, delegated computation~\cite{RUV13}, randomness generation~\cite{Col06} and
elsewhere. 

To be able to use a nonlocal game as a test for entanglement, it is essential to
be able to approximately compute two quantities: the best possible expected score when the players share either classical correlations or
entangled states, respectively called the ``classical'' and ``quantum'' (or ``entangled'')
values of the game, and denoted $\omega$ and $\omega^*$. To understand these quantities, think of a game with $k$
players as inducing a constraint satisfaction problem with a $k$-ary
predicate. Each question in the game is mapped to a variable in the CSP, and
each $k$-tuple of questions and set of accepted responses (a ``clause'') 
asked by the referee corresponds to a constraint. Classically, a
simple convexity argument shows that the players can
always stick to \emph{deterministic} strategies, where each question is assigned
a fixed answer; thus, $\omega$ is in fact identical to the value of
the CSP. Hence, thanks to various dichotomy theorems, we have a good understanding of
the difficulty of computing $\omega$: in some cases, we know a $\Ptime$
algorithm, and for most others, we know it is $\NP$-complete.

The quantum value $\omega^*$
is not as well understood. The main obstacle is that the set
of entangled strategies is very rich: the ``assignment'' to each variable is no
longer a value from a discrete set, but a linear operator over a Hilbert
space of potentially unbounded dimension.  As a result, we can say very little
in terms of upper bounds on the complexity of computing $\omega^*$. In fact, it is not known
whether even a constant-factor (additive) approximation to $\omega^*$
is Turing-computable. For general games, the best we can say is that it is
recursively enumerable: there is an algorithm, called the NPA or ncSoS hierarchy~\cite{NPA08,DLTW08}, that in the limit of infinite time
converges from above to the quantum value, but with no bound on the speed of
convergence. On the hardness side, more is known, and what we know is grounds
for pessimism: we have been able to show hardness results for approximating $\omega^*$
matching (e.g.~\cite{Vid13}) and in some cases exceeding (e.g.~\cite{ji2015classical}) the classical case by
constructing special games that force entangled players to use particular
strategies. Moreover, families of games have been found for which deciding whether $\omega^*
= 1$ is uncomputable~\cite{Slofstra16}. 
There are a few exceptions for which some positive results are known: for instance, the class of XOR games, in which
the answers are bits and the payoff depends only on their XOR (for any given set
of questions).  In the classical case, these games are as hard as general games
except in the ``perfect completeness'' regime: distinguishing $\geq 1-\eps$
satisfiability from $\leq \frac 12+\eps$ satisfiability is $\NP$-complete, but we
can determine whether an XOR game is perfectly satisfiable in polynomial time
using Gaussian elimination over $\bbF_2$. However, in the quantum case, it was shown
by Tsirelon~\cite{cirel1980quantum,tsirel1987quantum} that for two-player XOR games, the lowest level of
the  ncSoS algorithm converges exactly to the quantum value, rendering it
computable in polynomial time via semidefinite programming. (A similar technique
was also applied to approximating the entangled value of unique
games~\cite{KRT10}.) Yet these techniques seemed
unlikely to generalize to three or more players: it is known that distinguishing
$\geq 1 -\eps$ satisfiability from $\leq \frac{1}{2} + \eps$ for an entangled 3-player
XOR game is $\NP$-hard~\cite{Vid13}, and deciding the existence of perfect
strategies for the closely-related family of linear systems games is
uncomputable~\cite{Slofstra16}. For a summary of these results, see
Table~\ref{table:complexity}.


Another question which has been very fruitful in the study of classical CSPs is
understanding the typical value of a random instance. Research in this direction draws significantly on
insights from statistical mechanics and has proven that there exist sharp
satisfiable/unsatisfiable thresholds for random $k$-SAT and related games (often
using the equivalent constraint-satisfaction formulation).  But these techniques
do not carry over to the quantum case.  For random classical games, a basic method
is to look at the expected number of winning responses (the ``first moment
method'') or the variance (the ``second moment method'') as we randomize the
payoff function within some family such as random $k$-SAT or random $k$-XOR.
This suffices, for example, to show that random $3$-XOR games with $n$ variables and
$Cn$ clauses are satisfiable with high probability if and only if $C \lessapprox 0.92$~\cite{DM02}. Since quantum strategies do not form a discrete (or even
finite-dimensional) set, these methods are not possible.  Nor is it obvious how
to use more refined tools such as Shearer's Lemma or the Lovasz Local Lemma,
which address the question of when sets of overlapping constraints can be
simultaneously satisfied.  Indeed there
are famous examples (such as the Magic Square game) of quantum ``advantage''
(i.e.~the quantum value of a game is higher than the classical value) when there
exist strategies for apparently contradictory constraints that succeed with
probability 1.  These suggest that the barriers to extending our classical
intuition are not merely technical but reflect a genuinely different set of
rules. 

{\renewcommand{\arraystretch}{1.3}
\begin{table}[ht]
\begin{center}
\label{table:complexity}
\caption{Complexity of determining whether the  value of various games
  is $\geq c$ or $\leq c-\eps$.
}
\begin{tabular}{p{7em} | >{\raggedright\arraybackslash}p{18em}>{\raggedright\arraybackslash}p{17em}}
\toprule
Game type & Classical value & Entangled value \\
\midrule
2-XOR & $c=1$ in $\Ptime$ \newline $c<1$ $\Ptime$ or $\NP$-complete depending on $\eps$
& exact answer in $\Ptime$~\cite{cirel1980quantum,tsirel1987quantum}\\
Unique Games  & $c=1$ in $\Ptime$ \newline $c<1$ $\NP$-complete assuming
  $\UGC$
& $\UGC$-quality approximation in $\Ptime$~\cite{KRT10} \\
$k$-XOR, $k\geq 3$ \newline (symmetric)
 & $c=1$ in $\Ptime$ \newline $c<1$ $\NP$-complete
 & $c=1$ in $\Ptime$ (this paper, \thmref{exists_decidability_alg})\newline
$c<1$ $\NP$-hard~\cite{Vid13} \\
general & $c=1$, $c<1$ \newline both $\NP$-complete
& undecidable~\cite{Slofstra17}, \newline recursively enumerable~\cite{NPA08,DLTW08}
\\
\bottomrule
\end{tabular}
\end{center}
Unique Games and the Unique Games Conjecture ($\UGC$) are a generalization of 2-XOR games
that are defined in~\cite{trevisan2012khot}.
 If the $\UGC$ holds, then
it is $\NP$-complete to achieve any approximation ratio better than that known to be achieved
by existing algorithms.  For $k$-XOR with $k\geq 3$ it is $\NP$-complete to beat the trivial approximation.
\end{table}}

\section{Results}
\label{sec: Results}
In this section we informally describe our main results, and then give precise theorem
statements with links to proof sketches and full proofs later in the paper.

Our work introduces new techniques that let us make progress on the study
of both worst-case complexity and random instances of XOR games with more than
two players, in the regime where we are trying to decide whether $\omega^* = 1$.
We think of a nonlocal game as a system of equations whose
variables are linear operators, corresponding to the quantum measurements used
by the players; a strategy is a solution to this system. Our main innovation is
to consider a ``dual'' system of equations, whose solutions are objects that we
call \emph{refutations}. A refutation is a proof that the ``primal'' system of
operator equations induced by the game is infeasible, and thus that $\omega^*
\neq 1$. Surprisingly, for games that are symmetric under exchange of the
players, we show that the dual system reduces to a set of linear
equations over $\mathbb{Z}$, which can be solved efficiently. This leads to our
first result (Theorem~\ref{thm:exists_decidability_alg}), an algorithm for efficiently deciding whether $\omega^* = 1$ for a
symmetric $k$-player XOR game, which brings the best known upper bound on this problem
down from recursively enumerable~\cite{NPA08,DLTW08} to $\Ptime$. See Table~\ref{table:complexity}
for a summary of how our result fits in with known upper and lower bounds.
Subsequently, by taking the dual of the dual, we are able to explicitly construct a set of
quantum strategies (we call these Maximal Entanglement, Relative Phase, or MERP,
strategies) that attain value 1 for all symmetric games with $\omega^* = 1$
(Theorem~\ref{thm:merp}). An explicit example shows that the symmetry assumption is indeed necessary for our algorithm to work: we exhibit a simple non-symmetric game called the 123 game, for which a simple, non-MERP strategy achieves $\omega^* = 1$, while our algorithm is unable to detect this (Theorem~\ref{thm:123_intro}).

Our understanding of refutations and characterization of value-1
symmetric entangled games also lets us construct two specific families of
games with interesting properties. The first, Capped GHZ ($\CG$),
is a family where ncSoS takes $\exp(n)$ levels
and $\exp(\exp(n))$ time to detect
that $\omega^* < 1$ (Theorem~\ref{thm:Capped GHZ general result}), in contrast to our
algorithm which runs in polynomial time. 
The second, Asymptotically Perfect Difference ($\APD$),
is an explicit, deterministic family of $k$-XOR games with $\omega^* = 1$ and classical value
$\omega \rightarrow 1/2$ in the limit of large $k$ (\thmref{APD-violation}). In comparison,
there are randomized constructions of families of games whose bias ratio $\frac{\omega^* -
  1/2}{\omega-1/2}$ diverges for fixed $k\geq 3$ as $n\rightarrow\infty$~\cite{PWPVJ08,BrietV13}.
However, known examples of these constructions involve both $\omega^* \rightarrow 1/2$ and
$\omega \rightarrow 1/2$, potentially precluding experimental distinguishability.
To our knowledge this is the first construction of a family of XOR games
that asymptotically achieves a perfect bias \emph{difference}, $2(\omega^* - \omega) \rightarrow 1$,
addressing one of the main open questions presented in \cite{BrietV13}.

For random instances, we analyze the dual system to show the existence of an unsatisfiable (i.e. $\omega^*<1$) phase as in the classical case
(Theorem~\ref{thm: k-XOR non-sym unsat}). We also relate our methods to the ncSoS hierarchy.
 For random instances, we show that in the
unsatisfiable phase, a superlinear number of levels of ncSoS is necessary to certify that
$\omega^* < 1$ (Theorem~\ref{thm:sos-LB}).

The theorem statements of these results are as follows.
 Proof sketches are in Section~\ref{sec: Overview}.

\begin{thm}
\label{thm:exists_decidability_alg}
  There exists an algorithm that, given a $k$-player symmetric XOR game $G$
  with alphabet size $n$ and $m$ clauses, decides
  in time $\poly(n,m)$ \footnote{
  Note that $m$ and $k$ do not scale independently for symmetric games.
  Any symmetric game may be specified by $m'$ base clauses that are
  then symmetrized via at most $k!$ permutations each,
  meaning $m \leq k! m'$.
  We could thus naively rewrite this runtime as $\poly(n,k!m')$
  to extract the $k$ dependence. Because the core information about
  the symmetric game is really only contained in the $m'$ clauses,
  one might expect that it is possible to remove the factor of $k!$, and
  we hope to address this in a future work.
  }
  whether $\omega^*(G) = 1$ or
  $\omega^*(G) < 1$. 
\end{thm}
\begin{proof}
Section \ref{section: decidability algorithm}. Sketch in \ref{subsec: noPES games}.
\end{proof}

\begin{thm}
For every $k$-XOR game $G$ for which the algorithm of Theorem \ref{thm:exists_decidability_alg} shows $\omega^*(G) = 1$, there exists a
$k$-qubit tensor-product strategy achieving value $1$, and a description of the strategy can be computed in
polynomial time.
\label{thm:merp}
\end{thm}
\begin{proof}
Section \ref{sec:MERP Strategies}. Sketch in Sections \ref{subsec:Intro MERP Strategies} and \ref{subsec: Intro MERP-PES duality}.
\end{proof}

\begin{thm}
\label{thm:123_intro}
There exists a 6-player XOR game $G$ with alphabet size $3$ and $6$ clauses, for which $\omega^*(G) = 1$ but the algorithm of Theorem~\ref{thm:exists_decidability_alg} cannot detect this.
\end{thm}
\begin{proof}
Section~\ref{subsec: 123 Game}.
\end{proof}

\begin{thm}
\label{thm:Capped GHZ general result}
There exists a family of 3-XOR games with $\omega^* < 1$ but
for which the minimum
refutation length scales exponentially in the number of
clauses $m$ and alphabet size $n$. For these games exponentially many levels of ncSoS
are needed to witness that $\omega^* < 1$.
\end{thm}
\begin{proof}
Section \ref{subsec: Capped GHZ Game}.
\end{proof}

\begin{thm}\label{thm:APD-violation}
There exists a family of $k$-XOR games, parametrized by $K$,
for which $\omega^*(G(K)) = 1$ and the classical value is bounded by
\begin{align}
\dfrac 12 \leq \omega(G(K)) \leq  \dfrac{1}{2} + \sqrt{\dfrac{K+1}{2^{K+1}}}
\leq \dfrac{1}{2} + \sqrt{\dfrac{\log{k}}{k}}
.
\end{align}
\end{thm}
\begin{proof}
Section \ref{subsec: APD Games}.
\end{proof}

\begin{thm}
\label{thm: k-XOR non-sym unsat}
For every $k$, there exists a constant $C_k^{\text{unsat}}$ depending only on $k$ such that a random $k$-XOR game $G$ with $m \geq C_k^{\text{unsat}} n$ clauses has value $\omega^*(G) < 1$ with probability
$1 - o(1)$.
\end{thm}
\begin{proof}
Section \ref{subsec: Random Games}.
\end{proof}

\def\ThmSoSLB{
For any constant $C$, the minimum length refutation of a random  3-XOR
game with $m = Cn$ queries on an alphabet of size $n$ has length at least 
\begin{align}
\frac{en\log(n)}{8C^2 \log(\log(n))} - o\left(\frac{n\log(n)}{\log(\log(n))}\right) 
\end{align}
with probability $1 - o(1)$ (as $n \rightarrow \infty$). Hence, either $\omega^*=1$ or
$\Omega(n \log(n) / \log(\log(n)))$ levels of the ncSoS hierarchy are needed to witness
that $\omega^* < 1$ for such games.  }
\begin{thm}\label{thm:sos-LB}
\ThmSoSLB
\end{thm}
Note that we can choose $C \geq C_3^{\text{unsat}}$  (with $C_3^{\text{unsat}}$ from
\thmref{ k-XOR non-sym unsat}) such that for large enough $n$, typical random instances
will have $\omega^*<1$ but ncSoS will require $\Omega(n\log(n)/\log(\log(n)))$ levels to
detect this.
\begin{proof}
Section \ref{subsec: lower bounds}.
\end{proof}

\section{Future Work} 
\label{sec: Future Work}
We see four main directions in which our characterization of non-local
XOR games could be extended.

First, our linear algebraic
characterization of $\omega^* = 1$ games 
is incomplete: there exist games with $\omega^* = 1$ for which a MERP strategy
cannot achieve value 1. We expect a strengthing of the $\PREF{}$ condition may allow us to extend our decidability algorithm to detect these games. An approach to this
conjecture is described in Section~\ref{sec: Refutations} and a conjectured element
of this set of games is given in Section~\ref{subsec: 123 Game}.
Understanding the structure of such games would give further intuition 
about the behavior of optimal XOR \co{} strategies, in particular
strategies which may require more entanglement than the simple
explicit strategies in \thmref{merp}.


Second, our results leave open the possibility that determining whether $\omega^*=1$ for
nonsymmetric XOR games is outside $\Ptime$ or even undecidable.  In the realm of Binary
Constraint System (BCS) games, \cite{Slofstra16} shows that determining whether a general
BCS game has perfect value is undecidable.  The structural similarity between BCS games
and XOR games suggests that perhaps some of the group theoretic techniques of that work
could be applied to XOR games to arrive at a similar conclusion.  An interesting
class of games which may serve as an intermediate class between XOR and BCS games are ``incomplete'' XOR games in which there are $k$ players but queries can
involve $<k$ variables, effectively ignoring some players. Even for $k=2$,
Tsirelson's semidefinite programming characterization of $\omega^*$ does not
apply to incomplete XOR games, although in this case it is still easy to decide
whether $\omega^*=1$.

Thirdly, while in this work we have focused on computing the entangled game value
$\omega^*$, our methods may also be useful from the perspective of Bell inequalities,
in which the quantity of interest is the maximal violation achievable by an
entangled strategy. While this has conventionally been measured in terms of the bias
ratio $(\omega^* - 1/2)/(\omega - 1/2)$, the difference $2(\omega^* - \omega)$
is an equally natural measure of violation, and we hope that our techniques
will render it more amenable to analysis. Indeed, in addition to the
construction of Asymptotically Perfect Difference games mentioned above, our
results have the following simple consequence: for symmetric games with
$\omega^*= 1$, our characterization of the optimal strategies (MERP) together with the
Grothendieck-type inequality of~\cite{BBLV09} imply that the
bias ratio and difference are both bounded by constants depending only
on $k$, and that for the difference, this constant is strictly
less than one.

Finally, our results are almost entirely concerned with the question of whether
$\omega^*=1$ or $\omega^*<1$.  However, we note that the class of strategies appearing in
\thmref{merp} include the optimal strategy for the CHSH game~\cite{chsh1969}, but not for
all XOR games~\cite{OV16}.  It is an interesting open question to find a natural
characterization of games with $\omega^* < 1$ for which MERP strategies are optimal.  In
this setting there are still many classical tools which we do not know how to extend to
the classical case.  As an example, consider {\em overconstrained} games in which there
are many more constraints than variables and we choose the signs of those constraints
randomly.  In the classical case, the value is shown to be close to $1/2$ in
\thmref{asymptotic classical value bound} while in the quantum case we can only conclude
that it is $<1$ in \lemref{rank nullity PESc}.

\vfill{}
\pagebreak
\section{Notation}
\tabref{Notation} defines common notation used throughout the paper.
This section is intended as a reference for the reader, while future
sections provide more detailed technical definitions of these
concepts.

\begin{table}[h]
{\setlength{\tabcolsep}{0.5em}     
\renewcommand{\arraystretch}{1.2} 
\newcommand{\mcbar}[1]{\multicolumn{1}{|c}{#1}} 
\raggedright            
\begin{tabularx}{\linewidth}{c c X}
\hline
 & Symbol              & Definition  \\ \hline
\multirow{13}{5em}{\parbox{5em}{\centering XOR\\games}}&&\\
 &\mcbar{$G$}                 & An XOR game. Consists of an indexed set of clauses. \\
&\mcbar{$m$}               & The number of queries or clauses in an XOR game. \\
&\mcbar{$n$}               & The question alphabet size or number of variables of an XOR game.   \\
&\mcbar{$k$}               & The number of players in an XOR game. We often use the phrase ``$k$-XOR game'' to implicitly specify $k$. \\
&\mcbar{$q_i$}             & The $i$-th query of a game. Written as a length-$k$
vector. The $\alpha$-th entry $q_i^{(\alpha)} \in [n]$ specifies the question sent to player
$\alpha$. \\
&\mcbar{$s_i$}               & The parity bit $\in \curly{-1,1}$ of the $i$-th query in a game.  \\
&\mcbar{$c_i = (q_i,s_i)$}               & The $i$-th clause of a game. Written as a vector of length $k+1$ that collectively represents the $i$-th query and $i$-th parity bit.  \\
&\mcbar{$\omega(G)$}       & The classical value of XOR game $G$. \\
&\mcbar{$\omega^*(G)$}     & The \co{} value of XOR game $G$. \\
\multirow{6}{5em}{\parbox{5em}{\centering Linear\\algebra}}&&\\
&\mcbar{$A$} & The $m \times kn$, $\curly{0,1}$-valued game matrix for an XOR game.
If $q_i^{(\alpha)} = j$ then $A_{i,(\alpha-1)n + j} = 1$, otherwise $A_{i,(\alpha-1)n + j} = 0$. \\
&\mcbar{$\hat{s}$} & The length-$m$, $\curly{0,1}$-valued parity bit vector with entries
$\hat s_i = \frac{1+s_i}{2}$.\\
&\mcbar{$G \sim (A,\hat{s})$} & The XOR game corresponding to game matrix $A$ and parity bit vector $\hat{s}$. \\

\multirow{9}{5em}{\parbox{5em}{\centering Quantum\\strategies}}&&\\
& \mcbar{$\{O^{\alpha}_{\pm 1}(j)\}$} & The measurement made by the $\alpha$-th player upon receiving question $j \in [n]$. \\
& \mcbar{$O^{\alpha}(j)$} & The strategy observable for player $\alpha$
	on question $j$. Defined by the Hermitian operator
    $O^{\alpha}(j) := O^{\alpha}_{1}(j) - O^{\alpha}_{-1}(j)$. \\
&\mcbar{$Q_i$}        & The Hermitian operator
                                          representing
                                          the collective measurement made by the
                                          players of an XOR game upon
                                          receiving query
                                          $q_i$. Defined by $Q_i :=
                                          \prod_\alpha{O^{\alpha}(q_i^{(\alpha)})}$.  \\
& \mcbar{$\Ket{\Psi}$} & The shared state between the $k$ players of an XOR game. \\
 
\multirow{7}{5em}{\parbox{5em}{\centering Combina-\\torics}}&&\\
& \mcbar{$W$} & A word. In general a $k$-tuple with each element of the tuple a
concatenation of questions (``letters'') drawn from $[n]$. We frequently refer to the
$\alpha$-th entry of $W$ by $W^{(\alpha)}$, and the specific letter at row $i$ and offset $j$ by $w_{ij}$. \\
& \mcbar{$W^\dagger$} & The reversed form of a word. \\
& \mcbar{$\sim$} & Equivalence under neighboring pairs of letters canceling. \\
& \mcbar{$\ppsim$} & Equivalence under $\sim$ and parity-preserving permutations.
\end{tabularx}}

\caption{Notation used throughout the paper.}
\label{tab:Notation}
\end{table}\vfill
\pagebreak

\section{Technical Overview}
\label{sec: Overview}

We begin by formally defining a $k$-XOR game and its classical and quantum values.
\begin{defn}
\label{defn: k-XOR game}
Define a \textbf{clause} $c = (q, s)$ to be any
  $(k+1)$-tuple consisting of a \textbf{query} $q \in [n]^k$ and \textbf{parity bit} $s \in \{-1,1\}$.
  In a $k$-\textbf{XOR game} $G$ associated
  with a set of clauses $M$, a verifier selects a clause $c_i = (q_i, s_i)$
  uniformly at random from $M$. Next, the question $q_{i}^{(\alpha)}$ is sent to the
  $\alpha$-th player of the game, for all $\alpha \in [k]$. The players then each send back a
  single output $\in \{-1,1\}$, and win the game if their outputs multiply to $s_i$.
\end{defn}
The key property of a game $G$ is its value -- the maximum win probability achievable by players who cooperatively choose a strategy before the game starts, but cannot communicate while the game is being played. We distinguish various versions of the value by physical restrictions placed on the players.
\begin{defn}
\label{defn: k-XOR values}
  For a given game $G$, the \textbf{classical value} $\omega(G)$ is the
  maximum win probability achievable by players sharing no entanglement.
  
  The \textbf{tensor-product value} is the maximum
  win probability obtainable by players who share a quantum state but are
  restricted to making measurements on
  distinct factors of a tensor-product Hilbert space.
  
  Finally, the \textbf{\co{} value} $\omega^*(G)$ is the
  maximum win probability obtainable by players who may make any
  commuting measurements on a shared quantum state, not necessarily over
  a tensor-product Hilbert space. $\omega^*(G)$
  is often also referred to as the field-theoretic value of $G$.
\end{defn}

The \co{} value may differ from the tensor-product value of $G$~\cite{Slofstra16}, and the
question of whether it can differ from the closure of the set of values achievable by
tensor product strategies remains open\footnote{And hard! For general games this question
  is known to be equivalent to Connes' embedding
  conjecture~\cite{FritzNT14}.}. In this paper, we focus primarily on a
description of the commuting-operator value but in many cases can show that it coincides
with the tensor-product value.

For the purpose of analyzing both the classical and \co{} value of $k$-XOR games, we find it useful to define a linear algebraic representation for the game\footnote{There seems to exist an interesting parallel between this linear algebraic representation of an XOR game and the linear algebraic specification of a BCS game. While interesting, it is not explored in this work aside from its brief mention here and in Section \ref{sec: Future Work}.}. The linear algebraic view represents queries as a matrix and parity bits as a vector. In doing so, it abstracts away from the specifics of labels and player/query indices to reveal the underlying game structure.

\begin{defn}
\label{def:game matrix}
Given a $k$-XOR game with $m$ queries and alphabet size $n$,
define the \textbf{game matrix} $A$ as an $m \times kn$ matrix
describing query-player-question incidence. Specifically, $A$ can
be written as a segmented matrix with
$k$ distinct column blocks of size $n$ each, where
the $j$th column of block $\alpha$ consists of a 1 in
row $i$ if the $\alpha$th player receives question $j$ for query
$q_i$, and a 0 otherwise:
\begin{equation}
A_{i,(\alpha-1) n + j} := \begin{cases}
	1 \text{ if } q_i^{(\alpha)} = j \\
    0 \text{ otherwise}
 	\end{cases}.
\end{equation}

For such a game, define the length-$m$ \textbf{parity bit vector} $\hat{s} \in \mathbb{F}_2^{m}$ by
\be
\hat{s}_i := \begin{cases}
	0 \text{ if } s_i = 1 \\
	1 \text{ if } s_i = -1
	\end{cases}.
\ee
\end{defn}

An XOR game $G$ is completely specified by providing the game matrix $A$ and parity bit vector $\hat{s}$: $G \like (A,\hat{s})$.
For example, the GHZ game~\cite{greenberger1990bell} is defined by the clauses
(here we use the labels $\curly{\blue{x}, \orange{y}}$ for the questions instead
 of the typical $\curly{0,1}$):
\begin{equation}
\label{eqn: GHZ game matrix}
G_{GHZ} := \curly{
\bbm \blue{x} \\ \blue{x} \\ \blue{x} \\ +1 \ebm,
\bbm \orange{y} \\ \orange{y} \\ \blue{x} \\ -1 \ebm,
\bbm \orange{y} \\ \blue{x} \\ \orange{y} \\ -1 \ebm,
\bbm \blue{x} \\ \orange{y} \\ \orange{y} \\ -1 \ebm
}.
\end{equation}
We translate the GHZ queries into $A_{GHZ}$ and parity bits into $\hat{s}_{GHZ}$ by:
\begin{align}
\implies A_{GHZ} :=&
\begin{pmatrix}
\blue{1} & \gray{0} & \gray{0} & \blue{1} \\
\gray{0} & \orange{1} & \orange{1} & \gray{0} \\
\hline
\blue{1} & \gray{0} & \blue{1} & \gray{0} \\
\gray{0} & \orange{1} & \gray{0} & \orange{1} \\
\hline
\blue{1} & \blue{1} & \gray{0} & \gray{0} \\
\gray{0} & \gray{0} & \orange{1} & \orange{1}
\end{pmatrix}^\transp
\begin{matrix*}[l]
&\leftarrow (\text{Alice},\blue{x}) \\
&\leftarrow (\text{Alice},\orange{y}) \\
&\leftarrow (\text{Bob},\blue{x}) \\
&\leftarrow (\text{Bob},\orange{y}) \\
&\leftarrow (\text{Charlie},\blue{x}) \\
&\leftarrow (\text{Charlie},\orange{y})
\end{matrix*} \\
=& \left(\begin{array}{cc|cc|cc}
\blue{1} & \gray{0} & \blue{1} & \gray{0} & \blue{1} & \gray{0} \\
\gray{0} & \orange{1} & \gray{0} & \orange{1} & \blue{1} & \gray{0} \\
\gray{0} & \orange{1} & \blue{1} & \gray{0} & \gray{0} & \orange{1} \\
\blue{1} & \gray{0} & \gray{0} & \orange{1} & \gray{0} & \orange{1}
\end{array} \right)
\qand
\hat{s}_{GHZ} := \begin{pmatrix} 0 \\ 1 \\ 1\\ 1 \end{pmatrix}. 
\end{align}

Many of our results apply to two special classes of XOR games:
symmetric XOR games and random XOR games.
\begin{defn}
A \textbf{symmetric $k$-XOR} game is an XOR
game that additionally satisfies a clause symmetry property:
for every clause $c_i = (q_i, s_i)$ in the game,
the game must also contain all clauses $c_i' = (q_i', s_i)$
where $q_i'$ is a permutation of the questions in $q_i$ and
the parity bit is unchanged.
\end{defn}

\begin{defn}
A \textbf{random $k$-XOR} game on $m$ clauses and $n$ variables is an XOR game
with the $m$ clauses chosen independently at random from a uniform
distribution over $[n]^k\times\{-1,1\}$.
\label{def:random_game}
\end{defn}

\subsection{Strategies} \label{subsec: Introduction to strategies}
We next introduce both classical and \co{} strategies and state
claims regarding their value and constraints on
when these strategies play perfectly given an XOR game.
These claims are proved in Section~\ref{subsec:MERP PESc Duality}.
For any game, constructing a strategy and computing its value lower-bounds
the value of the game. In the \co{} case, this is generally intractable
and motivates the subsequent refutations picture.

\subsubsection{Classical Strategies}
For any game, the optimal classical strategy can be taken to be a deterministic assignment
of answers.  In the case of XOR games we will see that it is natural to view this
assignment as a vector in $\bbF_2^{kn}$.

\begin{defn}
\label{def: Classical strategy}
A \textbf{deterministic classical strategy} dictates that player $\alpha$ outputs
$\eta(\alpha,j) \in \curly{-1,1}$ when they receive question $j$ from the verifier.
Note that valid outputs must satisfy
\begin{equation}
\label{eqn: Classical square identity}
\eta^2(\alpha,j) = 1.
\end{equation}

To exploit the linear algebraic picture, it is useful to define a length-$kn$
\textbf{classical strategy vector} $\hat \eta \in \bbF_2^{kn}$ analogous to the parity bit vector.
It is defined by the relation
\be
\eta(\alpha,j) = (-1)^{\hat{\eta}_{n(\alpha-1) + j} } = \cos(\pi\hat\eta_{n(\alpha-1)+j}).
\ee
\end{defn}
Here the $\cos$ anticipates a generalization that we will see in the quantum case
when we construct MERP strategies.

\begin{claim}
\label{claim:Classical strategy value}
If the players play a game $G \sim (A,\hat{s})$ following strategy $\vcstrat$, the vector $\hat{o} = A
\vcstrat$ determines their output, i.e.~query $i$ has answer $(-1)^{\hat o_i}$.
The value of classical strategy $\vcstrat$ is
\begin{align}
\label{eqn: Classical strategy value cos}
v(G, \vcstrat) :=
\frac 1m \sum_{i=1}^m\frac{1 + (-1)^{\hat o_i - \hat s_i}}{2} =
\frac{1}{2} + \frac{1}{2m} \paren{\sum_{i=1}^m \cos(\pi \left[ \paren{A\vcstrat}_i - \hat{s}_i \right] )},
\end{align}
where again we have used an apparently unnecessary $\cos$, anticipating a quantum
generalization.  We also treat $\bbF_2$ and $\{0,1\}$ as equivalent here.
\end{claim}

These observations lead to a well known procedure using Gaussian elimination to find
 a classical value-1 strategy or determine that no such
 strategy exists.

\begin{defn} \label{defn: Classical constraint equation}
Define the \textbf{classical constraint equations} for game $G$ by
\begin{align}
\label{eq:cl_linalg}
A \hat{\eta} = \hat{s}
\end{align}
over $\mathbb{F}_2$.  Equivalently,
\be \prod_{\alpha=1}^{k} \eta(\alpha,q_i^{(\alpha)}) = s_i, \; \forall i\in [m].
\label{eq:cl_prod_eq}\ee
\end{defn}

\begin{claim} \label{claim:Classical value 1}
Every solution $\vcstrat$ to (\ref{eq:cl_linalg}) corresponds to a
strategy $\eta$ achieving
value 1 on game $G \sim (A,\hat{s})$, and vice versa.  In other words, 
a game $G$ has classical value $1$ iff (\ref{eq:cl_linalg}) has a
solution. 
\end{claim}

When $\omega(G) < 1$, on the other hand, there does not exist an efficient algorithm for
finding optimal classical strategies (assuming $\Ptime\neq\NP$)~\cite{haastad2001some}. 


\subsubsection{\CO{} Strategies}

\begin{defn}
Consider a $k$-XOR game with $n$ variables. For each $j \in [n]$, let the Positive-Operator Valued Measure (POVM)
$\{ O^{\alpha}_{1}(j), O^{\alpha}_{-1}(j) \}$ give the $\alpha$-th player's \textbf{\co{}
  strategy} upon receiving question $j$ from the verifier.  These POVMs act on some shared
state $\ket{\Psi}$, and different players' POVM elements commute due to the
no-communication requirement on the players. 
\end{defn}

Using the Naimark dilation theorem,  we can restrict our players' strategies to be
Projection-Valued Measures (PVMs).
 We make this restriction for the remainder of the paper. This allows us to define the following observables.
\begin{defn} \label{defn: Strategy Observables}
Given a strategy $\{ O^{\alpha}_{1}(j), O^{\alpha}_{-1}(j) \}$, define the \textbf{strategy observable}
$$O^{\alpha}(j) := O^{\alpha}_{1}(j) - O^{\alpha}_{-1}(j).$$
\end{defn}
Since $\{ O^{\alpha}_{1}(j), O^{\alpha}_{-1}(j) \}$ is a PVM, $O^{\alpha}(j)$ is a
Hermitian operator. 
Indeed \co{} strategies are equivalent to imposing the constraints for $\alpha \neq \beta$
\begsub{co-constraints}
[O^{\alpha}(j), O^{\beta}(j')] &= 0  & \text{(operators held by distinct players commute)} \\
\left(O^{\alpha}(j)\right)^2 &= I  & \text{(square identity, analogous to (\ref{eqn: Classical square identity}))} 
\endsub

The condition for commuting-operator strategies to achieve value 1 is the following
generalization of  (\ref{eq:cl_linalg}).

\begin{defn} \label{def:3XOR-sat-cond}
For a $k$-XOR game $G$, define the \textbf{commuting-operator constraint equations}:
\begin{align}
Q_i \ket{\Psi} := \Big(\prod_{\alpha} O^{\alpha}(q_{i}^{(\alpha)})\Big)\ket{\Psi} =
  s_i\ket{\Psi}, \, \, \, \all i \in [m] 
\label{eq:co-value-1}\end{align}
\end{defn}
These equations stipulate that applying the strategy observables for a given question
to the shared state $\ket{\Psi}$ produces the correct output for that question.
\begin{claim}
A game $G$ has commuting operator value 1 iff there exists some state and strategy
observables that satisfy~\eqref{eq:co-constraints} and~\eqref{eq:co-value-1}.
\end{claim}

While there is an efficient algorithm to solve  the classical constraint equations, no
such algorithm is known to exist for the \co{} constraint equations. This difficulty forces us to consider alternative techniques for characterizing the \co{} value of XOR games. 

\subsection{Refutations} \label{subsec: Intro Refutations}
In addition to lower bounding the value of a game by constructing strategies for it, we can also upper bound a game's value by showing no high-value strategy can exist.
In particular, we construct proofs that a game cannot have value 1, which we call refutations. Classically, refutations are well understood, and emerge naturally from the dual to the classical constraint equations.

\begin{defn} \label{defn: classical refutation}
Define a \textbf{classical refutation} $y \in \mathbb{F}_2^{m}$ as any vector satisfying the equations dual to (\ref{eq:cl_linalg}),
\begin{equation}\label{eq:cl_refut}
\begin{bmatrix} A^\transp \\ \hat{s}^\transp \end{bmatrix} y = \begin{bmatrix} 0 \\ 1 \end{bmatrix}
\end{equation}
where once again the algebra is over $\mathbb{F}_2$.
\end{defn}
\begin{fact} \label{fact:cl_duality}
Either a classical refutation $y$ exists satisfying \eqref{eq:cl_refut} or a classical
strategy $\hat{\eta}$ exists satisfying \eqref{eq:cl_linalg}.
\end{fact}
\noindent
The proof is standard but because dualities like Fact~\ref{fact:cl_duality} play a major role in our
paper, we review it here.
\begin{proof}
By the definition of $\im$ and $\ker$, we have $\im A \subseteq (\ker A^\transp)^\perp$.  The
rank-nullity theorem implies that $\dim\im A = \dim (\ker A^\transp)^\perp$, meaning that in fact
\be \im A = (\ker A^\transp)^\perp.\ee
Therefore
\be 
 \hat s  \not\in \im A 
\quad \Leftrightarrow \quad 
 \hat s  \not\in (\ker A^\transp)^\perp
\quad \Leftrightarrow \quad 
\exists y  \in \ker A^\transp, \;
\hat{s}^\transp y \neq 0.
\ee
\end{proof}

Another way to view $y$ as a refutation is by interpreting it as the indicator vector of a
subset of clauses. Recall from \eqref{eq:cl_prod_eq} that clause $i$ corresponds to the
equation $\prod_\alpha \eta(\alpha, q_i^{(\alpha)})=s_i$ over the variables $\eta(\cdot,\cdot)$. 
 If $y$ satisfies \eqref{eq:cl_refut} then multiplying
the equations corresponding to clauses with $y_i=1$ yields
\be 
\prod_{i : y_i=1}
\prod_{\alpha\in [k]} \eta(\alpha,q_i^{(\alpha)}) 
= 
\prod_{i : y_i=1} s_i \label{eq:prod_cl_clauses}\ee
From $A^\transp y=0$ and \eqref{eqn: Classical square identity} it follows that the LHS of
\eqref{eq:prod_cl_clauses} equals 1.
From $s^\transp y=1$ it follows that the RHS of
\eqref{eq:prod_cl_clauses} equals $-1$.
Thus the existence of $y$ satisfying \eqref{eq:cl_refut} means there is no $\eta$
satisfying \eqref{eq:cl_prod_eq}.




In this paper we consider the \co{} analogue of classical refutations.   We would like to
construct a dual to \eqref{eq:co-value-1}, meaning a characterization of certificates for 
the unsatisfiability of \eqref{eq:co-value-1}.  As there is no analogue to the linear
algebraic methods used in the classical case, we will instead attempt to generalize
\eqref{eq:prod_cl_clauses}.

Cleve and Mittal~\cite{cleve2014characterization} make use of a noncommutative
generalization of \eqref{eq:prod_cl_clauses}, which they call the substitution method, to
exhibit refutations of some Binary Constraint System games.
We will use a similar method for XOR games in which we multiply together constraints of
the form \eqref{eq:co-value-1} to obtain a contradiction.  Our contribution will be to
give a simple characterization of when such refutations exist in the case of symmetric
$k$-XOR games and in some cases, random asymmetric 3-XOR games.   Indeed, our
characterization will resemble the classical dual equations 
\eqref{eq:cl_refut} although the route by which we obtain it is quite different.

To explain this in more detail, we introduce some definitions.
\begin{defn} \label{def: Operator identity equivalence}
Let $Z_1$ and $Z_2$ be two operators formed from products of strategy observables. We say
$Z_1$ is equivalent to $Z_2$, written $Z_1 \sim Z_2$, if $Z_1 = Z_2$ is an identity for all strategy observables satisfying \eqref{eq:co-constraints}.
\end{defn}

Definitions~\ref{def:3XOR-sat-cond} and
\ref{def: Operator identity equivalence} then motivate the
definition of a (quantum) refutation, analogous to Definition \ref{defn: classical
  refutation}.  From now on,  a ``refutation'' will be a quantum refutation unless
otherwise specified.

\begin{defn}\label{def: Operator refutation}
Let $G$ be some $k$-XOR game with $m$ clauses. A \textbf{refutation for $G$} is defined to be a sequence $(i_1, i_2, \ldots , i_\ell) \in [m]^\ell$ satisfying
\begin{equation}
\begin{aligned}
Q_{i_1}Q_{i_2} \ldots Q_{i_\ell} \sim I
\end{aligned} \;\;\;\;\; \text{ and } \;\;\;\;\;
\begin{aligned}
s_{i_1}s_{i_2}\ldots s_{i_\ell} = -1.
\label{eq:refut-sign}
\end{aligned}
\end{equation}
\end{defn}
Refutations certify that $\omega^* < 1$, analogous to the way that classical
refutations certify that $\omega < 1$.  In
Theorem~\ref{thm:sat_iff_no_refutations}, we show that in
fact any game with $\omega^* < 1$ has a refutation.
The proof of this fact relies on a connection between refutations and
the ncSoS hierarchy analogous to a connection made by Grigoriev~\cite{Gri01}
between classical refutations and the
SoS hierarchy.

It is not obvious that one can find refutations more easily than one can find strategies.
However, we next establish a necessary condition for a game to admit a refutation,
and thus an easily-identified subclass of XOR games that certainly do not have a
refutation meaning they have $\omega^* = 1$.

\subsection{Games with no Parity-Permuted Refutations (no\PR{} Games)} \label{subsec: noPES games}
no\PR{} games are a subclass of entangled XOR games for which it is easy to show
no refutation can exist. To motivate their construction and prove some properties
about them, we must first redefine refutations from a combinatorial perspective.
A more complete treatment of these ideas is given in Section \ref{sec: Refutations}.

\begin{defn}[Combinatorial Construction of Refutations, informal]
\label{defn:combinatorial construction of refutations}
For a $k$-XOR game $G$, consider the combinatorial version of the query $q_i$\footnote{
We overload the notation $q_i$ here to indicate both the definitional and
combinatorial version of a query, with the relevant meaning clear from
context.
} to be a vector with $k$ \textbf{coordinates} (the player indices) with
\textbf{letter} $q_i^{(\alpha)}$ at coordinate $\alpha$. Define the set of
\textbf{words} contained in $G$ to be all vectors formed by concatenating
the queries of $G$ coordinate-wise (by player). The \textbf{sign} of a
word contained in $G$
\be W = q_{i_1} q_{i_2} \dots q_{i_\ell} \ee
is defined as
\be s_W := s_{i_1} s_{i_2} \dots s_{i_\ell}. \ee
We will refer to each coordinate of the word as a \textbf{wire}.
The identity $I$ under the concatenating 
action is the word that is blank on every wire.

Define an equivalence relation generated by all wire-by-wire permutations of the following base relations (in this setting the product of two vectors indicates their coordinate-wise concatenation). 
\begin{enumerate}
\item (Repeated elements cancel) : 
$ 
\begin{bmatrix}
j \\
\vdots
\end{bmatrix}
\begin{bmatrix}
j \\
\vdots
\end{bmatrix}
\sim 
\begin{bmatrix}
 \\
\vdots
\end{bmatrix} \all j \in [n]
$
\item (Elements on different wires commute) : 
$
\begin{bmatrix}
j \\
j' \\
\vdots
\end{bmatrix}
\sim
\begin{bmatrix}
j \\
\\
\vdots
\end{bmatrix}
\begin{bmatrix}
\\
j' \\
\vdots
\end{bmatrix}
\sim 
\begin{bmatrix}
\\
j' \\
\vdots
\end{bmatrix}
\begin{bmatrix}
j \\
\\
\vdots
\end{bmatrix} \all j,j' \in [n]
$
\end{enumerate}
A \textbf{refutation} is defined to be a sequence $(i_1 , i_2, i_3, ... i_\ell) \in [m]^\ell$ for which
\begin{equation}
\begin{aligned}
q_{i_1}q_{i_2}...q_{i_\ell} \sim I
\end{aligned}
\;\;\;\;\;\text{ and }\;\;\;\;\;
\begin{aligned}
s_{i_1} s_{i_2} ... s_{i_\ell} = -1.
\end{aligned}
\end{equation}
\end{defn}

We claim that this definition of a refutation is equivalent to the one given in Section \ref{subsec: Intro Refutations}. Intuitively, such a construction
explicitly manipulates the operator identities required by each clause of $G$
in a way that exploits the operator requirements of \eqref{eq:co-constraints} to produce a refutation as in Definition \ref{def: Operator refutation}.
We prove this fact in Section \ref{sec: Refutations}. We next motivate the no\PR{} class of games by making the following key observation.
\begin{obs} \label{observation: parity requirement for refutations}
All elements contained in queries at even positions in a refutation must cancel with queries at odd positions. 
\end{obs}

To exploit this observation, we find it useful to define a broader
equivalence relation $\ppsim$ that allows for a parity-preserving
permutation on each wire before canceling and commuting letters.

\begin{defn}[Informal]
\label{defn:psim}
We say $k$-XOR word $W_A$ is \textbf{parity-permuted equivalent}
to $W_B$---denoted $W_A \ppsim
  W_B$---if
$W_A \sim W_B'$ where some permutations of
the even positions and odd positions on each wire of $W_B$
can produce $W_B'$.
\end{defn}

Since this is just a broadening of the equivalence given in
Definition~\ref{defn:combinatorial construction of refutations},
$W_1 \sim W_2 \implies W_1 \ppsim W_2$. With this equivalence
relation in hand, we can state a necessary
condition for the existence of a refutation of a game $G$.


\begin{defn} \label{defn:PREF}
A game $G$ contains a \textbf{Parity-Permuted Refutation (PREF)} if
the game $G$ contains a word which is $\ppsim I$ with sign $-1$. 

The set of \textbf{\PR{} Games} are the set of XOR games that contain PREFs. The set of \textbf{no\PR{} Games} are the set of XOR games that do not.
\end{defn}

\begin{thm}[Necessary condition for refutation]
\label{thm:Necessary condition for refutation}
If a game $G$ admits a refutation, it contains a PREF.
\end{thm}
\begin{proof}[Proof (sketch)]
This follows essentially immediately from the observation that
$\sim$ implies $\ppsim$.
\end{proof}

\begin{cor} \label{cor:refutations and noPR}
Every no\PR{} game has \co{} value 1.
\end{cor}
\begin{proof}
This follows directly from \thmref{Necessary condition for refutation}
and the completeness of refutations (\thmref{sat_iff_no_refutations}).
\end{proof}

The significance of no\PR{} games is made clear by the two following theorems.
For both, a short proof sketch is given, while the full proofs are delegated
to Section \ref{sec: Refutations}.
\begin{thm}[Informal] \label{thm: can compute noPES}
There exists a poly-time algorithm that decides membership in the set of no\PR{} games. 
\end{thm}
\begin{proof}[Proof (sketch)] The key observation here is that a game
$G \sim (A,\hat{s})$ contains a PREF if and only if there is a solution to the
set of equations
\begin{align}
\label{eqn: PESc cond 1}
A^{\transp}z &= 0 \\
\label{eqn: PESc cond 2}
\hat{s}^\transp z &= 1 \pmod 2 
\end{align}
for some $z \in \mathbb{Z}^{m}$. If (\ref{eqn: PESc cond 1}) and (\ref{eqn: PESc cond 2}) can be satisfied,
the game $G$ contains a PREF built by interleaving the multisets
of clauses
\begsub{PREF-specification}
\cO &= \{q_i \text{ with multiplicity } \abs{z_i} \all i : z_i > 0 \} \\
\cE &= \{q_i \text{ with multiplicity } \abs{z_i} \all i : z_i < 0 \}
\endsub
such that their elements are placed in odd and even
positions, respectively.
The reverse direction requires a technical lemma relating the even and
odd clauses of a PREF.
Then standard techniques for solving linear Diophantine equations complete the proof. 
\end{proof}
The vector $z$ defined in the proof of Theorem \ref{thm: can compute noPES} is
sometimes referred to as a \PREF{} specification
due to (\ref{eq:PREF-specification}).\footnote{
Or a MERP refutation, for reasons described in Section \ref{subsec: Intro MERP-PES duality}}

\begin{thm}[Informal] \label{thm: nopes complete for symmetric games}
The no\PR{} characterization is complete for symmetric games. That is, every value 1 symmetric game is in the no\PR{} set.
\end{thm}

\begin{proof}[Proof (sketch)]
We use the structure of symmetric games to construct shuffle gadgets -- short words
that move letters from one wire to another when they are appended onto an existing
word. We then show shuffle gadgets are sufficient to construct a refutation given
a PREF contained in the game.
This shows that containing a PREF is both necessary and sufficient
for a symmetric game to have a refutation. Then a symmetric game
is either in the set of no\PR{} games or has value $<1$.
\end{proof}

Theorems \ref{thm: can compute noPES} and \ref{thm: nopes complete for symmetric games}
together show that the class of symmetric value 1 games has a poly-time deterministic algorithm, while
previously the question of whether such games took value 1 was not known to be decidable.
This progress is due to the no\PR{} characterization of games.  

Given that no\PR{} games form a large class of value 1 games, it is reasonable to try to construct a \co{} strategy to play them. We can ask if there exists a strategy dual to the PREF criteria, similar to what we have described in the classical and \co{} cases. In particular, we ask if a game $G$ not satisfying (\ref{eqn: PESc cond 1}) and (\ref{eqn: PESc cond 2}) guarantees existence of a solution to the constraint equations indicating some simple family of strategies can achieve value 1 for $G$.

Somewhat miraculously, the answer to this question turns out to be yes. We proceed by first defining this class of strategies, then showing that their constraint equations are dual to the PREF criteria for any game. 

\subsection{Maximal Entanglement, Relative Phase (MERP) Strategies}
\label{subsec:Intro MERP Strategies}
We introduce a family of ``Maximal Entanglement, Relative Phase'' (MERP) strategies: a useful subfamily of the set of tensor-product (and thus commuting-operator) strategies. MERP strategies are a
generalization of the GHZ strategy to arbitrary games. Crucially,
determining whether a MERP strategy achieves value 1 for a game, and
if so a construction for such a strategy, can be described in
time polynomial in $m$, $n$, and $k$.\footnote{
For symmetric games, $m \sim \exp{k}$, so in this case one can
decide MERP value 1 and describe a strategy in time polynomial in
$m$ and $n$.
}

Furthermore, the conditions for a MERP strategy to achieve value 1 are dual to the \PR{} condition for a game, meaning MERP achieves value 1 on any no\PR{} game.
In particular, this means MERP strategies achieve tensor-product value 1 on any symmetric XOR game with $\omega^* = 1$ (Theorem \ref{thm: nopes complete for symmetric games}) as well as on a family of non-symmetric games (APD games, Section \ref{subsec: APD Games})  with $\omega^* = 1$ and classical value $\omega \rightarrow \frac{1}{2}$.

We begin with the definition of a MERP strategy for a game $G$. 
\begin{defn}[MERP] \label{defn: MERP}
Given a $k$-XOR game $G$ with m clauses, a \textbf{MERP strategy} for $G$ is a tensor-product strategy in which:
\begin{enumerate}
\item The $k$ players share the maximally entangled state 
\begin{align}
\ket{\Psi} = \frac{1}{\sqrt{2}}\left[\ket{0}^{\otimes k} + \ket{1}^{\otimes k}\right]
\end{align}
with player $\alpha$ having access to the $\alpha$-th qubit of the state.
\item Upon receiving question $j$ from the verifier, player $\alpha$ rotates his qubit by an angle $\theta(\alpha,j)$ about the $Z$ axis, then measures his qubit in the $X$ basis and sends his observed outcome to the verifier.

Explicitly, we define the states 
\begin{align}
\ket{\theta(\alpha, j)_{\pm}} &:= \frac{1}{\sqrt{2}} 
\left[\ket{0} \pm e^{i\theta(\alpha, j)}\ket{1}\right]  
\end{align}
and pick strategy observables
\begin{align}
O^{\alpha}(j) := \ketbra{\theta(\alpha,j)_+}{\theta(\alpha,j)_+} - \ketbra{\theta(\alpha,j)_-}{\theta(\alpha,j)_-}.
\end{align}
\end{enumerate}
\end{defn}
There exists a useful parallel between MERP strategies and classical strategies, which we summarize below.
Almost identically to the classical value (\ref{eqn: Classical strategy value cos}),
\begin{claim} \label{claim:MERP value}
Let the length-$kn$ \textbf{MERP strategy vector} for a given MERP strategy be defined by
\begin{align}
\hat{\theta}_{(\alpha-1)n + j} := \frac{1}{\pi} \theta(\alpha, j). 
\end{align}
The value achieved by that MERP strategy on game $G$ is:
\begin{equation} \label{eqn: MERP like value}
v^{\MERP}(G,\hat{\theta}) := \frac{1}{2} + \frac{1}{2m} \paren{ \sum_{i=1}^{m} \cos(\pi \left[ (A \hat{\theta})_i - \hat{s}_i \right]) }.
\end{equation}
\end{claim}
\begin{proof}
Explicit calculation. Done in full in Section \ref{subsec:MERP strategy details}.
\end{proof}

Claim \ref{claim:MERP value} allows us to write down the constraint equations for MERP strategies to achieve $v^{\MERP} = 1$.

\begin{defn}
\label{def:MERP constraint equations}
Define the \textbf{MERP constraint equations} for game $G$ by
\begin{equation}
\label{eqn:MERP value 1 constraint linalg}
A \hat{\theta} = \hat{s} \pmod 2
\end{equation}
with $\hat{\theta} \in \mathbb{Q}^{kn}$.
\end{defn}
\noindent 
(We could have equivalently required $\hat\theta$ to be in $\bbR^{kn}$.  This is because
$A,\hat s$ have integer entries and so any real solution to \eqref{eqn:MERP value 1
  constraint linalg} will also be rational.)

\begin{claim} \label{claim:MERP Gaussian Elimination}
A MERP strategy achieves $v^{\MERP} = 1$ on a game $G$ iff its MERP constraint equations
have a solution. A solution $\hat\theta$ corresponds to the  MERP strategy in which player $\alpha$  uses $\theta(\alpha,j) = \pi \hat{\theta}_{(\alpha-1)n + j}$.
\end{claim}

Intuitively, MERP provides an explicit construction allowing players to return an \emph{arbitrary phase} on each input, rather than the classical $0$ or $\pi$. The MERP constraint equations then ensure that for each question the returned phases sum to $\pi \hat{s}_i$ up to multiples of $2 \pi$.
For any game, \claimref{MERP Gaussian Elimination} allows us to efficiently determine whether some MERP strategy achieves value 1 via Gaussian elimination over $\mathbb{Q}$.
We often refer to this optimal MERP strategy\footnote{Despite the language, we do not wish to suggest that there is a single optimal MERP strategy. Instead one should imagine some convention being used to specify a unique MERP strategy from the set of optimal ones.} for a game $G$ as simply \emph{the MERP strategy} for $G$.

\subsection{MERP - \PR{} Duality}
\label{subsec: Intro MERP-PES duality}
The set of games for which MERP achieves value 1 is exactly the set no\PR{}. As in the classical and \co{} cases, the MERP constraint equations (\ref{eqn:MERP value 1 constraint linalg}) are dual to the \PR{} conditions:

\begin{thm}
\label{thm:Diophantine equations unsolvable are varphi2}
For any game $G$, either there exists a \PREF{} specification,
or a MERP strategy with value 1.
\end{thm}
\begin{proof}
Technical proof in the style of a Theorem of Alternatives, analogous
to Fact~\ref{fact:cl_duality}. See Section \ref{subsec:MERP PESc Duality}.
\end{proof}

Because of Theorem \ref{thm:Diophantine equations unsolvable are varphi2} we also refer to a \PREF{} specification $z$ as a \emph{MERP refutation}.

Figure~\ref{fig:dualities} summarizes the extensions of the classical duality relations presented in this paper. The general quantum duality provides a complex but complete description of games with $\omega^* = 1$. The \PR{} conditions are efficient to compute, but are only \emph{necessary} conditions for constructing \co{} refutations, and thus the dual, MERP value 1, holds true for only a subset of all $\omega^* = 1$ games. We can make a stronger statement about symmetric games: PREFs are both necessary and sufficient for a symmetric game to have a refutation, so the duality ensures MERP achieves value 1 for all symmetric games with $\omega^* = 1$.

\begin{figure}[h]
\centering
\tikzstyle{mycbox} = [draw=OliveGreen, very thick, fill=OliveGreen!20, 
    rectangle, rounded corners, inner sep=10pt, inner ysep=15pt]
\tikzstyle{myqbox} = [draw=Plum, very thick, fill=Plum!20,
    rectangle, rounded corners, inner sep=10pt, inner ysep=15pt]
\tikzstyle{cfancytitle} =[fill=OliveGreen, text=white]
\tikzstyle{qfancytitle} =[fill=Plum, text=white]
\begin{tikzpicture}[ transform shape]
  \tikzstyle{every node} = [rectangle, draw = black, very thick, align
  = left]
  \node [mycbox] (clprimal) at (0, 4) { $\exists \hat{\eta}$ s.t.~$
  A\hat{\eta} = \hat{s}$ over $\mathbb{F}_2$};
  \node [cfancytitle,right=5pt] at (clprimal.north west) {$k$-XOR: primal};
  \node [mycbox] (cldual) at (0, 0) {$\nexists y \in \mathbb{F}_2^m$ \\
  s.t.~$A^\transp y = 0$ over $\mathbb{F}_2$\\
  and $\hat{s}^\transp y = 1$ over $\mathbb{F}_2$};
  \node [cfancytitle,right=5pt] at (cldual.north west) {$k$-XOR: dual};
  \node [myqbox] (qprimal) at (6, 4) { 
    $\exists$ entangled strategy  \\
    s.t. $ \all i :  Q_i \ket{\Psi} = s_i \ket{\Psi}$
  };
  \node [qfancytitle,right=5pt] at (qprimal.north west) {$k$-player
    XOR game};
  \node [myqbox] (qmerp) at (12, 4) { $\exists x, z$ s.t. $Ax = \hat{s} + 2z$,\\ with $x
    \in \mathbb{Q}^{kn}, z \in \mathbb{Z}^{m}$};
  \node [qfancytitle,right=5pt] at (qmerp.north west) {MERP};
  \node [myqbox] (npa) at (6, 0) { $\nexists (i_1, \dots, i_\ell)
    \in [m]$ \\
    s.t. $\prod_{j=1}^{\ell} Q_{i_j} \sim I$ \\ and $\prod_{j=1}^{\ell} s_{i_j}
    = -1 $};
  \node [qfancytitle,right=5pt] at (npa.north west) {Refutations};
  \node [myqbox] (qref) at (12, 0) {$\nexists z \in \mathbb{Z}^{m}$ \\
    s.t. $A^\transp z = 0$ over $\bbZ$\\
    and $\hat{s}^\transp z = 1 \pmod{2}$};
  \node [qfancytitle,right=5pt] at (qref.north west) {PREF};
  \node [draw=none] (cltitle) at (0, -2) {{\color{OliveGreen} \textsc{Classical
        games}}};
  \node [draw=none] (cltitle) at (9, -2) {{\color{Plum} \textsc{Entangled
        games}}}; 
  \draw [{Latex[length=10pt]}-{Latex[length=10pt]}, double distance=2pt,
  		line width=1pt, OliveGreen, double=OliveGreen!20]
        ($(clprimal.south)+(0.8,0)$) -- ($(cldual.north)+(0.8,0)$)
        node [midway, left, draw=none, black] 
        {Fact~\ref{fact:cl_duality}};
  \draw [{Latex[length=10pt]}-{Latex[length=10pt]}, double distance=2pt,
  		line width=1pt, Plum, double=Plum!20]
        ($(qprimal.south) + (0.8,0)$) -- ($(npa.north) + (0.8,0)$)
        node[midway, left, draw=none, black]
        {Thm~\ref{thm:sat_iff_no_refutations}};
  \draw [{Latex[length=10pt]}-, double distance=2pt,
  		line width=1pt, Plum, double=Plum!20]
  		($(npa.east) + (0,0.3)$) -- ($(qref.west) + (0,0.3)$)
        node[midway, above, draw=none, black]
        {Thm~\ref{thm:Necessary condition for refutation}};
  \draw [{Latex[length=10pt]}-, double distance=2pt, line width=1pt, red]
  		($(qref.west) - (0,0.3)$) -- ($(npa.east) - (0,0.3)$)
        node[midway, below, draw=none]
        {Thm~\ref{thm: nopes complete for symmetric games}};
  \draw [{Latex[length=10pt]}-{Latex[length=10pt]}, double distance=2pt,
  		line width=1pt, Plum, double=Plum!20]
       	(qref) -- (qmerp) node[midway, right, draw=none, black]
        {Thm~\ref{thm:Diophantine equations unsolvable are varphi2}};
  \draw[-{Latex[length=10pt]}, double distance=2pt,
  		line width=1pt, Plum, double=Plum!20]
        ($(qmerp.west) + (0,0.4)$) -- ($(qprimal.east) + (0,0.4)$)
        node[midway, above, draw=none, black]
        {Clm~\ref{claim:MERP Gaussian Elimination}};
  \draw[-{Latex[length=10pt]}, double distance=2pt, line width=1pt, red, dashed]
  		($(qprimal.east) - (0,0.4)$) -- ($(qmerp.west) - (0,0.4)$)
        node[midway,below,draw=none] {Thm~\ref{thm:merp}};
  \draw[-{Latex[length=10pt]}, double distance=2pt, line width=1pt, red, dashed]
  		(qref.north west) -- (qprimal.south east)
        node[midway,right, draw=none] {Thm~\ref{thm:exists_decidability_alg}};
  \draw [dashed] (2.9,6) -- (2.9,-2);
\end{tikzpicture}
\caption{We extend the well-understood duality relation for classical
XOR games (left) to a more complex set of dualities characterizing perfect
strategies for entangled XOR games
(right). The arrows indicate implications, with the red, unfilled arrows
holding for symmetric games only. The dashed red arrows follow from the
other arrows for symmetric games.}
\label{fig:dualities}
\end{figure}
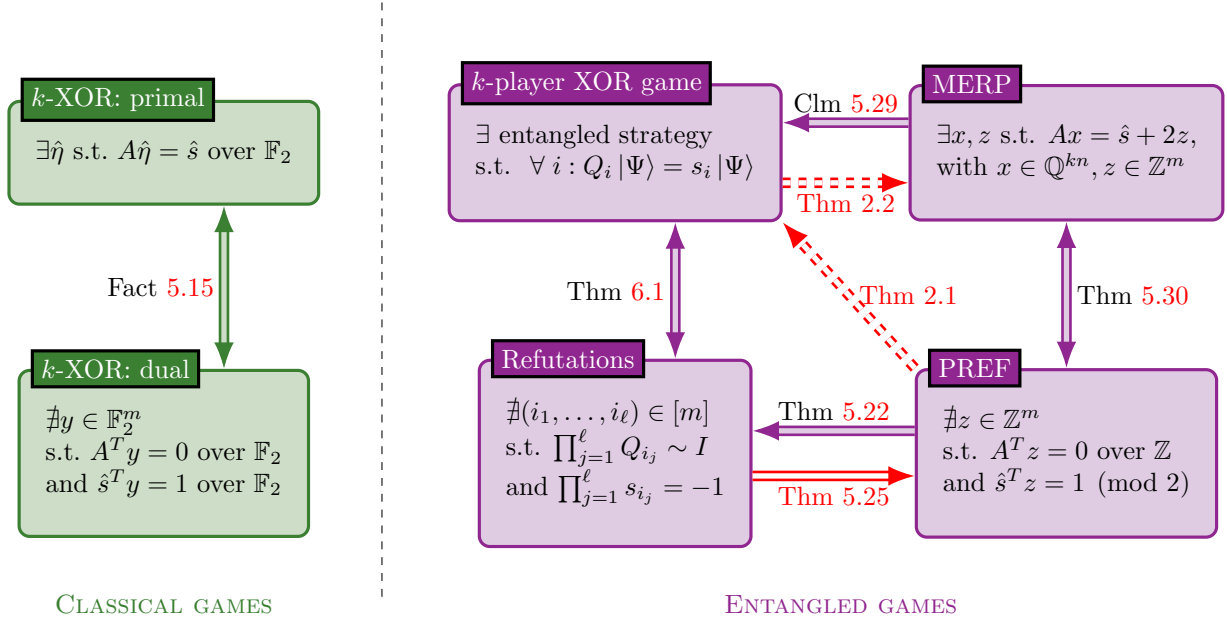

\subsection{Implications}
\label{subsec:Implications}
Finally we can use our main results to analyze some particular families of games and partially characterize the XOR game landscape.

In the $\omega^* = 1$ regime, we construct a family of games that generalize the GHZ game, termed the Asymptotically Perfect Difference (APD) family. Members are parameterized by scale $K$, with the $K$-th
member having $k = 2^K - 1$ players, and $K = 2$ reproducing GHZ.
The APD family is contained in the no\PR{} set ($\omega^* = 1$) and
has perfect difference in the asymptotic limit,
\be
\lim_{K \rightarrow \infty} 2(\omega^* - \omega) = 1.
\ee
This demonstrates that XOR games include a subset for which (at least asymptotically)
the best classical strategy is no better than random while a tensor-product
strategy (MERP) can play perfectly. Details of this construction are given in Section~\ref{subsec: APD Games}. 
We also give, in Section~\ref{subsec: 123 Game}, the construction for a (nonsymmetric) game for which $\omega^* = 1$ but which falls outside the no\PR{} set,
which shows the incompleteness of the \PR{} criteria.

To study the $\omega^* < 1$ regime, we consider the behavior of randomly generated XOR
games with a large number of clauses. We prove Theorem \ref{thm: k-XOR non-sym unsat} by
explicitly constructing a refutation for such games using insights developed in previous
sections. Interestingly, we also show such games have a minimal length refutation that
scales like $\Omega(n\log(n)/ \log(\log(n)))$, which implies that it takes the ncSoS
algorithm superexponential time to show that these games have $\omega^* < 1$ (Lemma
\ref{lem:no refutation of length l implies degree l pseudodistribution} and Theorem
\ref{thm:sos-LB}). These results can be seen as quantum analogues of
Grigoriev's~\cite{Gri01} integrality gap instances for classical XOR games.
Finally, we try to push the potentially superexponential runtime of ncSoS to its extremes. We demonstrate a family of symmetric games, called the Capped GHZ family, that provably have $\omega^* < 1$, but
have minimum refutation length exponential in the number of clauses (Section \ref{subsec: Capped GHZ Game}).
For games in this family the ncSoS algorithm requires time doubly exponential to prove that their \co{} value is $< 1$ while the no\PR{} criterion can be used to conclude this fact in polynomial time.  


\section{Refutations}
\label{sec: Refutations}




Refutations are a powerful tool for differentiating between XOR games
with perfect \co{} strategies ($\omega^* = 1$) and those with
$\omega^*$ bounded away from 1. In Section~\ref{subsubsec: gap},
we prove 
\thmref{sat_iff_no_refutations} and \thmref{refutation-gap-LB}, relating
refutations to the \co{} value of XOR games:

\begin{thm} \label{thm:sat_iff_no_refutations}
  An XOR game $G$ has \co{} value $\omega^*(G) = 1$
  if and only if it admits no refutations.
\end{thm}

\begin{thm} \label{thm:refutation-gap-LB}
Let $G$ be an XOR game consisting of $m$ queries, with $G$
yielding a length-$\ell$ refutation. The \co{} value of the game is
bounded above by
\begin{align}
\omega^*(G) \leq 1 - \frac{\pi^2}{4 m \ell^2}.
\end{align}
\end{thm}

Informally, \thmref{sat_iff_no_refutations} gives completeness and 
soundness of refutations when used as a proof system for checking if a 
game has $\omega^* < 1$. \thmref{refutation-gap-LB} improves the soundness.

We previously introduced the notion of the combinatorial view
of refutations (Definition~\ref{defn:combinatorial construction of refutations}) and
containing a PREF
as a necessary condition for a game to have a refutation 
(\corref{refutations and noPR}).
Section~\ref{subsec:tools for constructing refutations} presents the
combinatorial view in more detail, and proves
that a \PREF{} specification and existence of a particular set 
of ``shift gadgets''
is a \emph{sufficient} condition for a refutation to exist. Finally,
Section~\ref{section: decidability algorithm} demonstrates that for
symmetric XOR games, all desired ``shift gadgets'' are automatically
available, meaning that a refutation exists if and only if a
\PREF{} specification exists, thus providing an efficient technique
to decide whether any symmetric XOR game has perfect \co{} value.

\subsection{Upper Bound on Value} \label{subsubsec: gap}
We begin by proving \thmref{sat_iff_no_refutations}.
The main tool we use is the non-commuting Sum of
Squares (ncSoS) hierarchy, also known as the NPA
hierarchy~\cite{NPA08,DLTW08}. Given a game $G$, each level in the ncSoS hierarchy is a semidefinite
program depending on $G$ whose solution gives an upper bound on the value
$\omega^*(G)$; higher levels correspond to larger semidefinite programs and
tighter upper bounds. While we refer the reader to the references cited above
for a full description, we include here a definition of the key object used in
constructing the hierarchy: the \emph{pseudoexpectation} operator.
\begin{defn}
  Given an XOR game $G$, a \textbf{degree-$d$ pseudoexpectation operator} or \textbf{pseudodistribution} is a linear
  function $\psE{\cdot}$ that maps formal polynomials of degree at most $d$
  over the strategy observables
  $O^{\alpha}(j)$ to complex numbers. A pseudoexpectation $\psE{\cdot}$ is valid
  if
  \begin{itemize}
    \item for all polynomials $p$ of degree at most $d/2$, $\psE{p^\dag p} \geq
      0$,
    \item for all polynomials $p_1,p_2$ with $\deg(p_1p_2) \leq d - 2$ and indices
      $\alpha \in [k]$ and $j \in [n]$,
      \be \psE{p_1 \curly{(O^{\alpha}(j))^2 - I}p_2} 
      = 0. \ee
    \item for all polynomials $p_1,p_2$ with $\deg(p_1p_2) \leq d - 2$ and indices
      $\alpha \neq \alpha' \in [k]$ and $j, j' \in [n]$,
      \be \psE{p_1 \curly{O^{\alpha}(j) O^{\alpha'}(j') -
      O^{\alpha'}(j') O^{\alpha}(j) }p_2} 
= 0. \ee
    \end{itemize}
    Intuitively speaking, these requirements state that any algebraic
    manipulations allowed by \eqref{eq:co-constraints} are also
    allowed under the pseudoexpectation, as long as they never result in a
    polynomial of degree greater than $d$. We further say that a
    pseudoexpectation satisfies a clause $c_i = (q_i,s_i)$ if for all
    polynomials $p_1,p_2$ with degrees summing to $\leq d - k$,
    $\psE{p_1(Q_i - s_i I)p_2} = 0$. 
\end{defn}
The full ncSoS algorithm involves optimizing over all valid pseudoexpectation
operators that satisfy clauses in the game; it can be shown that this
optimization reduces to a semidefinite program in matrices whose dimension is
the number of monomials of degree at most $d/2$ in the observables
$O^{\alpha}(j)$. In the special case of determining whether the game value is
$1$, it reduces to checking for the existence of such a pseudoexpectation operator.

In \cite{Gri01}, Grigoriev showed a connection between refutations of classical games and pseudodistributions which appear to satisfy all clauses of a classical XOR game. In our analysis, we will adapt some of these arguments to the
quantum setting. In particular, Lemma \ref{lem:no refutation of length l implies degree l pseudodistribution} gives a quantum analogue of Grigoriev's central insight that, in the special case of deciding whether the game value is $1$, the sum-of-squares hierarchy reduces to checking for the existence of a refutation. 

In addition to being key to the proof of Theorem \ref{thm:sat_iff_no_refutations}, Lemma \ref{lem:no refutation of length l implies degree l pseudodistribution} also gives a bound on the time it takes the ncSoS algorithm to show a XOR game has value $<1$ in terms of the minimum length refutation admitted by the game.

\begin{lem} \label{lem:no refutation of length l implies degree l pseudodistribution}
For any $k$-XOR game $G$ with no refutation of length $\leq 2 \ell$ there exists a degree-$k\ell$
pseudodistribution whose pseudoexpectation satisfies every clause in $G$. Consequently,
it takes time at least  $\Omega((nk)^{k\ell})$ for the ncSoS algorithm to prove
$\omega^*(G) \neq 1$.    
\end{lem}
\begin{proof}
To construct this pseudodistribution, we follow a procedure of
Grigoriev~\cite{Gri01}. For each clause $c_i = (q_i,s_i)$, define
\begin{equation}
\psE{Q_i} = \psE{\prod_{\alpha} O^{\alpha}(q_i^{(\alpha)})} := s_{i},
\end{equation}
and for any \emph{word}\footnote{
In this context, we borrow this terminology from the
combinatorial picture to indicate any product of strategy observables.
} $w$ which can be obtained as a product of $N \leq \ell$ queries,
\begin{equation}
w := \prod_{x = 1}^{N} Q_{i_x},
\label{eq:w-prod-Q}\end{equation}
define the
pseudoexpectation of $w$ to be the product of the parity bits 
$s_{i_x}$ associated with each query $Q_{i_x}$ in the operator
construction:
\begin{equation}
\psE{w} := \prod_{x=1}^{N} s_{i_x}.
\label{eq:psE-w-prod}\end{equation}
We need to argue that this prescription is well-defined, i.e.~that \eqref{eq:w-prod-Q} and
\eqref{eq:psE-w-prod} never assign two different values to the same $\psE{w}$.  Suppose to the
contrary that $w = \prod_{x\in [M]} Q_{i_x} = \prod_{y\in [N]} Q_{j_y}$ with $M,N \leq
\ell$ but that $\prod_{x\in [M]} s_{i_x} \neq \prod_{y\in [N]} s_{j_y}$.  Since
\eqref{eq:psE-w-prod} can only take on the values $\pm 1$ we have $\prod_{x\in [M]}
s_{i_x} \cdot \prod_{y\in [N]} s_{j_y} = -1$.  Also each $Q_i$ is Hermitian, so
\be 1 = ww^\dag = Q_{i_1}\cdots Q_{i_N} Q_{j_M} \cdots Q_{j_1}.\ee
This constructs a refutation of length $M+N \leq 2\ell$, contradicting our hypothesis that
no such refutation exists.  We conclude that $\psE{w}$ is well-defined for the choices of
$w$ resulting from \eqref{eq:w-prod-Q}.

For all other words (i.e. those 
that cannot be obtained as products of queries or have length $> \ell$),
set their 
pseudoexpectation to $0$. Finally, extend the definition by linearity
to sums and scalar multiples of operator products.


Moreover, $\psE{\cdot}$ induces an
equivalence relation on words: we say that words $w_a \psesim w_b$
if $\psE{w_a^\dagger w_b} \neq 0$.  This relation therefore
partitions the set of words into equivalence classes $C_1, C_2, \dots$.
We pick a representative element $w_i$ for each class $C_i$.
A key feature of the equivalence relation is that for $w_a, w_b \in C_i$,
\be
w_a\psesim w_b \quad \implies \quad \psE{w_a^\dag w_b} = \psE{w_a^\dag w_i}\psE{w_i^\dag w_b}.
\label{eq:Grig-equiv}
\ee
  
To show that $\psE{\cdot}$ is a pseudodistribution, it
suffices to show that for any polynomial $p$ of degree at most $k\ell/2$ 
in the operators $O^{\alpha}(j)$, $\psE{p^\dagger p} \geq 0$.
Group the monomials in $p$ 
according to the equivalence classes, so that $p = p_1 + p_2 + \dots$ where each $p_i$ 
is a sum of terms from equivalence class $C_i$. It follows that
\begin{equation}
\psE{p^\dagger p} = \sum_{i} \sum_{j} \psE{p_i^\dagger p_j} = \sum_{i}
  \psE{p_i^\dagger p_i}.
\end{equation}

So we have reduced the problem to showing that $\psE{q^\dagger q}
\geq 0$ for any polynomial $q$, all of whose terms belong to the
same equivalence class. Write $q$ as a linear combination of words in
equivalence class $C_i$,
\be
q = \alpha_1  w_1 +
\dots + \alpha_s w_s.
\ee
Then
\begin{align}
  \psE{q^\dagger q} &= \psE{ \sum_{a,b =1}^{s} \alpha_a^* \alpha_b
                      w_a^\dagger w_b } \\
                    &= \sum_{a,b=1}^{s} \alpha_a^* \alpha_b
                      \psE{w_a^\dagger w_b} \\
                    &\stackrel{\eqref{eq:Grig-equiv}}{=} \sum_{a,b = 1}^{s} \alpha_a^* \alpha_b
                      \paren{\psE{w_i^\dag w_a}}^\dagger \psE{w_i^\dag w_b}  \\
                    &= \Big| \sum_a \alpha_a \psE{w_i^\dag w_a} \Big|^2 \\
                    &\geq 0.
\end{align}

The existence of this pseudodistribution implies that the ncSoS algorithm would need to run to level at least $k \ell$ in the ncSoS hierarchy to show $G$ has \co{} value $<1$. This can be converted to a lower bound on the runtime by standard results in semidefinite programing. 
\end{proof}

Finally, we state and prove the duality between refutations and $\omega^* = 1$.
\begin{repthm}{thm:sat_iff_no_refutations}
  An XOR game $G$ has \co{} value $\omega^*(G) = 1$
  if and only if it admits no refutations.
\end{repthm}
\begin{proof}
In one direction, \defref{3XOR-sat-cond} immediately implies that if the game has commuting value $1$, then there are no refutations.

  In the other direction, suppose there are no refutations. Then, by \lemref{no refutation of length l implies degree l pseudodistribution} and taking $\ell \rightarrow \infty$ we see there exists a
  pseudodistribution under which every clause is satisfied, and this
  pseudodistribution satisfies the constraints of all levels of the
  ncSoS hierarchy~\cite{NPA08}. Since it is known that the ncSoS hierarchy converges to
  the commuting value of the game, it follows that this value is $1$.
\end{proof}

Classical refutations prove that a constraint satisfaction problem is not feasible, and so if there are $m$ constraints they trivially yield an upper bound of $1-1/m$.  In the \co{} case, even this statement is not obvious. In particular, one could worry that a game with a quantum refutation still admits a sequence of \co{} strategies with limiting value 1. 

However, here we prove \thmref{refutation-gap-LB}, showing that even in the \co{} case, refutations yield explicit upper bounds on $\omega^*(G)$ that are strictly less than 1. An argument similar to the one presented here was known previously, and used to derive a comparable result in Section 5 of \cite{cleve2014characterization}.


\begin{repthm}{thm:refutation-gap-LB}
Let $G$ be an XOR game consisting of $m$ queries, with $G$
yielding a length-$\ell$ refutation. The \co{} value of the game is
bounded above by
\begin{align}
\omega^*(G) \leq 1 - \frac{\pi^2}{4 m \ell^2}.
\end{align}
\end{repthm}

\begin{proof}
Recall from \defref{3XOR-sat-cond} that the Hermitian operator $Q_i$ is defined for some XOR game $G$, and represents the collective measurements made by the players upon receiving query $q_i$. It has eigenvalues $\pm1$, which correspond to the value of the XOR'd bit received by the verifier. Define $\tilde{Q_i} := s_i Q_i$, so the 1 eigenspace of $\tilde{Q_i}$ corresponds to measurement outcomes on which the players win the game given query $q_i$, and the $-1$ eigenspace corresponds to measurement outcomes on which the players lose the game. Let $(i_1, i_2, \ldots  i_\ell)$ be the assumed refutation for $G$. Letting $\ket{\Psi}$ be the state shared by the players, we have
\begin{align}
\tilde{Q}_{i_1} \tilde{Q}_{i_2} \ldots  \tilde{Q}_{i_\ell} \ket{\Psi} &= \left(s_{i_1}s_{i_2}\ldots s_{i_\ell}\right) Q_{i_1}Q_{i_2}\ldots Q_{i_\ell}\ket{\Psi} = -1 \ket{\Psi}.
\label{eq:refutation-on-psi}\end{align}
On the other hand, 
if we let $P_i = \frac{I- \tilde Q_i}{2}$ be the projector on to the $- 1$ eigenspace of $\tilde Q_i$ then the losing probability is
\be \delta := \frac{1}{m}\sum_{i=1}^m \Tr \Big[P_i \ket{\Psi} \bra{\Psi} \Big]
\label{eq:losing-prob} \ee

We now follow an argument similar to the union bound proof of \cite{Gao15}.  Let $\angle(\ket\alpha,\ket\beta) = \arccos|\braket{\alpha}{\beta}|$ and observe that it satisfies the triangle inequality, i.e. $\angle(\ket\alpha,\ket\gamma)\leq \angle(\ket\alpha,\ket\beta) + \angle(\ket\beta,\ket\gamma)$. Then
\begin{align}
\pi &\stackrel{\text{\eqref{eq:refutation-on-psi}}}\leq
 \sum_{x=1}^\ell \angle \Big( \ket\Psi,Q_{i_x}\ket\Psi \Big)
 & \text{Note that $Q_i$ is unitary.} \\
& = \sum_{x=1}^\ell\arccos \paren{1 - 2\Tr \Big[ P_{i_x} \ket{\Psi}\bra{\Psi} \Big]} \\
& \leq \sum_{x=1}^\ell 2\sqrt{\Tr \Big[ P_{i_x} \ket{\Psi}\bra{\Psi} \Big]} \\
& \stackrel{\text{\eqref{eq:losing-prob}}}{\leq}
2 \sum_{x=1}^{\ell} \sqrt{m\delta} \\
&=2\ell\sqrt{m\delta}.
\end{align}

\end{proof}


\subsection{Tools for Constructing Refutations}
\label{subsec:tools for constructing refutations}
Having demonstrated the utility of refutations, we return to
the combinatorial picture of refutations and prove 
necessary and sufficient conditions for an XOR game to contain a refutation.



\subsubsection{Combinatorics}
We now formally reintroduce $k$-XOR games from a combinatorial standpoint. Several
definitions mirror those in Section \ref{subsec: Intro Refutations} but are
presented here in a slightly different form to enable discussion of combinatorial
proofs. 

\begin{defn}
A $k$-XOR \textbf{game} on $m$ clauses with $n$ questions is defined to be a set of $m$ $k$-tuples, consisting of elements of $[n]$, with $m$ associated parity bits. An individual $k$-tuple is called a query, and is denoted by 
\begin{align}
q_i = 
\begin{bmatrix} q_i^{(1)} \\ 
q_i^{(2)} \\
\vdots \\
q_i^{(k)} \end{bmatrix}.
\end{align}
\end{defn}
\begin{defn} \label{definition: words and games}
A \textbf{word} $W$ on alphabet $[n]$ is a $k$-tuple of the form
\begin{align}
W = \begin{bmatrix} w_{11} \; w_{12} \; \ldots  \; w_{1\ell_1} \\ 
w_{21} \; w_{22} \; \ldots  \; w_{2\ell_2} \\
\vdots \\
w_{k1}w_{k2}\ldots w_{k\ell_k} \end{bmatrix}
\end{align}
with all $w_{ij} \in [n]$. Each row of $W$ is referred to as a \textbf{wire} of the word, and the $\alpha$-th row is sometimes denoted by $W^{(\alpha)}$. When all wires have length $\ell$, (so $\ell_1 = \ell_2 = \ldots  \ell_k = \ell$) we say $W$ has length $\ell$. 

The product of two words is defined to be their coordinate-wise concatenation. The notation $q_{i_1}q_{i_2}\ldots q_{i_\ell}$ then refers to the length $\ell$ word given by
\begin{align}
q_{i_1}q_{i_2}\ldots q_{i_\ell} =
\begin{bmatrix}
q_{i_1}^{(1)} \; q_{i_2}^{(1)} \; \ldots  \; q_{i_\ell}^{(1)} \\
q_{i_1}^{(2)} \; q_{i_2}^{(2)} \; \ldots  \; q_{i_\ell}^{(2)} \\
\vdots \\
q_{i_1}^{(k)} \; q_{i_2}^{(k)} \; \ldots  \; q_{i_\ell}^{(k)} \\
\end{bmatrix}.
\end{align}
Finally, define the \textbf{identity word} $I$ to be the empty $k$-tuple,
which satisfies 
\begin{align}
IW = IW = W
\end{align}
for any word $W$.
\end{defn}

\begin{defn}
A game $G$ \textbf{contains a word $W$ with sign
$s_W \in \curly{\pm 1}$} if
\begin{align}
W &= q_{i_1}q_{i_2} \ldots  q_{i_\ell} \text{ and} \\
s_W &= s_{i_1}s_{i_2} \ldots s_{i_\ell}
\end{align}
for some $(i_1, i_2, \ldots  i_\ell) \in [m]^{\ell}$.
\end{defn}

\begin{defn}
Relations are used to express equivalence between words. There are two basic types (shown here for 3-XOR, and defined analogously for $k$-XOR). 
\begin{enumerate}
\item (Commute Relations): \begin{align}
\begin{bmatrix}
j \\
j' \\ 
j'' 
\end{bmatrix} \sim \begin{bmatrix}
  \\
  \\ 
j'' 
\end{bmatrix} \begin{bmatrix}
  \\
j' \\ 
  \\ 
\end{bmatrix} \begin{bmatrix}
j \\
  \\
  \\
\end{bmatrix} \sim \begin{bmatrix}
  \\
j' \\ 
  \\ 
\end{bmatrix} \begin{bmatrix}
j \\
  \\
  \\
\end{bmatrix} \begin{bmatrix}
  \\
  \\ 
j'' 
\end{bmatrix} \sim \begin{bmatrix}
j \\
  \\
  \\
\end{bmatrix} \begin{bmatrix}
  \\
  \\ 
j'' 
\end{bmatrix} \begin{bmatrix}
  \\
j' \\ 
  \\ 
\end{bmatrix} 
\qquad \all j,j',j'' \in [n]
\end{align}
\item (Cancellation Relations):
\begin{align}
\begin{bmatrix}
j^2 \\
 \\
 \\
\end{bmatrix} \sim 
\begin{bmatrix}
    \\
j^2 \\
    \\
\end{bmatrix} \sim \begin{bmatrix}
    \\
    \\
j^2 \\
\end{bmatrix} \sim I \qquad \all j \in [n]
\end{align}
\end{enumerate}
The relationship property is associative (as suggested by the notation), so more complicated equivalences can be constructed by concatenating the ones above.
\end{defn} 

\begin{defn}
Given a $k$-XOR game $G$, a length $\ell$ \textbf{refutation} for that game is a length $\ell$ word $W$ contained in $G$ with sign $-1$ and 
\begin{align}
W \sim I.
\end{align}
\end{defn}

\subsubsection{PREFs and Shuffle Gadgets}
\label{sec:shuffle}

The key difference between entangled and classical strategies is that in the
entangled case, the strategy observables do not all commute with each other.
In other words, strings of queries can be acted on nontrivially by 
permutations. In this section we consider equivalence under a
restricted class of parity-preserving permutations, and use the
fact that \emph{at least one element} of a
class equivalent to some refutation must be contained in a
game for it to admit a refutation, giving a tractable
necessary condition for a refutation to exist.
We then define gadgets that perform these 
permutations while preserving the associated parity bits. The result
will be a useful set of sufficient conditions for a refutation to exist. 

We recall the formal definitions related to these equivalence classes.

\begin{repdefn}{defn:psim}
  Given two 1-XOR words $W_1, W_2$, we say that $W_1$ is
  \textbf{parity-permuted equivalent} to $W_2$---denoted $W_1 \ppsim
  W_2$---if there exists a permutation $\pi$ mapping even indices to even
  indices and odd indices to odd indices such that $W_1 \sim
  \pi(W_2)$.
  
  For $k$-XOR words $W_A, W_B$, we say $W_A \ppsim W_B$ if $W_A^{(\alpha)} \ppsim W_B^{(\alpha)}$ for all $\alpha \in [k]$. 
\end{repdefn}
From the definition, we see that $\ppsim$ is necessary for $\sim$, i.e.
\begin{align}
W_1 \sim W_2 \implies W_1 \ppsim W_2.
\end{align}
We can then conclude that a game $G$ contains a refutation only if it contains a word $W \ppsim I$ with sign $-1$. To make this necessary condition more useful to us, we will move from an operational definition of the $\ppsim$ relation to a structural one. This is done by talking about the even and odd subsets of a given word. The relevant definitions are given below. 

  


\begin{defn}\label{defn:multiplicity equivalent}
Two multisets of queries $\cT_1$ and $\cT_2$ are said to be \textbf{multiplicity equivalent} if all player-question combinations occur with the same multiplicity in both sets. That is, $\cT_1$ and $\cT_2$ are multiplicity equivalent iff
\begin{align}
\abs{\left\{q \in \cT_1 : q^{(\alpha)} = j\right\}} = \abs{\left\{q' \in \cT_2 : q'^{(\alpha)} = j\right\}} \all \alpha, j. 
\end{align}
\end{defn}

\begin{defn}
Given a word contained in a game $G$
\begin{align}
W = q_{i_1}q_{i_2} ... q_{i_\ell}
\end{align}
define its even and odd multisets $\cE$ and $\cO$ in the natural way, so\footnote{
Here and beyond we use the multiset operation $\biguplus$ to indicate
union with addition of multiplicities. When applied to single elements
we mean to treat each element as a single-element multiset.
}
\begin{align}
\cE := \biguplus_{x \text{ even}}q_{i_x} \qquad \text{and} \qquad \cO := \biguplus_{x \text{ odd}}q_{i_x}.
\end{align}
\end{defn}
The key feature of the multiplicity equivalence condition is that a word contained in a game $G$ is $\ppsim I$ iff its even and odd multisets are multiplicity equivalent. A slightly more general form of this statement is proved below.  
\begin{lem} \label{lem:ppsim and even odd multisets}
Given two words $W_1$ and $W_2$ contained in $G$, the following are equivalent:
\begin{enumerate}
\item $W_1 \ppsim W_2.$
\item The even and odd multisets of the word $W_1W_2^{-1}$ are multiplicity equivalent.
\end{enumerate}
\end{lem}
\begin{proof}
This proof is easiest if we generalize from the concept of even and odd multisets of clauses to even and odd multisets of variable-player combinations. In particular, given a word $W$ (not necessarily contained in a game $G$), its even and odd variable multisets are defined by
\begin{align}
\cE'(W) := \biguplus_{i \text{ even}, \alpha} \left(w_{i,\alpha}, \alpha \right) \\
\cO'(W) := \biguplus_{i \text{ odd}, \alpha} \left(w_{i,\alpha}, \alpha \right) 
\end{align}
where the tuple notion tracks the fact that variables given to different players are treated as distinct. To prove Lemma \ref{lem:ppsim and even odd multisets}, we must now claim some basic facts about $\cE'$ and $\cO'$. 
\begin{enumerate}[(A)]
\item \label{item: Multiplicity equivalent = identical generalized multisets} For a word $W$ contained in $G$, the even and odd multisets of $W$ are multiplicity equivalent iff 
\begin{align}
\cE'(W) = \cO'(W).
\end{align}
\item \label{item: Parity preserving permutations preserve multisets} Applying a parity preserving permutation to a word $W$ does not change $\cE'(W)$ or $\cO'(W)$. 
\item \label{item: cancellations preserve multisets} For any two words $W_1 \sim W_2$, we have 
\begin{align}
\cE'(W_1) \uplus \cO'(W_2) = \cE'(W_2) \uplus \cO'(W_1).
\end{align}
\end{enumerate}
Claims \ref{item: Multiplicity equivalent = identical generalized multisets} and \ref{item: Parity preserving permutations preserve multisets} come directly from the definition of $\cE'$ and $\cO'$. To prove claim \ref{item: cancellations preserve multisets} we consider two words $W_1 \sim W_2$. If we never used a cancellation relation, we would immediately have
\begin{align}
\cE'(W_1) = \cE'(W_2) \text{ and } \cO'(W_1) = \cO'(W_2).
\end{align}
Now a cancellation on a word always occurs between an element at an even position and one at an odd one, that is, it removes elements equally from $\cE'$ and $\cO'$. Letting $\cC_1$ be the multiset of elements removed from $\cE'(W_1)$ (and equivalently $\cO'(W_1)$) by cancellation, with $\cC_2$ defined similarly for $W_2$, we find 
\begin{align}
\left(\cE'(W_1)\backslash \cC_1 \right) \uplus \left(\cO'(W_2)\backslash \cC_2 \right) &= \left(\cE'(W_2)\backslash \cC_2 \right) \uplus \left(\cO'(W_1)\backslash \cC_1 \right) \\
\Leftrightarrow \cE'(W_1) \uplus \cO'(W_2) &= \cE'(W_2) \uplus \cO'(W_1).
\end{align}
Now, to prove Lemma \ref{lem:ppsim and even odd multisets} we note
\begin{align}
W_1 &\ppsim W_2 \\
\Leftrightarrow \exists \text{ parity preserving }\pi : \pi(W_1) &\sim W_2 & (\text{definition})\\
\Leftrightarrow \cE'(\pi(W_1)) \uplus \cO'(W_2) &= \cE'(W_2) \uplus \cO'(\pi(W_1)) & \ref{item: cancellations preserve multisets} \\
\Leftrightarrow \cE'(W_1) \uplus \cO'(W_2) &= \cE'(W_2) \uplus \cO'(W_1) & \ref{item: Parity preserving permutations preserve multisets}\\
\Leftrightarrow \cE'(W_1 {W_2^{-1}}) &= \cO'(W_1 {W_2^{-1}}) & (\text{reordering word}) \label{eqn: subtle step in multiplicity proof}\\
\Leftrightarrow \text{ The even and odd subsets of } &W_1W_2^{-1}\text{ are multiplicity equivalent.}  & \ref{item: Multiplicity equivalent = identical generalized multisets}
\end{align}
(\ref{eqn: subtle step in multiplicity proof}) is a somewhat subtle step, but follows formally (for example) from a proof by cases considering even and odd length words $W_1$ and $W_2$ and noting that the length of $W_1$ and $W_2$ must be equivalent mod 2. 
\end{proof}
\begin{repdefn}{defn:PREF}
A game $G$ contains a \textbf{Parity-Permuted Refutation (PREF)} if the queries of the game can be combined to form two multiplicity equivalent multisets for which the parity bits corresponding to the queries multiply to $-1$. Equivalently (Lemma~\ref{lem:ppsim and even odd multisets}),
the game $G$ contains a word which is $\ppsim I$ with sign $-1$. 

The set of \textbf{\PR{} Games} are the set of XOR games that contain PREFs. The set of \textbf{no\PR{} Games} are the set of XOR games that do not. 
\end{repdefn}

We can finally restate and prove our necessary condition formally:
\begin{repthm}{thm:Necessary condition for refutation}
[Necessary condition for refutation]
If a game $G$ admits a refutation, it contains a PREF.
\end{repthm}
\begin{proof}
By definition, a refutation $R$ admitted by game $G$ must
be $\sim I$ and therefore $R \ppsim I$. $R$ must also have
sign $-1$. By Definition~\ref{defn:PREF}, game $G$ then
contains a PREF.
\end{proof}

Phrasing this necessary condition in terms of
even and odd multiplicity equivalent multisets then provides
an efficient means of computing whether or not a game
satisfies this PREF criterion (Section~\ref{section: decidability algorithm}).

\begin{center}
\rule{0.2\textwidth}{0.5pt}
\end{center}

We next consider the structural requirements on refutations to
derive a stronger condition that is \emph{sufficient} for a game
to admit a refutation.
As a first step we show that we can map between words which are $\ppsim$ to each other using a restricted class of permutations.
\begin{lem}
  Let $W = w_1 w_2 \dots w_{2\ell}$ be a 1-XOR word of even length such that $W
  \ppsim I$; i.e.~there exists a parity-preserving permutation $\pi\in S_{2\ell}$ such that $\pi(W)\sim I$.  Then there exists a permutation $\pi' \in S_{2\ell}$, also satisfying $\pi'(W)\sim I$, also parity-preserving, and with an additional ``pair preserving'' property.  This means that it permutes the pairs $(1,2), (3,4), \ldots, (2\ell-1,2\ell)$ without separating or reordering the elements in each pair:
  \be \pi'(2i-1) = \pi'(2i) - 1\qquad \forall i\in[\ell].
  \label{eq:pp-pair}\ee
  \label{lem:permute_pairs}
\end{lem}

\begin{proof}
Every letter in $\pi(W)$ will cancel with a unique other letter.  We call a letter even or odd based on the parity of its location in $\pi(W)$.  Deleting a canceled pair does not change the parity of any other location, and $\pi$ also preserves the parities.  Thus the letter in location $2i$ will cancel a letter in some odd position, which we call $2f(i)-1$ (i.e.~$w_{2i}=w_{2f(i)-1}$).  Since each odd letter cancels exactly one even letter, $f$ is a permutation of $[\ell]$.
Next we decompose $f$ into disjoint cycles: $f = (i_1, i_2, \dots
  i_{\ell_1})(i_{\ell_1 + 1} \dots i_{\ell_2}) \dots(i_{\ell_{c-1} +
    1} \dots i_{\ell_c})$ where $\ell_c=\ell$. We claim that, written in two-line notation,
  \be
  t' := \begin{pmatrix}
  	1 & 2 & \dots & \ell_c \\
  	i_1 & i_2 & \dots & i_{\ell_c}
  \end{pmatrix}
  \ee
  is a permutation of the pairs satisfying
  the desired properties.
  This map from $f$ to $t'$ is known as the Foata correspondence.
  Let $\pi'$ be the corresponding pair-preserving permutation of $[2\ell]$.
  Then 
  \be \pi'(w) = \hspace{1ex}
  \underbracket{\hspace{-1ex}w_{2i_1-1} 
  \hspace{1ex}\overbracket{\hspace{-1ex}w_{2i_1} 
  w_{2i_2-1}\hspace{-5ex}}\hspace{5ex}
  \hspace{1ex}\overbracket{\hspace{-1ex} w_{2i_2} \,\cdot} 
  \cdots  
  \hspace{-1ex}\overbracket{\hspace{1ex}\cdot\, w_{2i_{\ell_1}-1}\hspace{-6ex}}\hspace{6ex}
   w_{2i_{\ell_1}} \hspace{-3.5ex}} \hspace{3.5ex}
  \hspace{1ex}\underbracket{\hspace{-1ex}w_{2i_{\ell_1+1}-1} \cdots w_{2i_{\ell_2}}
  \hspace{-4ex}}\hspace{4ex}.
  \ee
We can see that $\pi'(w)$ fully cancels following the pattern marked by the square brackets, with each cancellation using the fact that $w_{2i}=w_{2f(i)-1}$.
\end{proof}
A pair (and hence parity) preserving permutation $\pi' \in S_{2 \ell}$ can be specified uniquely by some $\pi \in S_{\ell}$, given the relation
\begin{align}
\pi(i) = \pi'(2i)/2.
\end{align}
We will frequently use of this alternate description of pair-preserving permutations, in a way made formal in Definition \ref{defn: unpacking function}. 

Before introducing this formally, we will the concept of a shuffle.

\begin{defn}
\label{defn:shuffle function}
A function $f: [\ell] \rightarrow [\ell]$ is called a shuffle function if the sequence 
\begin{align*}
f^{-1}(1), f^{-1}(2), \ldots  , f^{-1}(\ell)
\end{align*}
can be partitioned into two increasing subsequences. That is, for any shuffle function $f$, there exist disjoint increasing sequences $s_A$ and $s_B$ with $|s_A| + |s_B| = l$ and $f^{-1}$ increasing on $s_A$ and $s_B$. 

Operationally, the set of shuffle functions are the set of
permutations which can be obtained by partitioning the elements of $[\ell]$
into two sets, considering those sets as increasing sequences, and then mixing those sequences using a
dovetail (riffle) shuffle. 
\end{defn}
\begin{defn}
Let $A$ be an arbitrary set, and let $t = (a_1, a_2, \ldots  , a_\ell)$ be a sequence consisting of elements of $A$. Define the set of shuffles of $t$
\begin{align}
\shuffle(t) := \{(a_{f(1)}, a_{f(2)}, \ldots  , a_{f(\ell)}) :\, f \text{ a shuffle function} \}
\end{align}
and let this function act on sets in the natural way, so  
\begin{align}
\shuffle (\cT) := \bigcup_{t \in \cT} \shuffle(t)
\end{align}
where $\cT \subseteq A^*$ and $A^* = \bigcup_{n\geq 0} A^n$ is the set of all sequences of
elements of $A$.
\end{defn}

Shuffles are a subset of the set of permutations.  However, a standard result~\cite{bayer1992trailing} regarding dovetail shuffles states that any permutation can be expressed as a short sequence of dovetail shuffles.  Since our definition of shuffles contains a choice of partition that generalizes dovetail shuffles, the same result applies to our family of shuffles.
\begin{lem}[Theorem 1 of \cite{bayer1992trailing}] \label{lemma:permutation_from_shuffles}
Let $t$ be any sequence of length $\ell$, $p \geq \ceil{\log(\ell)}$, and let $t'$ be any permutation of $t$. Then
\begin{align}
t' \in \shuffle^p(t).
\end{align}
\end{lem}

Our next goal is constructing a gadget from $k$-XOR clauses that allows us to shuffle pairs of letters on any wire of a word without changing the overall parity bit. The construction of this gadget relies on a simpler ``shift gadget'' which allows us to move words between wires. This definition and construction are given below. 
\begin{defn} 
  For any string of letters $y = y_1 y_2 ... y_\ell$, a
  \textbf{$1 \to 2$ shift gadget} for $y$ is a
  word $S^{1 \to 2}(y)$ that equals the identity on all wires except
  the first two, and is equal to $y^{-1} := y_\ell ... y_2 y_1$ on wire $1$, i.e. a word
  of the form
  \be
  q_{i_1}q_{i_2} \ldots  q_{i_\ell} := S^{1 \to 2}(y) \sim \bbm y^{-1} \\ h(y) \\ \, \ebm,
  \ee
  for some arbitrary string of letters $h(y)$. For $\alpha, \beta \in [k]$, define $\alpha \to \beta$ shift gadgets analogously.
  
  Note that any shift gadget has a natural inverse
  \begin{align}
	q_{i_\ell}q_{i_{\ell-1}} \ldots  q_{i_1} := S^{1 \leftarrow 2}(y) \sim
    \bbm y \\ h(y)^{-1} \\ \\ \ebm = \paren{S^{1 \rightarrow 2}(y)}^{-1}.
  \end{align}
\end{defn}

Intuitively, $S^{1 \rightarrow 2}(y)$ removes $y$ from the first wire
and ``saves'' it on the second wire in the form of the string $h(y)$. 
$S^{1\leftarrow 2}(y)$ then ``loads'' $y$ back onto the first wire
while removing $y'$ from the second wire. We now use these shift gadgets
to construct a gadget that shuffles pairs of letters.

\begin{defn} \label{defn: unpacking function}
Define
$\unpack: ([n]^2)^{\ell/2} \rightarrow [n]^{\ell}$ to map sequences
of pairs into an ``unpacked'' sequence in the natural way, so that 
\begin{align}
\unpack\paren{(t_1, t_2), (t_3, t_4), \dots, (t_{\ell -1},t_{\ell})} = \paren{t_1, t_2, \dots t_{\ell}}.
\end{align}
Note that any permutation $\pi' \in S_{\ell}$ is pair preserving iff it satisfies 
\begin{align}
\pi' = \unpack \circ \pi \circ \unpack^{-1}
\end{align}
for some $\pi \in S_{\ell/2}$. 
\end{defn}

\begin{lem}[Shuffle Gadget] \label{lemma: shuffle gadget} 
Let $t = (t_1, t_2, \dots t_{\ell/2})$ be a length $\ell/2$ sequence of pairs of letters, with each $t_i := (t_i^{(1)} t_i^{(2)}) \in [n]^2$.
Let $G$ be an XOR game that contains all shift gadgets in the set \footnote{
There is nothing special about player 1 here but we state the lemma in
terms of player 1 for notational simplicity. }
\begin{align}
\left\{S^{1 \rightarrow \alpha}(t_i^{(1)} t_i^{(2)}) : \alpha \in \{\alpha_1,\alpha_2\}, i \in [\ell/2] \right\},
\end{align}
where
 $\alpha_1 \neq \alpha_2$ are elements of $[k]\backslash\curly{1}$ and each shuffle gadget has length at most $K$.
Then, for all $t' \in \shuffle(t)$, $G$ contains a word $W$ with sign $s_W = +1$, length at most $K\ell$,
and 
\begin{align*}
W \sim \begin{bmatrix} \unpack(t)^{-1} \unpack(t') \\ \\ \end{bmatrix}.
\end{align*}
\end{lem}

\begin{proof}
Let $f$ be the shuffle function satisfying $f(t) = t' \in \shuffle(t)$.
Since $f$ is a shuffle function we can
choose disjoint sequences $s_A$ and $s_B$ with $s_A \cup s_B = [\ell/2]$ and
$f^{-1}$ increasing on both. We construct a word of the desired form by
first saving the pairs in $s_A$ and $s_B$ onto wires $\alpha_1$ and $\alpha_2$,
respectively, then loading them back onto the first wire, interleaving
in the appropriate order.

For any sequence $s$, let $s^{r}$ be shorthand for that sequence written in reverse order. Define the function $g : [\ell/2] \rightarrow \{\alpha_1,\alpha_2\}$ by 
\begin{align*}
g(i)  =
\begin{cases} 
\alpha_1 \text{ if } i \in s_A \\ 
\alpha_2 \text{ if } i \in s_B \end{cases}.
\end{align*}
Then the word $W$ given below satisfies the lemma:
\begin{align}
W &= \prod_{i=1}^{\ell/2}
 \left( S^{1 \rightarrow g(i)}(t_i^{(1)}t_i^{(2)}) \right) \prod_{i=1}^{\ell/2} \left( S^{1 \leftarrow g(f^{-1}(i))}(t_{f^{-1}(i)}^{(1)}t_{f^{-1}(i)}^{(2)}) \right) \\
&\sim 
	\bbml \prod_{i \in s_{\ell/2}^r}(t_{i}^{(1)}t_{i}^{(2)})^{-1} \\  \prod_{i \in s_A^r} h(t_{i}^{(1)}t_{i}^{(2)}) \\ \prod_{i \in s_B^r} h(t_{i}^{(1)}t_{i}^{(2)}) \\ \ebml
	\bbml \prod_{i \in s_{\ell/2}}(t_{f^{-1}(i)}^{(1)}t_{f^{-1}(i)}^{(2)}) \\  \prod_{i \in {s_A}} h(t_{i}^{(1)}t_{i}^{(2)})^{-1} \\  \prod_{i \in {s_B}} h(t_{i}^{(1)}t_{i}^{(2)})^{-1} \\ \ebml
    = \bbm \unpack(t)^{-1} \unpack(t') \\ \\ \\  \ebm.
\end{align}
By assumption, $G$ contains each shift gadget used in the construction of $W$,
and each shift gadget has length at most $K$. Therefore $W$
is contained in $G$ and has length at most $2 K (\ell/2) = K \ell$.
For each shift gadget used, its inverse
is also used. By construction, the sign of each shift gadget is the
same as its inverse, so the overall sign of $W$ is $s_W = +1$.
\end{proof}

\begin{note}
\label{note:unpack}
For any game $G$ and sequence of pairs $t$ that meets the conditions of
Lemma~\ref{lemma: shuffle gadget}, $G$ will also meet the conditions for any
sequence of pairs $u = \pi(t)$ produced through permutation $\pi$ of $t$.
Then, under the assumptions of
Lemma~\ref{lemma: shuffle gadget} we get ``for free'' that a word
is contained in $G$ with sign $+1$ and has the form
\begin{align}
\bbm
\unpack(\pi(t))^{-1} \unpack(f(\pi(t))) \\
\\
\ebm
\end{align}
with $\pi$ any permutation on pairs and $f$ any shuffle function (see Definition~\ref{defn:shuffle function}).
\end{note}

Combining our newly constructed shuffle gadget with our understanding of parity preserving permutations allows us to derive a set of sufficient conditions for a game $G$ to contain a refutation. These will be used in a critical way in Section \ref{section: decidability algorithm}.

\begin{lem} \label{lemma: sufficient conditions for refutation} 
Let $G$ be a $k$-XOR game containing a length $\ell$ word $W$ whose first wire is given by
\begin{align}
W_1 = 
\bbm
w_{11} \; w_{12} \; \ldots  \; w_{1\ell} \\
\ebm \ppsim I.
\end{align}
Also let $G$ contain all shift gadgets in the set 
\begin{align}
\{S^{1 \rightarrow \alpha}(w_{1(2i-1)}w_{1(2i)}) : \alpha \in \curly{\alpha_1, \alpha_2}, i \in [\ell/2] \},
\end{align}
where $\alpha_1 \neq \alpha_2 \in [k]\backslash \curly{1}$ and
each gadget has length at most $K$.

Then $G$ contains a word with sign $+1$ and length at most $K\ell\log(\ell)$ whose first wire is given by
\begin{align}
W_1^{-1} = 
\bbm
w_{1\ell} \; w_{1(l-1)} \; \ldots  \; w_{11} \\
\ebm.
\end{align}
and which is $\sim I$ on all wires other than the first. 

\end{lem}
\begin{proof} By Lemma \ref{lem:permute_pairs} there exists a permutation $\pi$ on $[\ell/2]$ satisfying
\begin{align}
\bbm \unpack \circ \pi((w_{11}w_{12}), (w_{13}w_{14}), \ldots  , (w_{1(l-1)}w_{1\ell})) \ebm
\sim I.
\end{align}  
By Lemma \ref{lemma:permutation_from_shuffles}, there then exists a sequence $(f_1, f_2, \ldots  f_p)$ of $p \leq \log(\ell)$ shuffle functions with 
\begin{align}
f_p \ldots  f_2 f_1 = \pi.
\end{align}
Now let $\pi'$ be an arbitrary permutation of $[\ell/2]$, $f'$ be an arbitrary shuffle of $[\ell/2]$, and define the word $Y(\pi',f')$ to have first coordinate
\begin{align}
Y_1(\pi', f') := \bbm 
\unpack \circ \pi'((w_{11}w_{12}), \ldots , (w_{1(l-1)}w_{1\ell})) 
\ebm^{-1}
\bbm 
\unpack \circ f'(\pi'((w_{11}w_{12}), \ldots , (w_{1(l-1)}w_{1\ell}))) 
\ebm
\end{align}
and all remaining $k-1$ coordinates the identity.
By Lemma \ref{lemma: shuffle gadget} and Note~\ref{note:unpack}, we have that $G$ contains a word with sign $+1$ and length at most $K \ell$ which is $\sim$ $Y(\pi',f')$.

By concatenating a carefully chosen string of these words,
we see $G$ also contains a word with sign $+1$ and length at most
$K\ell\log(\ell)$ which is 
\begin{align}
\sim Y(e, f_1) Y(f_1, f_2) Y(f_2f_1, f_3) \ldots  Y(f_{k-1}f_{k-2}\ldots f_1, f_k) \sim W^{-1}_{1}.
\end{align}
\end{proof}
Lemma \ref{lemma: sufficient conditions for refutation} suggests we can construct refutations for a game $G$ by finding a word contained in $G$ which is $\ppsim I$ and has sign $-1$, and then checking to see if $G$ contains the necessary shift gadgets. First, we demonstrate that the first two wires of some permutation of such a word can be made to cancel without using any shift gadgets, then determine a sufficient set of shift gadgets required thereafter.
\begin{lem}
\label{lem:pairwise permuted word with first two wires canceled}
Let game $G$ contain word $W' = q_{i_1} q_{i_2} \dots q_{i_\ell} \ppsim I$.
There exists a permutation $\pi \in S_\ell$ such that
\be
W := q_{i_{\pi(1)}} q_{i_{\pi(2)}} \dots q_{i_{\pi(\ell)}} \ppsim I
\ee
and both $W^{(1)} \sim I$ and $W^{(2)} \sim I$ with $W^{(2)} = x_1 x_1 x_2 x_2 \dots x_{\ell/2} x_{\ell/2}$ where $x_i \in [n]$.
\end{lem}
\begin{proof}
By Lemma~\ref{lem:ppsim and even odd multisets}, we have that the even
and odd multisets of $W'$, $\cE$ and $\cO$ respectively, are multiplicity
equivalent.
Thus, for each $\alpha \in [k]$, there exists a bijection
$f_\alpha : \cE \mapsto \cO$ that maps a query
$(q^{(1)}, \dots, q^{(k)}) \in \cE$ to a query $(q'^{(1)}, \dots,
  q'^{(k)}) \in \cO$ such that $q^{(\alpha)} = q'^{(\alpha)}$. From the bijections $f_1,
  f_2$, we will define a new map  $f: \cE \cup \cO \mapsto \cE \cup \cO$
  that maps each
  query $q \in \cE$ to $f_1(q) \in \cO$ and each $q' \in \cO$ to $f_2^{-1}(q') \in
  \cE$. Since $f_1$ and $f_2$ are bijections, so is $f$. Applying the Foata
  correspondence, as in Lemma~\ref{lem:permute_pairs}, to the
  permutation of $\cE \cup \cO$ associated with $f$ yields a sequence
  of queries that make a word $W$
  with the property that the first two wires completely
  cancel to identity and wire 2 takes the desired form, i.e. 
  \[ W = \bbm \hspace{2ex}\overbracket{\hspace{-2ex} w_{11} \hspace{2ex}\underbracket{ \hspace{-2ex}w_{12}  \cdot} \cdots
      \underbracket{\cdot w_{1(\ell - 1)}\hspace{-6ex}}\hspace{6ex}
      w_{1 \ell}\hspace{-3ex}}\hspace*{3ex}    \\ \hspace{2ex}\underbracket{\hspace{-2ex} w_{21} w_{22} \hspace{-3ex}}\hspace{3ex}
    \cdot \cdots \cdot \hspace{2ex}\underbracket{\hspace{-2ex}w_{2 (\ell - 1)}
      w_{2 \ell}\hspace{-3ex}}\hspace*{3ex} \\ W^{(3)} \\ \cdots \\ W^{(k)} \ebm
    \sim \bbm \\ \\ W^{(3)} \\ \cdots \\ W^{(k)} \ebm , \] 
  where $W^{(3)}, \dots, W^{(k)}$ are even-length strings of letters.
\end{proof}

For a refutation to exist we then simply need to be able to shuffle
the pairs on the remaining wires $3, \dots, k$.

\begin{thm}[Sufficient condition for refutation]
\label{thm:Sufficient condition for refutation}
Let $G$ be a PR game which by definition contains some
word $W' \ppsim I$ of some even length $\ell$.
Let $W \ppsim I$ be the pairwise permuted word as in
\lemref{pairwise permuted word with first two wires canceled}.
If $G$ contains all shift gadgets in the set
\be
\curly{S^{\alpha \rightarrow \alpha'}(W^{(\alpha)}_{2i-1} W^{(\alpha)}_{2i}) : \alpha \in \curly{3, \dots, k}, \alpha' \in \curly{1,2}, i \in [\ell/2]}
\ee
then $G$ contains a refutation.
\end{thm}
\begin{proof}
By the definition of a PR game (Definition~\ref{defn:PREF}), $G$
contains the word $W$ with sign $-1$. By
Lemma~\ref{lemma: sufficient conditions for refutation},
$G$ contains all words $W''_\alpha$ with $\alpha$-th wire
given by $\paren{W''_\alpha}^{(\alpha)} = \paren{W^{(\alpha)}}^{-1}$,
all other wires $\sim I$,
and sign $+1$. Therefore $G$ contains the word
\begin{equation}
R := W \prod_{\alpha} W''_\alpha \sim I
\end{equation}
with sign $s_R = -1$, which is a refutation.
\end{proof}

It turns out that for the special case of 
symmetric XOR games, the symmetric structure guarantees existence
of all required shift gadgets automatically. 
Theorems~\ref{thm:Sufficient condition for refutation} and
\ref{thm:Necessary condition for refutation}
then give that a symmetric game contains a PREF if and only if it
contains a refutation. Further, whether a game contains a PREF
is an efficiently decidable criterion.
A formal definition of symmetric XOR games and
a proof of these facts are demonstrated in
Section~\ref{section: decidability algorithm}.


\subsection{Algorithm for Symmetric Games}
\label{section: decidability algorithm}
\label{subsec: Algorithm}
We begin with a formal definition of symmetric games.

\begin{defn}
  A $k$-XOR game $G$ is \textbf{symmetric} if whenever it contains a clause
  $c = (q^{(1)}, q^{(2)}, \dots, q^{(k)}, s)$,
  it also contains all clauses
  $c' = (q^{(\pi(1))}, q^{(\pi(2))}, \dots, q^{(\pi(k))}, s)$,
  where $\pi: [k] \mapsto [k]$ permutes the query while the parity
  bit $s$ is unchanged. 
\end{defn}

\begin{defn}
A \textbf{random symmetric $k$-XOR game} $G_{sym}$ on $m = k!m'$ clauses is a game constructed by randomly generating $m'$ clauses, then including all clauses related by
permutations (as above) in $G_{sym}$. 
\end{defn}


For symmetric games, we can now prove that all required shift gadgets
are certainly included.
\begin{lem}
Let $W$ be a word contained in symmetric game $G$ of even length
$\ell$ with second wire of the form
$W^{(2)} = x_1 x_1 x_2 x_2 \dots x_{\ell} x_{\ell}$,
where $x_i \in [n]$.
For any wire $\alpha \in \{3, \dots, k\}$ and pairs of letters
$y_1, y_2$ that appear at adjacent positions $2i - 1, 2i$
in $W^{(\alpha)}$, there exists shift gadgets from
$\alpha \to 2$ and from $\alpha \to 1$ for $y_1y_2$ with
length $O(1)$.
\label{lem:symm_shift_gadget}
\end{lem}
\begin{proof}
Since the game is symmetric, it suffices to show the existence of
the gadget for $\alpha \to 2$.
Let the queries containing $y_1, y_2$ in $W$ be
$q_1 = (q_1^{(1)}, q_1^{(2)}, \dots, y_1, \dots)$
and $q_2 = (q_2^{(1)}, q_2^{(2)}, \dots,  y_2, \dots)$, respectively. Then by the
assumption of symmetry, all permutations of these queries exist in the given
game. We can thus construct the shift gadget $S^{\alpha
\to 2}(y_1y_2)$ by the product of four clauses as follows:
\begin{equation}
  S^{\alpha \to 2}(y_1y_2)
    = \bbm q_2^{(1)} \\ q_2^{(2)} \\ \dots \\ y_2 \\ \dots \ebm \bbm q_2^{(1)} \\
    y_2 \\ \dots  \\ q_2^{(2)} \\ \dots  \ebm \bbm q_1^{(1)}\\  y_1
      \\ \dots \\ q_1^{(2)} \\ \dots \ebm \bbm q_1^{(1)} \\ q_1^{(2)} \\ \dots \\
      y_1  \\ \dots\ebm  = \begin{bmatrix} \\ h(y_1 y_2)  \\ \\ y_2 y_1 \\ \phantom{-} \end{bmatrix},
\end{equation}
where $h(y_1 y_2) := q_2^{(2)} y_2 y_1 q_1^{(2)}$ and the equality holds
because $y_1$ and $y_2$ appear at an odd and
following even position of $W$ so by the form of the second wire
$q_1^{(2)} = q_2^{(2)}$.
\end{proof}

We now prove Theorem~\ref{thm:exists_decidability_alg}, by showing
that the PREF criterion is both necessary and sufficient for a
symmetric game to have a refutation, and can also be expressed as a
system of linear Diophantine equations and thus solved efficiently.

\begin{repthm}{thm:exists_decidability_alg}
  There exists an algorithm that, given a $k$-player symmetric XOR game $G$
  with alphabet size $n$ and $m$ clauses, decides
  in time $\poly(n,m)$ whether $\omega^*(G) = 1$ or
  $\omega^*(G) < 1$.
\end{repthm}
\begin{proof}
  By Theorem~\ref{thm:sat_iff_no_refutations}, deciding whether the
  commuting-operator value is $1$ is equivalent to deciding whether
  the game admits a refutation (of any length).
  By Theorem~\ref{thm:Necessary condition for refutation}
  for a game to admit a refutation it is necessary that it
  contains a PREF.
  Further, \thmref{Sufficient condition for refutation} and
  Lemma~\ref{lem:symm_shift_gadget} together show that
  for a symmetric game to admit a refutation it is also sufficient
  to contain a PREF. Thus for a symmetric game,
  deciding whether $\omega^* = 1$ reduces to determining whether
  or not the game contains a PREF.
  
  
  We can now rephrase the condition for a game to contain a PREF as a
  system of linear Diophantine equations. For each query in the game
  $q_i = (q_i^{(1)}, \dots, q_i^{(k)})$, let $z_i$ be an
  integer-valued variable
  representing the number of times query $i$ appears in the even
  multiset of the PREF minus the number of times it appears
  in the odd multiset.
  The condition that these $z_i$ in fact correspond to multiplicity
  equivalent sets can then be stated as a system of linear
  Diophantine equations,
  \begin{equation}
	A^\transp z = 0
    \label{eq:dioph}
  \end{equation}
  where $A$ is the
  game matrix as defined in \defref{game matrix} and we have collected the $z_i$
  into a vector $z \in \mathbb{Z}^m$.
  The condition that the signs of the clauses in the PREF multiply to
  $-1$ can be expressed as an additional linear Diophantine equation
  in terms of $z$ and parity bit vector $\hat{s}$ (\defref{game matrix}):
  \begin{equation}     \hat{s}^\transp z = 1
    \pmod{2}. \label{eq:dioph_sign}
  \end{equation}
  By applying a standard algorithm, such as the one described in Chapter~5
  of~\cite{schrijver86}, this system can be solved in time polynomial in the
  size of the system, measured as the number of bits necessary to specify the
  system of equations. This means a runtime that is $\poly(n, m)$.

\end{proof}
\begin{note}
Finding a solution to \eqref{eq:dioph} and \eqref{eq:dioph_sign} tells us not only that a refutation exists but also bounds its length. In particular, by following the steps of the preceding proof, it can be shown that for a symmetric game with $\omega^*(G) < 1$, the minimum-length refutation has length L satisfying 
\[ \Omega(\|z\|_1) \leq L \leq O(k \|z\|_1 \log \|z\|_1),\]
where $z$ is a solution to (108) and (109) minimizing $\|z\|_1$.
\end{note}

We demonstrate an explicit refutation construction using the symmetric
game shift gadgets in the context of
the Capped GHZ game in Appendix~\ref{subsec: Explicit refn Capped GHZ}.

Finally, note that this linear algebraic description of
the necessary PREF
criterion for an entangled refutation parallels
the classical condition for refutation
(Definition~\ref{defn: classical refutation}).
The only distinction is that \eqref{eq:dioph} is
considered an equation over $\bbF_2$
for classical games and over $\bbZ$ for entangled games.
As described in Section~\ref{subsec:MERP PESc Duality},
these Diophantine equations
then give rise to a dual condition similar to the classical picture:
a MERP strategy achieves value 1 exactly when these equations
do not admit a solution.


\section{MERP Strategies}
\label{sec:MERP Strategies}
Section \ref{subsec:Intro MERP Strategies} introduced the family of
Maximal Entanglement, Relative Phase (MERP) strategies. The primary
feature of the MERP strategies is that deciding whether
$v^\MERP = 1$ and computing the accompanying MERP strategy vector
may be done efficiently via Gaussian elimination.
Beyond computability, the MERP strategies actually achieve value 1
on a large class of games where that is possible. Specifically,
MERP achieves value 1 exactly where a \PREF{} does not exist
(no\PR{} games), including all symmetric value 1 games.
This MERP - \PR{} duality is analogous to the duality between a classical
linear algebraic refutation and the construction of a classical value 1
strategy.

Here, we motivate the definition of MERP strategies (Section \ref{subsec:MERP motivation}) and prove their value defined in \claimref{MERP value} (Section \ref{subsec:MERP strategy details}). We then investigate the duality between MERP value 1 and \PREF{}s (Section \ref{subsec:MERP PESc Duality}).

\subsection{Generalizing GHZ}
\label{subsec:MERP motivation}
The MERP family of strategies is motivated by the GHZ strategy for
solving the GHZ game. We begin by reviewing the GHZ game and
value 1 strategy.
\begin{defn}
Recall that the \textbf{GHZ game} is defined by the clauses
\begin{equation}
G_{GHZ} := \curly{
\bbm \blue{x} \\ \blue{x} \\ \blue{x} \\ +1 \ebm,
\bbm \orange{y} \\ \orange{y} \\ \blue{x} \\ -1 \ebm,
\bbm \orange{y} \\ \blue{x} \\ \orange{y} \\ -1 \ebm,
\bbm \blue{x} \\ \orange{y} \\ \orange{y} \\ -1 \ebm
}.
\end{equation}
\end{defn}
The GHZ strategy~\cite{greenberger1990bell}, defined as follows, achieves value 1 for this game.
\begin{defn}
Define the \textbf{GHZ Strategy} for $G_{GHZ}$ to be the tensor-product strategy in which:
\begin{enumerate}
\item The $k = 3$ players share the maximally entangled state
\begin{align}
\ket{\Psi} = \frac{1}{\sqrt{2}} \left[ \ket{000} + \ket{111} \right]
\end{align}
with player $\alpha$ having access to the $\alpha$-th qubit of the state.
\item Upon receiving symbol $j$ from the verifier, player $\alpha$
rotates his qubit by an angle
\begin{align} \label{eqn:GHZ theta choice}
\theta(\alpha,j) = \begin{cases}
0 \text{ if } j = \blue{x} \\
\frac{\pi}{2} \text{ if } j = \orange{y}
\end{cases}
\end{align}
about the Z axis, then measures his qubit in the X basis and sends his observed outcome to the verifier.
Defining the states $\ket{\theta(\alpha,j)_{\pm}}$ by
\begin{align}
\ket{\theta(\alpha, j)_+} &= \frac{1}{\sqrt{2}} \left[\ket{0} + e^{i\theta(\alpha, j)}\ket{1}\right] \text{ and }\\ 
\ket{\theta(\alpha, j)_-} &=  \frac{1}{\sqrt{2}} \left[\ket{0} - e^{i\theta(\alpha, j)}\ket{1}\right]
\end{align}
the GHZ strategy may be given by the strategy observables
\begin{align}
O^{\alpha}(j) := \ket{\theta(\alpha,j)_+}\bra{\theta(\alpha,j)_+} - \ket{\theta(\alpha,j)_-}\bra{\theta(\alpha,j)_-}.
\end{align}
\end{enumerate}
\end{defn}

We now consider why this strategy is successful. Recall that a $\varphi$ rotation in the $Z$ basis is represented by the operator
\be e^{i\varphi/2} \ketbra 0 + e^{-i\varphi/2}\ketbra 1.\ee  Thus the
rotations $\varphi_1, \varphi_2, \varphi_3$ applied by the players
to their shared state $\ket{\Psi}$ results in
\begin{align}
\ket{\Psi_{\varphi}} &:= \frac{1}{\sqrt{2}} \left[
	e^{-i\frac{\varphi}{2}} \ket{000} +
    e^{i\frac{\varphi}{2}} \ket{111} \right] \\
\varphi &:= \varphi_1 + \varphi_2 + \varphi_3.
\end{align}
Note that $X\ot X\ot X \ket{\Psi_{\varphi}} = \ket{\Psi_{-\varphi}}$. This
gives expected value of the measurements performed by the three players,
\be
\bra{\Psi_{\varphi}} X\ot X\ot X \ket{\Psi_{\varphi}} =
\frac{e^{i\varphi} + e^{-i\varphi}}{2} = \cos \varphi.
\ee
Thus the \emph{relative phase} between the kets $\ket{000}$ and
$\ket{111}$ introduced by the $Z$ rotations determines the probabilities
that the players output $+1$ or $-1$. For the GHZ game, the prescription
for Z rotations given in (\ref{eqn:GHZ theta choice}) results in
relative phase $\varphi = 0$ for the first clause and $\varphi = \pi$
for the remaining three clauses, exactly matching the desired outputs.

This description of GHZ motivates the MERP family as a generalization.
For a game $G$, the MERP construction assigns a distinct angle to
each player-question combination such that the relative phase for each
query in $G$ gives optimal probability of winning. The set of
games for which
MERP can achieve value 1 is exactly the set for which the game
admits independently setting the relative phase for each query to
$\pi \hat{s}_i$. This is exactly the statement of \claimref{MERP Gaussian Elimination}.

We proceed by recalling the definition of a MERP strategy in light of
the GHZ analogue, proving our value claim, and finally demonstrating
the duality with \PR{} games.

\subsection{MERP Strategy Value}
\label{subsec:MERP strategy details}
Recall the definition of a MERP strategy:
\begin{repdefn}{defn: MERP}
Given a $k$-XOR game $G$ with m clauses, a \textbf{MERP strategy} for $G$ is a tensor-product strategy in which:
\begin{enumerate}
\item The $k$ players share the maximally entangled state 
\begin{align}
\ket{\Psi} = \frac{1}{\sqrt{2}}\left[\ket{0}^{\otimes k} + \ket{1}^{\otimes k}\right]
\end{align}
with player $\alpha$ having access to the $\alpha$-th qubit of the state.
\item Upon receiving question $j$ from the verifier, player $\alpha$ rotates his qubit by an angle $\theta(\alpha,j)$ about the $Z$ axis, then measures his qubit in the $X$ basis and sends his observed outcome to the verifier.

Explicitly, we define the states 
\begin{align}
\ket{\theta(\alpha, j)_{\pm}} &:= \frac{1}{\sqrt{2}} \left[\ket{1} \pm e^{-i\theta(\alpha, j)}\ket{0}\right]  
\end{align}
and pick strategy observables
\begin{align}
O^{\alpha}(j) := \ketbra{\theta(\alpha,j)_+}{\theta(\alpha,j)_+} - \ketbra{\theta(\alpha,j)_-}{\theta(\alpha,j)_-}.
\end{align}
\end{enumerate}
\end{repdefn}

We now demonstrate that a MERP strategy achieves the claimed tensor-product (and thus \co{}) value
by explicit calculation.
\begin{repclaim}{claim:MERP value}
The value achieved by that MERP strategy on game $G$ is:
\begin{align} \label{eqn:MERP value vector}
v^{\MERP}(G,\hat{\theta}) :=& \frac{1}{2} + \frac{1}{2m} \paren{ \sum_{i=1}^{m} \cos(\pi \left[ (A \hat{\theta})_i - \hat{s}_i \right]) } \\
=& \frac{1}{2} + \frac{1}{2m} \paren{ \sum_{i=1}^{m} \cos(
\sum_{\alpha=1}^k \theta(\alpha,q_{i}^{(\alpha)}) - \pi \hat{s}_i) }.
\label{eqn:MERP value angle}
\end{align}
\end{repclaim}
\begin{proof} 
Consider a particular clause $c_i = (q_i,s_i)$. We calculate the probability
that a MERP strategy parameterized by $\theta(\alpha,q_{i}^{(\alpha)})$
returns output $s_i$ correctly.

If players $1,\ldots, k$ apply rotations by $\varphi_1,\ldots,\varphi_k$ to their qubits
in state
$\ket{\Psi} = \frac{1}{\sqrt{2}}\left[\ket{0}^{\otimes k} + \ket{1}^{\otimes k}\right]$
then they will be left with 
\ba
 \ket{\Psi^k_\varphi} 
&:= \frac{1}{\sqrt{2}}\left[e^{i\frac{\varphi}{2}}\ket{0}^{\otimes k} 
+ e^{-i\frac{\varphi}{2}} \ket{1}^{\otimes k}\right]
\\
\varphi & := \varphi_1 + \ldots + \varphi_k.
\ea
Note that $X^{\ot k}  \ket{\Psi^k_\varphi} =  \ket{\Psi^k_{-\varphi}}$.
The expected value of the product of the $k$ measurements is then
\be 
\bra{\Psi^k_\varphi}  X^{\ot k}\ket {\Psi^k_\varphi} = 
\frac{e^{i\varphi} + e^{-i\varphi}}{2} = \cos\varphi
.\ee

We now plug in the values from the clause and the corresponding angles in the MERP
strategy.  The angles are $\varphi_\alpha = \theta(\alpha, q_i^\alpha)$ so that 
\be \varphi =
\sum_{\alpha \in [k]}\theta(\alpha, q_i^\alpha) = (A\hat\theta)_i.
\label{eq:varphi-equiv}\ee  The probability of
obtaining the correct answer for the clause is 
\be 
\frac{1 + s_i \bra{\Psi^k_\varphi}  X^{\ot k} \ket{\Psi^k_\varphi}}{2}
 = \frac{ 1  +s_i \cos(\varphi)}{2}
 = \frac{ 1  +\cos(\varphi - \pi \hat s_i)}{2}.\ee
Averaging over all clauses and substituting \eqref{eq:varphi-equiv} for $\varphi$ we obtain \eqref{eqn:MERP value vector}
and \eqref{eqn:MERP value angle}.
\end{proof}

\subsection{MERP - \PR{} Duality}
\label{subsec:MERP PESc Duality}

It is well-known that the structure of the game matrix over $\mathbb{F}_2$
gives insight into the classical value of an XOR game.
The construction of a classical value 1 strategy is dual to 
the existence of a classical refutation. In much the same way,
the construction of a \co{} value 1 MERP strategy is dual to
the existence of a \PREF{}.

MERP is restricted to achieving value 1 on only a subset of \co{} value 1 XOR games. By the duality to
\PR{} this subset is exactly those games that our algorithm
can decide have value 1. In particular, this means that all symmetric
games with value 1 can be played optimally using MERP, making it a
powerful family of strategies.

We begin with a review of the classical value 1 - refutation duality,
which informs our later proof of the MERP - \PR{} duality.
From \claimref{Classical strategy value}, we have the value of a
classical strategy
\begin{align}
v(G, \vcstrat) &= \frac{1}{2} + \frac{1}{2m} \paren{\sum_i \cos(\pi \left[ \paren{A\vcstrat}_i - \hat{s}_i \right] )}
\end{align}
where the vector algebra is taken over $\mathbb{F}_2$. 
Using this
linear algebraic form for the value, we can prove \claimref{Classical value 1}.
\begin{repclaim}{claim:Classical value 1}
Every solution $\vcstrat \in \mathbb{F}_2^{kn}$ to
\be
A \vcstrat = \hat{s} \tag{\ref{eq:cl_linalg}}
\ee corresponds to a
strategy $\eta$ achieving
value 1 on game $G \sim (A,\hat{s})$, and vice versa. In particular, 
a game $G$ has classical value $1$ iff (\ref{eq:cl_linalg}) has a
solution. 
\end{repclaim}
\begin{proof}
If a solution $\vcstrat$ exists,
\begin{align}
v(G,\vcstrat) &= \frac{1}{2} + \frac{1}{2m} \paren{\sum_i \cos(\pi \left[ \paren{A\vcstrat}_i - \hat{s}_i \right] )} \\
	&= \frac{1}{2} + \frac{1}{2m} \paren{\sum_i \cos(0)} = 1.
\end{align}
Conversely, to achieve value 1, we must have the argument of every cosine
equal to some multiple of $2\pi$. Therefore we need
$A\vcstrat - \hat{s} = 0$ over $\mathbb{F}_2$.
\end{proof}

Recall that this classical value 1 constraint has a dual set of equations, such that
there exists a classical \emph{refutation} that solves the dual equations
if and only if the classical value 1 constraints are not satisfiable.
\begin{repfact}{fact:cl_duality}
Either a classical refutation $y$ exists satisfying
\begin{equation}
\begin{bmatrix} A^\transp \\ \hat{s}^\transp \end{bmatrix} y = \begin{bmatrix} 0 \\ 1 \end{bmatrix}
\tag{\ref{eq:cl_refut}}
\end{equation}
or a classical
strategy $\hat{\eta}$ exists satisfying \eqref{eq:cl_linalg}.
\end{repfact}
We use an analogous duality relation to prove the MERP - \PR{} duality shortly.

Before that, we mention one more consequence of this characterization of
classical value 1 games --
a linear algebraic specification, in terms of game matrix $A$,
of the set of $\hat{s}$ for which the
game $G \sim (A,\hat{s})$ has $\omega(G) = 1$.
\begin{defn}
Define the vector space $\cYtwo \subseteq \mathbb{F}_2^{m}$ by
\begin{equation}
\cYtwo := \curly{ A \vcstrat : \vcstrat \in \mathbb{F}_2^{kn} } = \im_{\mathbb{F}_2}\paren{A}
\end{equation}
Define the dimension of this vector space as
\begin{equation}
\nli := \dim \cYtwo.
\end{equation}
\end{defn}

\begin{cor}
\label{cor:Classically achievable set is ytwo}
Given a game matrix $A$, the set of possible accompanying $\hat{s}$
that produce a game $G \sim (A,\hat{s})$ with classical value 1 is
exactly the $2^{\sigma_2}$ parity-bit vectors in $\cYtwo$.
\end{cor}
\begin{proof}
This follows immediately from \claimref{Classical value 1}.
\end{proof}

The main use of \corref{Classically achievable set is ytwo} is to
characterize the classical value of games with randomly chosen
$s_i$ (Section~\ref{subsec: APD Games}).

\begin{center}
\rule{0.2\textwidth}{0.5pt}
\end{center}

We now use an analogue of Fact~\ref{fact:cl_duality} to demonstrate that the set of games on which MERP achieves \co{} value 1 is exactly the
complement of those for which a \PREF{} specification exists. First, recall the MERP constraint equations that define the set
of games for which MERP achieves value 1.
\begin{repclaim}{claim:MERP Gaussian Elimination}
A MERP strategy achieves $v^{\MERP} = 1$ on a game $G$ iff its MERP constraint equations
\begin{equation}
A \hat{\theta} = \hat{s} \pmod 2
\end{equation}
have a solution $\hat{\theta} \in \mathbb{Q}^{kn}$.
\end{repclaim}
\begin{proof}
If a solution $\hat{\theta}$ exists, (\ref{eqn:MERP value vector})
gives the MERP value using this strategy vector:
\begin{equation}
v^{\MERP}(G,\hat{\theta}) = \frac{1}{2} + \frac{1}{2m} \paren{\sum_{i=1}^{m} \cos{\paren{0}}} = 1.
\end{equation}
Conversely, the only way to achieve value $m$ inside the sum over
cosines is for the argument to each cosine to be a multiple of $2\pi$.
This is only possible if $(A\hat{\theta})_i - \hat{s}_i = 0 \pmod 2$
for each $i$.
\end{proof}


\begin{repthm}{thm:Diophantine equations unsolvable are varphi2}
Either there exists a MERP refutation $z \in \mathbb{Z}^m$ satisfying the
\PR{} equations
\begin{align}
A^\transp z &= 0 \label{eqn:lin dio 1}\\
\hat{s}^\transp z &= 1 \pmod 2 \label{eqn:lin dio 2}
\end{align}
or a MERP strategy with value 1 exists for game $G \sim (A,\hat{s})$.
\end{repthm}
\begin{proof}
We begin by reformatting the linear Diophantine equations
(\ref{eqn:lin dio 1}) and (\ref{eqn:lin dio 2})
to remove the modulo 2 and collect the \PR{} constraints into
a single matrix equation
\begin{equation}
\label{eqn:lin dio reformed}
\begin{bmatrix} A^\transp & 0 \\ \hat{s}^\transp & 2 \end{bmatrix}
\begin{bmatrix} z \\ z' \end{bmatrix} = \begin{bmatrix} 0 \\ 1 \end{bmatrix}
\end{equation}
with $z' \in \mathbb{Z}$.

By \cite[Corollary 4.1a]{schrijver86}, the dual to 
(\ref{eqn:lin dio reformed}) is the system of constraints
\begin{align}
\label{eqn:lin dio dual 1}
\begin{bmatrix} A & \hat{s} \\ 0 & 2 \end{bmatrix}
\begin{bmatrix} w \\ w' \end{bmatrix}
&\in \mathbb{Z}^{m+1} \qquad \text{ and } \\
\label{eqn:lin dio dual 2}
\begin{bmatrix} 0 & 1 \end{bmatrix} \begin{bmatrix} w \\ w' \end{bmatrix}
&\notin \mathbb{Z}.
\end{align}
Here ``dual'' means that (\ref{eqn:lin dio dual 1}) and (\ref{eqn:lin dio dual 2}) are
satisfiable iff (\ref{eqn:lin dio reformed}) is unsatisfiable.
The bottom rows of (\ref{eqn:lin dio dual 1}) and (\ref{eqn:lin dio dual 2})
can be satisfied iff
\be
w' = a + \frac{1}{2}, \, a \in \mathbb{Z}.
\label{eqn:lin dio dual prefinal}
\ee
The top row of (\ref{eqn:lin dio dual 1}) then becomes:
\begin{align}
A w + \hat{s} w' &= a' \in \mathbb{Z}^m \\
\Leftrightarrow A(2w) + (2a + 1) \hat{s} &= 2a' \\
\Leftrightarrow A(2w) &= \hat{s} \pmod 2.
\label{eqn:lin dio dual final}
\end{align}
Setting $\hat{\theta} = 2w$ and picking arbitrary $a \in \mathbb{Z}$, (\ref{eqn:lin dio dual prefinal}) and (\ref{eqn:lin dio dual final})
can be satisfied
iff there is a solution to
\be A\hat{\theta} = \hat{s} \pmod 2, \, \hat{\theta} \in \mathbb{Q}^{kn}. 
\ee
\end{proof}

\thmref{Diophantine equations unsolvable are varphi2} tells us that every game that we can decide has value 1 using the algorithm of Section~\ref{section: decidability algorithm} also has an accompanying MERP strategy with value 1. Further, we demonstrated in that section that a symmetric game contains a PREF iff it
has value $\omega^* < 1$. We conclude that for symmetric games, the
MERP family of strategies achieves value 1 everywhere it is possible to
do so.

Following the classical case, it is also illuminating to note a
linear algebraic specification, in terms of game matrix $A$, of the
$\hat{s}$ for which a MERP strategy can achieve value 1 on
game $G \sim (A,\hat{s})$. First, we define a mapping between
the space in which the image of $A$ lives,
$\im_\mathbb{Q}(A) \subseteq \mathbb{Q}^m$,
and the space in which the parity bits live, $\hat{s} \in \mathbb{F}_2^m$.

\begin{defn}
Define a mapping\footnote{Note $\varphi_2$ is not in general a linear function, but it is linear over inputs in $\mathbb{Z}^m$.} $\varphi_2 : \mathbb{Q}^m \rightarrow \mathbb{F}_2^m$ by 
$$\varphi_2(z) := \begin{cases} z \pmod 2 \text{ if } z \in \mathbb{Z}^m \\
0 \text{ otherwise. }
\end{cases} $$ 
\end{defn}

Now, we can define an analogue to $\cYtwo$,
here considering $A$ as a map over $\mathbb{Q}$ and naturally extending $\varphi_2$ to act on subsets of $\mathbb{Q}^m$.
\begin{defn}
\label{def: cYQ}
Define the vector space $\cYQ \subseteq \mathbb{F}_2^m$ by 
\begin{equation}
\cYQ := \varphi_2 (\im_\mathbb{Q}(A)).
\end{equation}
\end{defn}

We then find that, accounting for the $\varphi_2$ technicality
due to the mod 2 involved in computing an overall output,
the set of games with MERP value 1 is the image of $A$ over $\mathbb{Q}$.
\begin{cor}
\label{cor:MERP value 1 is varphi2 cYQ}
Given a game matrix $A$, the set of possible accompanying $\hat{s}$
that produce a game $G \sim (A,\hat{s})$ with MERP value 1 is exactly
the parity bit vectors in $\cYQ$.
\end{cor}
\begin{proof}
This follows directly from \claimref{MERP Gaussian Elimination}.
\end{proof}

In this sense, \corref{MERP value 1 is varphi2 cYQ} demonstrates
that the advantage of MERP over a classical strategy
is simply exploiting entanglement to
enable the players to ``output'' values in $\mathbb{Q}$ instead of
$\mathbb{F}_2$.

\section{Specific Games}
\label{sec: Value 1 Landscape}
In this section we use the machinery of the previous sections to construct some games with interesting properties.

The first is a simple game, the 123 Game, which illustrates conditions under which the \PREF{} condition can be fooled. It is a relatively small (6 player, 6
query) non-symmetric game which does contain a \PREF{}, but still does not contain any refutations. We show this by giving an explicit value 1 strategy for the 123~Game. 

The second is a family of games, called Capped GHZ ($\CG$),
which are designed to be hard instances for the ncSoS algorithm.
In particular, the $\CG$ game on $n$ variables (denoted $\CGn$) is a
symmetric game with value strictly less than 1, meaning the decision
algorithm of Section~\ref{section: decidability algorithm} can show
the game has value $<1$ in poly time, but with a minimum refutation of
length at least exponential in $n$. This shows a doubly exponential
improvement in the runtime of our decision algorithm as compared to
the ncSoS algorithm, and an exponential improvement over the previous
best known ncSoS lower bounds for this
problem~\cite{HNW16}\footnote{In fact, to our knowledge, our results
  are the first exponential degree lower bound for the ncSoS hierarchy
  applied to \emph{any} problem.}. This game construction is based primarily on the
theorems of Section~\ref{sec: Refutations}, which outline the
relationship between refutations and ncSoS runtime, as well as our
decision algorithm. 

Finally, we construct a family of games with \co{} value $1$ and a low classical value. These games are called Asymptotically Perfect Difference ($\APD$) games, and are parameterized by $K$. 
The classical value of the $K$-th $\APD$ game ($\APDK$) approaches
$1/2$, which is the lowest possible, in the limit of large $K$. The existence of
such a family was posed as an open question in~\cite{BrietV13}. The construction of these games is based primarily on the difference between the linear equations defining MERP value 1 and classical value 1, which is discussed in
Section~\ref{subsec:MERP PESc Duality}.

These games are summarized in the following table, with a full discussion of each in the subsequent sections.

\begin{table}[ht]
\centering
\begin{tabular}{ c | c c c c c P{1.2 in}}
\toprule
Game & $n$ & $m$ & $k$ & $\omega^*$ & $\omega$ & minimum refutation length \\ \hline 
& & & & & & \\[-8pt] 
123 Game & $3$ & $6$ & $6$ & $\mathbf{1}$ & $5/6$ & $-$ \\[5pt]
$\CGn$ & $n$ & $3n-1$ & $3$ &  $< 1 - 1/\exp(n)$ & $1-1/m$ & $\mathbf{2^{n+1} - 2}$ \\[5pt]
$\APDK$ & $2$ & $2^K$ & $2^{K}-1$ & $\mathbf{1}$ & $\mathbf{1/2 + \sqrt{K/2^K}}$ & $-$ \\
\bottomrule
\end{tabular}
\caption{Overview of the games constructed in this section. Quantities of note are denoted in bold.}
\end{table}




\subsection{123 Game}
\label{subsec: 123 Game}
We begin with a discussion of the intuition behind the 123 game, then follow with an explicit value  1 strategy. It is instructive to begin by analyzing the ``Small 123 Game''.
\begin{defn}
Define the \textbf{Small 123 Game} to be the $k = 3$ player game
with $n = 3$ and $m = 6$ clauses
\begin{equation}
G_{123}^{\text{small}} := \curly{
\bbm 1 \\ 1 \\ 1 \\ 1 \ebm,
\bbm 1 \\ 2 \\ 3 \\ 1 \ebm,
\bbm 3 \\ 3 \\ 3 \\ -1 \ebm,
\bbm 2 \\ 3 \\ 1 \\ 1 \ebm,
\bbm 2 \\ 2 \\ 2 \\ 1 \ebm,
\bbm 3 \\ 1 \\ 2 \\ 1 \ebm
}.
\end{equation}
\end{defn}
In this form, it is clear the Small 123 Game has
$\omega^{*}(G_{123}^{\text{small}}) < 1$, since
placing its clauses in the order presented forms a refutation.

The game matrix $A$ has a one-dimensional left nullspace (corresponding
to the space of candidate \PREF{} specifications $z$ satisfying $A^\transp z = 0$):
\begin{equation}
z \propto \begin{bmatrix} 1 & -1 & 1 & -1 & 1 & -1 \end{bmatrix}^\transp.
\end{equation}
Any odd multiples of this basis vector produce a \PREF{} specification $z$.

We now add players to this game while preserving this \PREF{} specification,
until we exclude all refutations formed by permutations of
a single copy of each of the clauses.
To preserve the \PREF{} specification, for each question $j$
given to a new player, we ensure $j$ is given to the player once in an even clause
(2, 4, or 6) and once in an odd clause (1, 3, or 5).
We must add three players to exclude all possible reorderings of the length-$6$ refutation given
by the clauses of the Small 123 Game, and in doing so end up with the ``123 Game''
(clauses reordered to expose the game structure):
\begin{defn}\label{def:123 Game}
Define the \textbf{123 Game} by the following set of clauses:
\begin{equation}
G_{123} := \curly{
\bbm 1 \\ 1 \\ 1 \\ 1 \\ 1 \\ 1 \\ 1 \ebm,
\bbm 2 \\ 2 \\ 2 \\ 2 \\ 2 \\ 2 \\ 1 \ebm,
\bbm 3 \\ 3 \\ 3 \\ 3 \\ 3 \\ 3 \\ -1 \ebm,
\bbm 1 \\ 2 \\ 3 \\ 1 \\ 2 \\ 3 \\ 1 \ebm,
\bbm 2 \\ 3 \\ 1 \\ 3 \\ 1 \\ 2 \\ 1 \ebm,
\bbm 3 \\ 1 \\ 2 \\ 2 \\ 3 \\ 1 \\ 1 \ebm
}.
\label{eqn:G123 clauses}
\end{equation}
\end{defn}
The 123-game has been constructed to make it difficult to reorder valid PREF specifications into refutations (for instance, it can be shown that no permutation of the valid length-6 PREF specifications 
\[
\pm \begin{bmatrix} 1& 1& 1 & -1 &-1 &-1\end{bmatrix}^\transp
\]
corresponds to a valid refutation).




In Section~\ref{section: decidability algorithm} we
demonstrated that 
symmetric games have a refutation whenever they contain a PREF
by construction of all required shift gadgets.
In the (non-symmetric) 123 Game, there are no obvious shift gadgets present. This structure gives some intuition for why one would expect
this game to have value 1 even though it has a \PR{}. In the next section we prove that this intuition is correct; the 123-Game does in fact have value 1. 

\subsubsection{Value 1 Strategy}

We define a simple strategy: measure in the Z basis if sent a 1, X if sent a 2, and Y if sent a 3.\footnote{This is motivated by the observation that the 123 game provably does not have a refutation if we assume the measurements for different questions anticommute. We plan on addressing this intuition formally in an upcoming paper.} If each player plays the 123 Game uses this strategy, it results in the following set of query observables: 

\begin{equation}
\cQ_{123} := \curly{
\bbm Z \\ Z \\ Z \\ Z \\ Z \\ Z  \ebm,
\bbm X \\ X \\ X \\ X \\ X \\ X  \ebm,
\bbm Y \\ Y \\ Y \\ Y \\ Y \\ Y  \ebm,
\bbm Z \\ X \\ Y \\ Z \\ X \\ Y  \ebm,
\bbm X \\ Y \\ Z \\ Y \\ Z \\ X  \ebm,
\bbm Y \\ Z \\ X \\ X \\ Y \\ Z  \ebm
}.
\label{eqn:G123 operators}
\end{equation}

We also define a state on which these measurement can be made.\footnote{This state was found through simple trial and error.} 
\begin{align}
\ket{\psi_{123}} := \frac{1}{\sqrt{8}}\left(\Big[\ket{000000} + \ket{111111} \Big] - \Big[\ket{100100} + \ket{001010} + \ket{010001} + \ket{011011} + \ket{110101} + \ket{101110} \Big]\right)
\end{align}
\begin{thm}
The strategy observables in $\cQ_{123}$ measured on the state $\ket{\psi_{123}}$ win the 123 Game with probability 1. (The 123 Game has value 1.)
\end{thm}
\begin{proof}
	For every string in $\ket{\psi_{123}}$, its compliment is also in $\ket{\psi_{123}}$ with the same sign. Additionally, every string in $\ket{\psi_{123}}$ has even Hamming weight. Overall, we may then conclude
\begin{align}
XXXXXX\ket{\psi_{123}} = ZZZZZZ\ket{\psi_{123}} =\ket{\psi_{123}}
\end{align}
and hence 
\begin{align}
YYYYYY\ket{\psi_{123}} = (i)^6XXXXXX\Big[ZZZZZZ\ket{\psi_{123}}\Big] = -\ket{\psi_{123}}.
\end{align}
It remains to check the outcomes for the last 3 queries. Explicit calculation gives
\begin{align}
ZXYZXY \ket{000000} = (-1) \ket{0} \ket{1} (-i) \ket{1} (-1) \ket{0} \ket{1} (-i) \ket{1} = -\ket{011011}. 
\end{align}
as well as 
\begin{align}
ZXYZXY \ket{111111} = \ket{1} \ket{0} i\ket{0} \ket{1} \ket{0} i \ket{0} = -\ket{100100}
 \end{align}
Similarly, we can check
\begin{align}
ZXYZXY \ket{001010} = (-1) \ket{0} \ket{1}  i\ket{0} (-1) \ket{0} \ket{0} (-i) \ket{1} = \ket{010001}.
\end{align}
and
\begin{align}
ZXYZXY \ket{110101} = \ket{1} \ket{0} (-i) \ket{1} \ket{1} \ket{1} i \ket{0} = \ket{101110}.
\end{align}
Putting this all together we see
\begin{align}
ZXYZXY \ket{\psi_{123}} = \ket{\psi_{123}},
\end{align}
with similar (permuted) arguments holding for $XYZYZX$ and $YZXXYZ$.
\end{proof}

\subsection{Capped GHZ (\texorpdfstring{$\CG$}{CG}) Games}
\label{subsec: Capped GHZ Game}
We begin by considering a family of symmetric games with \co{}
value $< 1$. The key property of this family is that
to detect that $\omega^* < 1$ requires an
exponentially high level in the ncSoS hierarchy, whereas 
the algorithm presented in Section~\ref{section: decidability algorithm} can do
so in polynomial time.
\begin{defn}
Define the $n$-th order \textbf{Capped GHZ game} as the
$3$-XOR game with alphabet size $n$
and $m = 3n - 1$ clauses defined by
\begin{align}
\CGn := \left\{
\bbm 1 \\ 1 \\ 1 \\ -1 \ebm, \bbm 1 \\ 2 \\ 2 \\ +1 \ebm, \bbm 2 \\ 1 \\ 2 \\ +1 \ebm,
\bbm 2 \\ 2 \\ 1 \\ +1 \ebm, \bbm 2 \\ 3 \\ 3 \\ +1 \ebm, \bbm 3 \\ 2 \\ 3 \\ +1 \ebm,
\bbm 3 \\ 3 \\ 2 \\ +1 \ebm,
\dots,
\bbm (n-1) \\ n \\ n \\ +1 \ebm,
\bbm n \\ (n-1) \\ n \\ +1 \ebm, \bbm n \\ n \\ (n-1) \\ +1 \ebm,
\bbm n \\ n \\ n \\ +1 \ebm  \right\}.
\end{align}
\end{defn}

We claim $\omega^*(\CGn) < 1$, and that it requires level at least
$2^{n+1}-2$ in the ncSoS hierarchy to detect this fact. 
Define the $i$-th triple of $\CGn$ to be the clause set 
\begin{align}
A_i := \left\{ \bbm i \\ (i+1) \\ (i+1) \\ +1 \ebm,
\bbm (i+1) \\ i \\ (i+1) \\ +1 \ebm,
\bbm (i+1) \\ (i+1) \\ i \\ +1 \ebm \right\}.
\end{align}
The clauses 
\begin{align}
\bbm 1 \\ 1 \\ 1 \\ -1 \ebm \text{ and } \bbm n \\ n \\ n \\ +1 \ebm 
\end{align}
are called the \textit{caps} (upper and lower) of the game and, for notational convenience, are referred to by $A_0$ and $A_n$.
Our first claim shows that any refutation for $\CGn$ must include both the upper and lower caps. 
\begin{lem} \label{lem: Capped GHZ refutations contain caps}
Let $\cE$, $\cO$ be minimal multiplicity equivalent multisets of queries taken from $\CGn$, so $\cE \sim \cO$ and no clause appears in both $\cE$ and $\cO$.
If $\cE \uplus \cO$ contains some $x \in A_j$ with $j \notin \{ 0 , n \}$, then $\cE \uplus \cO$ also contains clauses drawn from $A_{j-1}$ and $A_{j+1}$. 
\end{lem}
\begin{proof}
Without loss of generality, we assume 
\begin{align}
x = \bbm j \\ (j+1) \\ (j+1) \\ +1 \ebm \in \cE.
\end{align}
We then proceed by contradiction. If no clause from $A_{j-1}$ is contained in $\cO$ then the multiplicity of letter $j$ for wire 1 in $\cO$ cannot match $\cE$, and the contradiction is immediate.

To prove the second claim, assume $x$ occurs in $\cE$ with multiplicity $\lambda$, and no terms from $A_{j+1}$ are contained in $\cO$. Then, in order to match the $(j+1)$ multiplicity on the $2$nd and $3$rd wires, clauses
\begin{align}
y_1 = \bbm (j+1) \\ j \\ (j+1) \\ +1 \ebm \text{ and }
y_2 = \bbm (j+1) \\ (j+1) \\ j \\ +1 \ebm
\end{align}
must both occur in $\cO$ with multiplicity $\lambda$. Then we find $(j+1)$ occurs on the first wire of $\cE$ with multiplicity $0$, and on the first wire of $\cO$ with multiplicity $2\lambda$. Then $\cE$ and $\cO$ cannot be multiplicity equivalent, and this contradiction proves our result. 
\end{proof}

A bound on the minimum length refutations for $\CGn$ follows in a straightforward manner from Lemma~\ref{lem: Capped GHZ refutations contain caps}.

\begin{thm} \label{thm:long refutation for capped GHZ}
The minimal length refutation for $\CGn$ has length at least $2^{n+1}-2$.
\end{thm}

\begin{proof}
We show the minimal sized multiplicity equivalent multisets
$\cE$ and $\cO$ formed by elements of $\CGn$ have size at least $2^{n+1}-2$. By Lemma~\ref{lem: Capped GHZ refutations contain caps} the lower cap $A_0$ of $\CGn$ is contained in either $\cE$ or $\cO$. Without loss of generality, assume it is contained in $\cE$.

Then $\cE$ contains letter $1$ on every wire, and by minimality we know $A_0 \cap \cO = \emptyset$. Since $\cE$ and $\cO$ are multiplicity equivalent multisets, we conclude $A_1 \subseteq \cO$. But then $\cO$ has two $2$s on each wire, and by minimality $A_1 \cap \cE = \emptyset$. So we conclude $\left(A_2\right)^2 \in \cE$, where the notation $A_2^2$ denotes the multiset containing two copies of each element of $A_2$, and containment of one multiset in another implies containment of each element with at least it's multiplicity. Continuing in this vein, we see (assuming even $n$ for the assignment of $A_n$ below, though this does not affect the counting):
\begin{align}
A_0 \uplus \left(A_2\right)^2 \uplus \left(A_4\right)^8 ... \left(A_{n}\right)^{2^{n-1}} \subseteq \cE \text{ and } A_1 \uplus \left(A_3\right)^4 \uplus \left(A_5\right)^{16} ... \left(A_{n-1}\right)^{2^{n-2}} \subseteq \cO.
\end{align}
The total number of clauses contained in $\cE \uplus \cO$ is then given by
\begin{align}
1+3(2^0)+3(2^1)+ \dots 3(2^{n-2})+2^{n-1} &= 1 + 3(2^{n-1}-1) + 2^{n-1} \\
&= 2^{n+1} - 2.
\end{align}
Any refutation gives rise to even and odd multiplicity equivalent
multisets $\cE$ and $\cO$, and the above demonstrates that their
combined size must be $\geq 2^{n+1} - 2$, proving the lower bound
on refutation length.
\end{proof}

\thmref{long refutation for capped GHZ} shows that there exists a pseudodistribution on the clauses of $\CGn$ which appears to have value 1 to a level exponential in the ncSoS hierarchy (proving \thmref{Capped GHZ general result}). The minimal length multisets
constructed in the proof of Theorem \ref{thm:long refutation for capped GHZ}
are in fact multiplicity equivalent and the parity bits multiply to $-1$ (there
is exactly one copy of $A_0$, which is the only question with $s_i = -1$) meaning $\CGn$ contains a \PREF{}.
Since $\CGn$ is a symmetric game, these two properties are sufficient to
ensure a refutation exists (Section~\ref{section: decidability algorithm}) giving $\omega^*(\CGn) < 1$.

\subsection{Asymptotically Perfect Difference (APD) Games}
\label{subsec: APD Games}
We next construct a family of $k$-XOR games, parameterized by
$K \in \mathbb{N}$, with $k = 2^K - 1$, $m = 2^K$ clauses,
and asymptotically perfect difference: each game in the family is
a no\PREF{} game, meaning
\begin{equation}\omega^*(\APDK) = 1,\end{equation}
while
\begin{equation}\omega(\APDK) \like \frac{1}{2} + \sqrt{\frac{K}{2^K}} \like \frac{1}{2} + \sqrt{\frac{\log{k}}{k}}\end{equation}
indicating that the difference is asymptotically as large as possible,
\begin{equation}
\lim_{K \rightarrow \infty} 2 \paren{\omega^*(\APDK) - \omega(\APDK)}
	= 1.
\end{equation}

\begin{defn}
Define the \textbf{Asymptotically Perfect Difference} family of XOR games parameterized by a scale $K \in \mathbb{N}$ as the set of games with alphabet size $n=2$, $k = 2^{K}-1$ players, and $m = 2^{K}$ clauses:
\begin{equation}
\APDK := \curly{\bbm q_i \\ s_i \ebm : q_{i}^{(\alpha)} =
  B_{(K)}^{\alpha, i} + 1  }.
\end{equation}
The $s_i$ are defined to adversarially minimize $\omega(\APDK)$ and the matrix $B_{(K)} \in \curly{0,1}^{2^K \times 2^K}$ is recursively defined by
\begin{align}
B_{(0)} &= \begin{bmatrix} 1 \end{bmatrix} \\
B_{(K+1)} &= \begin{bmatrix} \bar{B}_{(K)} & B_{(K)} \\ B_{(K)} & B_{(K)} \end{bmatrix}
\end{align}
with $\bar{B}$ produced by switching $0 \leftrightarrow 1$ for all entries of $B$. Equivalently, $\bar{B} = J - B$, with $J$ the all-ones matrix.

Note that by the game definition, the $m \times kn = (2^{K}) \times (2*(2^{K}-1))$ game matrix $A_{(K)}$ for $\APDK$ consists of the first $2^K - 1$ columns of $B_{(K)}$ interleaved with the first $2^K - 1$ columns of $\bar{B}_{(K)}$:
\begin{equation}
A_{(K)} = \begin{bmatrix}
B_{(K)}^{\cdot,1} & \bar{B}_{(K)}^{\cdot,1} &
B_{(K)}^{\cdot,2} & \bar{B}_{(K)}^{\cdot,2} & \dots &
B_{(K)}^{\cdot, 2^K - 1} & \bar{B}_{(K)}^{\cdot,2^K - 1}
\end{bmatrix}.
\end{equation}

The pairs of columns in $A_{(K)}$ corresponding to the two possible outputs from each player
are complementary, making $A_{(K)}$ a valid game matrix.
\end{defn}

Note that $APD_2$ is exactly the GHZ game, so the APD family is a
particular many-player generalization of GHZ:
\begin{align}
B_{(2)} &= \begin{bmatrix}
	\bar{B}_{(1)} & B_{(1)} \\
	B_{(1)} & B_{(1)} \end{bmatrix} \\
&= \begin{bmatrix} 
	1 & 0 & 0 & 1 \\
    0 & 0 & 1 & 1 \\
    0 & 1 & 0 & 1 \\
    1 & 1 & 1 & 1
\end{bmatrix} \\
\implies A_{(2)} &= \begin{bmatrix}
	1 & 0 & 0 & 1 & 0 & 1 \\
    0 & 1 & 0 & 1 & 1 & 0 \\
    0 & 1 & 1 & 0 & 0 & 1 \\
    1 & 0 & 1 & 0 & 1 & 0
\end{bmatrix}.
\end{align}
Exchanging columns $3 \Leftrightarrow 6$ and $4 \Leftrightarrow 5$,
corresponding to a relabeling of players and inputs,
gives $A_{GHZ}$ as defined in (\ref{eqn: GHZ game matrix}). The choice of
parity bits in GHZ is known to minimize the classical value, exactly
matching the definition of $APD_2$.

We now prove our claims about the \co{} and classical values of APD games.

\subsubsection{\CO{} Value}
\begin{lem}
\label{lem: BK has trivial kernel}
For all $K$, $B_{(K)}$ has trivial kernel.
\end{lem}
\begin{proof}
We proceed by induction.
\begin{enumerate}
\item \textbf{Base case:} $B_{(0)} = \begin{bmatrix} 1 \end{bmatrix}$ has trivial kernel by inspection.
\item \textbf{Induction step:} Assume $B_{(K)}$ has trivial kernel, i.e. $B_{(K)} x = 0 \implies x = 0$. We now demonstrate that $B_{(K+1)}$ has trivial kernel by contradiction.

Assume to the contrary that $B_{(K+1)} x = 0$ for $x \neq 0$. We can expand the blocks of this equation:
\begin{align}
B_{(K+1)} x &= \begin{bmatrix} \bar{B}_{(K)} & B_{(K)} \\ B_{(K)} & B_{(K)} \end{bmatrix}
	\begin{bmatrix} x_1 \\ x_2 \end{bmatrix} = 0.
\end{align}
By the bottom block,
\begin{align}
B_{(K)} x_1 + B_{(K)} x_2 = 0 & &\\
B_{(K)} (x_1 + x_2) = 0 & &\\
\implies x_2 = -x_1. & & \text{(Induction hypothesis)}
\end{align}
Using this relation in the top block, we have
\begin{align}
0 &= \bar{B}_{(K)} x_1 - B_{(K)} x_1 \\
&= \paren{J - B_{(K)} - B_{(K)}} x_1 \\
2B_{(K)} x_1 &= J x_1 \label{eqn: BK to J relation} \\
2B_{(K)} x_1 &= \begin{bmatrix} \sum_i x_1^i \\ \sum_i x_1^i \\ ... \end{bmatrix}. \label{eqn: BK to expanded J relation}
\end{align}
Noting that the bottom row of $B_{(K)}$ is always the all-ones vector by the definition, we can consider the bottom element of (\ref{eqn: BK to expanded J relation})
\begin{align}
2 \sum_i x_i &= \sum_i x_i \\
\implies \sum_i x_i &= 0.
\end{align}
This means $J x_1 = 0$, which together with (\ref{eqn: BK to J relation}) gives:
\begin{equation}
B_{(K)} x_1 = 0.
\end{equation}
By the induction hypothesis, this must mean $x_1 = 0 = x_2$, contradicting $x \neq 0$.
\end{enumerate}
\end{proof}

\begin{thm}
\label{thm:APD entangled value 1}
For all $K$, $\APDK$ is a no\PREF{} game, 
and thus has a MERP strategy with value 1
and $\omega^*(\APDK) = 1$.
The same holds for any choice of $\hat s$.
\end{thm}
\begin{proof}
First, we demonstrate that $(A_{(K)})^\transp$ has trivial kernel.

We have from Lemma \ref{lem: BK has trivial kernel} that $B_{(K)}$ has trivial kernel, and thus its rank is $m = 2^{K}$. $A_{(K)}$ includes all columns of $B_{(K)}$ except the last, the all-ones vector. $A_{(K)}$ also includes columns of $\bar{B}_{(K)}$. Adding a column of $B_{(K)}$ to the corresponding column of $\bar{B}_{(K)}$ produces the all-ones vector, so it must be in the column-span of $A_{(K)}$ as well. Finally, this means the column span of $A_{(K)}$ includes the column span of $B_{(K)}$ and so the rank must be $m$. By the rank-nullity theorem, matrix $(A_{(K)})^\transp$ has trivial kernel.

The \PR{} constraints are unsatisfiable, so $\APDK$ is a noPREF game.
By \thmref{Diophantine equations unsolvable are varphi2}, $\APDK$ has a
MERP strategy with value 1 and $\omega^*(\APDK) = 1$.
\end{proof}

\subsubsection{Classical Value}
We extend the motivating classical results
presented in Section \ref{subsec:MERP PESc Duality} to analyze the
classical value of the APD family.
\corref{Classically achievable set is ytwo} demonstrates that
the set of outputs achievable by a deterministic classical strategy is 
given exactly by $\cYtwo := \im_{\mathbb{F}_2}\paren{A}$.
Recalling that $\sigma_2 = \dim \cYtwo$, we see that when
$\sigma_2 \ll m$, the set of deterministically achievable outputs is 
much smaller than the total space of possible parity bit 
vectors, and so we should be able to find a vector
$\hat{s} \in  \mathbb{F}_2^m$ with large Hamming distance from all
outputs in $\cYtwo$. In this section the probabilistic method
is used to formalize this argument. 

\begin{thm} \label{thm:asymptotic classical value bound} Let $A$ be an
XOR game matrix, for which $\sigma_2 \leq \delta m$. Then there exists a parity bit vector $\hat{s} \in \mathbb{F}_2^{m}$ for which the game
$G \sim (A,\hat{s})$ has value at most 
\begin{align}
\frac{1}{2} +  \sqrt{\frac{\delta}{2}}
\end{align}
\end{thm}
\begin{proof}
This argument is a close variant of the usual Hamming bound on error-correcting codes.
Let $S$ denote the set of $\hat s$ within distance $m(1/2-\epsilon)$ of some point in
$\cY_2$.  Using the fact that $|\cY_2| = 2^{\sigma_2} \leq 2^{\delta
  m}$ we have
\be |S| \leq 2^{\delta m} \sum_{k\leq m\paren{\frac 12-\epsilon}} \binom{m}{k}.\ee
We bound the sum over binomial coefficients with the Chernoff bound to obtain
\be |S| \leq 2^{\delta m} 2^{m(1-2\epsilon^2)}.\ee
Then for any $\eps > \sqrt{\delta/2}$ there exists a $\hat s$ with
distance $\geq m(1/2-\eps)$ from any point in $\cY_2$.  This corresponds to value $1/2 +
\eps$. 
\end{proof}

We now consider the specific case of the APD game and demonstrate
the asymptotic limit of the classical value.
\begin{lem}
\label{lem:APD sigma2}
Given $K \in \mathbb{N}$, the APD game $\APDK$ has $\sigma_2(\APDK) = K+1$.
\end{lem}
\begin{proof}
Recall that $\sigma_2$ is the dimension of $\cYtwo$, the image of $A_{(K)}$ viewed as a map taking $\mathbb{F}_2^{kn} \rightarrow \mathbb{F}_2^{m}$. Equivalently, $\cYtwo$ is the column span of $A_{(K)}$ taken over $\mathbb{F}_2$, and for this proof we use this view. By the same argument as \thmref{APD entangled value 1}, the
column span of $A_{(K)}$ is identical to the column span of
$B_{(K)}$. We prove this Lemma by induction over the $B_{(K)}$:
\begin{enumerate}
\item \textbf{Base case:} $B_{(0)} = \begin{bmatrix} 1 \end{bmatrix}$ giving $\sigma_2 = 1$ by inspection.
\item \textbf{Induction step:} Assume $\sigma_2(\APDK) = K+1$, meaning the dimension of the column span of $B_{(K)}$ over $\mathbb{F}_2$ is $K+1$. We can write $B_{(K+1)}$ in block format:
\begin{equation}
B_{(K+1)} = \begin{bmatrix} \bar{B}_{(K)} & B_{(K)} \\ B_{(K)} & B_{(K)} \end{bmatrix}
= \begin{bmatrix} (J - B_{(K)}) & B_{(K)} \\ B_{(K)} & B_{(K)} \end{bmatrix}
= \begin{bmatrix} (J + B_{(K)}) & B_{(K)} \\ B_{(K)} & B_{(K)} \end{bmatrix} (\text{over } \mathbb{F}_2) \label{eqn: Block expansion of AKp1 over F2}
\end{equation}
All columns in the right block of (\ref{eqn: Block expansion of AKp1 over F2}) take the form $\begin{bmatrix} x & x \end{bmatrix}^\transp$,
so their span is
\be
\cS := \curly{\begin{bmatrix} r & r \end{bmatrix}^\transp : r \in \cYtwo(\APDK)}.
\ee
On the other hand, all columns in the left block take the form $\begin{bmatrix} 1 \oplus x & x \end{bmatrix}^\transp = \begin{bmatrix} 1 & 0 \end{bmatrix}^\transp + \begin{bmatrix} x & x \end{bmatrix}^\transp$. The form of the right block span guarantees $\begin{bmatrix} 1 & 0 \end{bmatrix}^\transp$ is linearly independent from the right columns. Thus the total column span is
\begin{equation}
\cYtwo(APD_{K+1}) = \cS \cup \paren{\cS + \begin{bmatrix} 1 \\ 0 \end{bmatrix}}
\end{equation}
and $\sigma_2(APD_{K+1}) = \sigma_2(\APDK) + 1 = (K+1) + 1$, completing the induction step.
\end{enumerate}
\end{proof}

\begin{repthm}{thm:APD-violation} 
The APD family has classical value
\be
\dfrac 12 \leq 
\omega(\APDK) 
\leq \dfrac{1}{2} + \sqrt{\frac{K+1}{2^{K+1}}}
\leq \dfrac{1}{2} + \sqrt{\dfrac{\log{k}}{k}}.
\ee
\end{repthm}
\begin{proof}
The lower bound of $\frac 12$ applies to all XOR games since a random assignment will
satisfy half the clauses in expectation.

For the first upper bound, note that for APD family, $m=2^K$ and from \lemref{APD sigma2},
$\sigma_2 = K+1$. 
Then 
\thmref{asymptotic classical value bound} yields the bound
\be 
\omega(\APDK) \leq
\frac 12 + \sqrt{\frac{\sigma_2}{2m}} = 
\frac 12 + \sqrt{\frac{K+1}{2^{K+1}}}
\leq 
\frac 12 + \sqrt{\frac{K}{2^K}}.\ee
The last bound in the theorem statement is obtained by noting that $K=2^k$.
\end{proof}

Finally, we conclude by mentioning that even though the APD construction may require an
exponential time to choose the adversarial $s_i$, one can achieve the same asymptotic
difference with high probability by choosing the $s_i$ uniformly at random.  This is
implicit in the proof of 
\thmref{asymptotic classical value bound}, which implies that a randomly chosen $\hat s$
has value $\geq 1/2+\eps$ with probability $\leq 2^{(\delta-2\eps^2)m}$.  Note as well
that $\omega^* = 1$ for any choice 
of $s_i$, according to \thmref{APD entangled value 1}.

\section{Random Games}
The previous sections give a complete characterization of symmetric games with \co{} value 1. However, as demonstrated by the final example of the previous Section~(\ref{subsec: 123 Game}), non-symmetric games remain, in general, hard to characterize.
One area where we can make some progress is in understanding the value of
randomly generated XOR games. We will work in a model, specified in
Definition~\ref{def:random_game}, where each clause is
sampled uniformly with replacement from the set of all possible clauses.

The classical value of random CSPs\footnote{As noted in Section~\ref{sec:
    Background}, CSPs and games are closely related. Classically, the difference between a CSP and the associated symmetric game is that each player
  in a game may play according to a different assignment of the variables; thus,
  the value of a CSP is always less than or equal to the classical value of the
  associated symmetric game.}  in this model has been intensely
studied for several predicates including XOR, and it is useful to summarize the
classical results. While determining the exact classical value of a random
$k$-XOR instance for $k \geq 3$ remains hard, 
union bound arguments can give probabilistic bounds on the classical value of
random $k$-XOR instances, in terms of the number of variables $n$ and the number of
clauses $m$. Combining these with second moment-type arguments and combinatorial
analysis has revealed the existence of SAT and UNSAT phases for random
instances in the limit of large $n$, which are separated by a sharp threshold in
$m$~\cite{PS16}. For $k=3$, this threshold occurs at $m/n \approx
0.92$~\cite{DM02}. When $m/n$ is below the threshold, a random 3-XOR instance has value $1$
with probability approaching $1$ as $n \to \infty$, while when $m/n$ is above the
threshold, a random instance has value $1$ with probability approaching
$0$, and in fact, in the UNSAT phase, it is known that the value is
close to $1/2$. In addition to the true value, one can study the performance of the SoS
algorithm on random instances. A key result in this direction is that
of Grigoriev~\cite{Gri01}, who showed the existence of a region in the UNSAT
phase with classical value close to $1/2$, but for which the classical SoS algorithm 
reports a classical value of $1$ until a high level in the SoS hierarchy. In the language we have 
developed thus far, he showed this by proving that random XOR games with
appropriately chosen $m$ and $n$ do not admit any short-length classical
refutations. One can interpret this result as showing the existence of a phase
which is both UNSAT and computationally intractable. 

The goal of this section is to prove a quantum analogue of these
results. We are limited in one important sense: classically, the existence of an
UNSAT phase with value close to $1/2$ is shown via a union bound over the set of
possible classical strategies, but this tool is no longer available to us for
quantum strategies. Using our refutation-based technology, the best we can say
is that the \co{} value of a game is bounded a small distance away
from 1 (see Section \ref{subsubsec: gap}). At the same time, the quantum case
presents us with an opportunity to go beyond what is possible classically: while
the classical SoS algorithm
has a natural upper bound at level $kn$, no such bound exists for the
ncSoS algorithm.  We could thus potentially improve on Grigoriev's result to prove a superexponential lower
bound on the runtime of ncSoS.

We work subject to these considerations. In one direction, we know that for any
$G$, $\omega^*(G) \geq \omega(G)$, and so we immediately get an entangled SAT phase for 3-XOR games with $m \lessapprox 0.92n$.
In the other direction we show the existence of an entangled UNSAT phase:
specifically, we show that there exists a constant $C_k$ depending only on the
number of players $k$ such that random games with more than $C_kn$ queries have
\co{} value $<1$ with high probability. For 3-XOR games we find $C_3 \lessapprox
4$. Our bounds on the entangled SAT and UNSAT phases are only a constant factor
apart, leaving open the possibility of a sharp threshold behavior as in the
classical case. 

Further, in analogy with Grigoriev's results, we also show that random XOR games with $m
= O(n)$ queries have, w.h.p., no refutation with length less than
$\Omega\left({n \log(n)}/{\log(\log(n))}\right)$. By Lemma \ref{lem:no
  refutation of length l implies degree l pseudodistribution}, this implies
ncSoS takes superexponential time to show these games have value $< 1$.  

\subsection{SAT Phase}\label{subsec:sat_phase}
To start, we will show how the existence of a SAT phase for $k$-XOR viewed as a
CSP implies the existence of such a phase for $k$-player XOR games. This is a
simple consequence of the connection between games and CSPs.

\begin{lem} \label{lem:csat2qsat}
  For every $k$-XOR game $G$ with $m$ clauses and $n$ variables, there exists a
  corresponding $k$-XOR CSP instance $\Phi_G$ with the same number of clauses and
  variables, such that if $\mathrm{val}(\Phi_G) = 1$, then $\omega(G) =
  1$. Moreover, when $G$ is chosen at random according to the distribution in
  Definition~\ref{def:random_game}, the induced definition over $\Phi_G$ is
  the one generated by uniformly sampling $m$ clauses over $n$ variables with replacement.
\end{lem}
\begin{proof}
  For each clause $(q_{i_1}, \dots, q_{i_k}, s_{i}) \in G$, create a clause
  $x_{i_1} x_{i_2} \dots x_{i_k} = s_{i} $ in $\Phi_G$, where the variables $x_i$
  are taken over $\{\pm 1\}$. (We allow $\Phi_G$ to contain repeated clauses.) It is clear that $\Phi_G$ has the same number of
  clauses and variables as in $G$, and that if $G$ is random then $\Phi_G$ is
  distributed as in the lemma statement.

  If $\Phi$ has value $1$, let $x$ be a
  satisfying assignment for $\Phi$. Then the classical strategy where all
  players play according to $x$ is a strategy for $G$ achieving value $1$.  Hence $\mathrm{val}(\Phi_G) = 1$
  implies that $\omega(G) = 1$.
\end{proof}

\begin{cor}
  For every $k$, there exists a constant $B_k$ such that for any $b < B_k$, a
  random $k$-XOR game with $m = bn$ clauses will have $\omega(G) = \omega^*(G) =
  1$ with probability approaching $1$ as $n \to \infty$.
\end{cor}
\begin{proof}
  The analogous statement for $k$-XOR CSP instances is proved
  in Theorem 16 of~\cite{PS16}. Let $B_k$ be the threshold appearing in that theorem.
  By  Lemma~\ref{lem:csat2qsat}, if we sample a random $k$-XOR game $G$ with $m = bn$
  clauses, then the associated CSP instance $\Phi_G$ will be a random CSP
  instance with $bn$ clauses, and thus have value $1$ with probability
  approaching $1$. Hence, $\omega(G) = \omega^*(G) = 1$ with probability
  approaching $1$ as well.
\end{proof}
For $k = 3$, the constant $B_k$ can be computed to be $\approx 0.92$~\cite{DM02}. 

\subsection{UNSAT Phase}\label{subsec: Random Games}
Since we are considering general random XOR games, we cannot appeal
to the shift gadgets available by symmetry. Instead we use  probabilistic analysis
to show that such gadgets exist with high probability, given enough clauses.
Below we give the analysis for the specific case of random 3-XOR games. 
The analysis for general $k$ proceeds identically, with different
constants depending on the number of players.
\begin{lem} \label{lemma: shift gadgets for random games}
Let $G$ be a randomly generated 3-XOR game defined by the set $M$ of queries and
associated parity bits, with $|M| = m \geq 3.3n$. Then with probability $1-o(1)$, there
exists a set $N_{3,1} \subseteq [n]$  with $|N_{3,1}| > {0.95n}$ such that for all
$a, b \in N_{3,1}$, $G$ contains a shift gadget  
\begin{align*}
S^{3 \rightarrow 1}(ab).
\end{align*}
\end{lem}
\begin{proof}
Consider a bipartite graph between two sets of $n$ vertices. Label one set of vertices by $([n],3)$, and the other by $([n],2)$. Add an edge between $(j,3)$ and $(j',2)$ iff there exists a query 
\begin{align*}
\begin{bmatrix}
r \\
j' \\
j \\
\end{bmatrix} \in M
\end{align*}
where $r \in [n]$ is arbitrary. Label the edge by the index of the query corresponding to it. Our key observation is that that $S^{3 \rightarrow 1}(ab)$ can be constructed from the queries corresponding to a walk from $(a,3)$ to $(b,3)$ in the graph. 

Because queries are randomly generated, edges in this graph are randomly generated as well. So our graph is a $G_{n, n, m}$ Erd\"os-R\'enyi random bipartite graph. A technical result (Lemma \ref{lem: Random graphs are equivalent}) gives that this graph is at least as connected as $\hat{G}_{n, n, p}$ -- a random bipartite graph in which each edge is present independently with probability $p = m/n^2 - \epsilon/n = (3.3 -\epsilon)/n$, where $\epsilon$ is an arbitrary small constant.

Finally, applying a Galton-Watson style argument to this random graph shows~\cite[Theorem~9]{johansson2012giant} that with probability $1-o(1)$ it contains a ``giant component'' that touches at least ${\gamma n}$ vertices of $([n],3)$, where $\gamma$ is the unique solution in the interval $(0,1]$ to the equation
\begin{align*}
\gamma + \exp\left(pn(\exp(-pn\gamma) - 1)\right) = 1  \implies \gamma > 0.95.
\end{align*}
\end{proof}

\begin{lem}[Relating random graph models] \label{lem: Random graphs are equivalent}
Let $G \sim G_{N,N,m}$ with $m = CN$. Further let $\hat{G} \sim \hat{G}_{N,N,p}$ with $p = (C-  \epsilon)/N$, for arbitrary small constant $\epsilon$. For any value $Z$, if $\hat{G}$ contains a connected component of size $Z$ with probability $1-o(1)$ then $G$ also contains a connected component of size $Z$ with probability $1-o(1)$.
\end{lem}
\begin{proof}
We couple the distributions used to generate $\hat{G}$ and $G$. In particular, a graph $G$ can be generated by choosing a graph $\hat{G}$, then randomly adding or removing edges until the graph has exactly $m$ edges. As long as we only add edges, this process will only increase the size of the largest connected component in the graph. Letting $E(\hat{G})$ be the set of edges of a graph $\hat{G}$, we find
\begin{align*}
\E{|E(\hat{G})|} = N^2p = (C - \epsilon)N
\end{align*}
and so
\begin{align*}
\p{|E(\hat{G})|>m} \leq \exp(-\epsilon^2N/3) = o(1)
\end{align*}
by a Chernoff bound. 
\end{proof}
Lemma \ref{lemma: shift gadgets for random games} tells us that, given a large enough number of queries, most variables can be shuffled in exactly the manner described in Section \ref{sec:shuffle}. If we consider only queries involving these variables, we should then be able to construct refutations from PREFs using exactly the techniques described in the later half of that section. In fact, we only need to restrict to those variables on $k-2$ of the wires, since cancellations on the first two wires are automatic (see, in particular, the proof of Lemma~\ref{lem:pairwise permuted word with first two wires canceled}).

If a large enough number of queries remain one would expect that they  admit a \PR{} with high probability. This fact is proved below. 

\begin{lem}
\label{lem:rank nullity PESc}
For any $k$-XOR game $G$ with $m$ queries, alphabet size $n$ and
\begin{align}
m - kn = \delta > 0,
\end{align}
if the parity bits for $G$ are picked randomly then $G$ has a \PREF{} with probability at least $1 - 2^{-\delta}$.
\end{lem}
\begin{proof}
By definition, a \PREF{} specification is any vector $z \in \mathbb{Z}^m$ satisfying
\begin{align}
A^\transp z &= 0 \text{ and } \label{eqn: parity equivalent sets condition}\\
\hat{s}^\transp z &= 1. \label{eqn: parity bits give refutation condition}
\end{align}
When $m > kn$, the matrix $A^\transp$ has rank $\leq kn$.
By the rank-nullity theorem, the kernel of $A^\transp$ has dimension $\geq m - kn$, and so the are at least $\delta$ linearly independent vectors $z$ satisfying \eqref{eqn: parity equivalent sets condition}. If the parity bits are chosen randomly, each of these vectors $z$ satisfy \eqref{eqn: parity bits give refutation condition} with probability $1/2$, and the result follows. 
\end{proof}

Finally, we use our lemmas
to prove the specific $k=3$ case of the random game threshold.
\begin{repthm}{thm: k-XOR non-sym unsat}
Let $G$ be a random 3-XOR game with $m = \ceil{3.3n}$ clauses on an alphabet of size $n$. Then, with probability $1 - o(1)$, $G$ has value $<1$. 
\end{repthm}
\begin{proof}
Let $N_{3,1}$ be defined as in Lemma \ref{lemma: shift gadgets for random games}, and extend this definition to $N_{3,2}$ analogously. Define 
\begin{align}
N_3 := N_{3,1} \cap N_{3,2}.
\end{align}
Let $\gamma$ be defined as in Lemma \ref{lemma: shift gadgets for random games}. A union bound then gives that the expected size of of $N_3$ is bounded below by
\begin{align*}
(1 - 2(1 - \gamma)) > 0.9n.
\end{align*}
Finally we let $M$ be the set of queries for $G$, then define
\begin{align*}
M' := \{(q^{(1)},q^{(2)},q^{(3)}) \in M : q^{(3)} \in N_3 \}.
\end{align*}
If $N_3$ were independent of $M'$, we could conclude
\begin{align}
\E{\abs{M'}} = m \frac{\abs{N_3}}{n} > 3.01n
\end{align}
and then, by concentration, 
\begin{align}
\p{\abs{M'} < 3.009} \lesssim \exp(-n) = o(1). \label{eqn: M' is large whp}
\end{align}
$M$ and $N_3$ are not independent, but a technical lemma (Lemma~\ref{lem: M' is positively correlated}) shows their correlation can only increase the size of $M'$, hence (\ref{eqn: M' is large whp}) remains valid.

Now consider a game $G'$ consisting of only the clauses of $G$ with queries in $M'$. $M'$ has been constructed such that $G'$ has shuffle gadgets for any wire of a pair of queries drawn from $M'$. Furthermore $|M'| - 3n \geq 0.009 n$ with high probability, so by Lemma \ref{lemma: sufficient conditions for refutation} and Lemma \ref{lem:rank nullity PESc}, we can then conclude $G'$ contains a complete refutation with probability $1 - o(1)$. Since $G'$ contains a subset of the clauses of $G$, this also means $G$ contains a complete refutation with probability $1 - o(1)$. 
\end{proof}

\begin{lem}
 \label{lem: M' is positively correlated} Let $G$ be a random 3-XOR game on $m$ clauses, and let $N_3$ and $M'$ be defined as in the proof of Theorem~\ref{thm: k-XOR non-sym unsat}. If there exists some constant $\delta$ for which
\begin{align*}
\E{\abs{N_3}} \geq \delta n
\end{align*}
with probability $1 - o(1)$, then we have, for any $\epsilon > 0$ that 
\begin{align*}
\E{\abs{M'}} \geq (\delta - \epsilon)m 
\end{align*}
with probability $1-o(1)$ as well. 
\end{lem}

\begin{proof}
We first move from the random game $G$ to the random game $\hat{G}$, in which the total number of clauses isn't fixed, but rather every possible clause appears in the game with probability\footnote{Note the factor of 2 in the denominator comes from the choice of parity bit.}\footnote{Here and below we use $\epsilon_i$ to indicate arbitrary small constants.}
\begin{align}
\frac{m - \epsilon_1}{2n^3}.
\end{align}
We also define the variables $\hat{N}_3$, $\hat{M}$ and $\hat{M}'$, which depend on $\hat{G}$ in exactly the same way the unhatted variables depends on $G$. By an argument identical to the one used in the proof of Lemma \ref{lem: Random graphs are equivalent}, lower bounds on the size of $\hat{M}'$ will carry over to lower bounds on the size of $M'$ for $G$ with high probability. 

The techniques used to bound the size of $N_3$ work equally well on $\hat{N}_3$, and so 
\begin{align}
|\hat{N}_3| \geq (\delta - \epsilon_1 - \epsilon_2) n
\end{align}
with probability $1 - o(1)$. 

Now we let $A$ be some arbitrary subset of $[n]$ of size $\floor{(\delta - \epsilon_1 - \epsilon_2)n}$, and define 
\begin{align*}
\hat{M}(A) = \{(q^{(1)},q^{(2)},q^{(3)}) \in \hat{M} : q^{(3)} \in A \}.
\end{align*}
Since $A$ is arbitrary, it is immediate that 
\begin{align}
\E{|\hat{M}(A)|} =  (\delta - \epsilon_1 - \epsilon_2) m
\end{align}
and (by concentration)
\begin{align}
\p{|\hat{M}(A)| < (\delta - \epsilon_1 - \epsilon_2 - \epsilon_3) m} = o(1). 
\end{align}
Finally, we define the indicator random variables $I_q$ to take on value $1$ if $q \in \hat{M}$, and $0$ otherwise. Our key observation is that 
\begin{align}
\E{I_q \mid A \subseteq N_3} &= \E{I_q} \left(\frac{\p{A \subseteq N_3 \mid I_q = 1}}{\p{A \subseteq N_3}}\right) \\
&\geq \E{I_q}
\end{align}
and this remains true even after conditioning on the outcomes of other $I_q$'s. 

The indicator for the event 
\begin{align}
\curly{|\hat{M}(A)| < (\delta - \epsilon_1 - \epsilon_2 - \epsilon_3) m}
\end{align}
is a decreasing function of the $I_q's$, and so we can conclude
\begin{align}
\p{|\hat{M}(A)| \leq (\delta - \epsilon_1 - \epsilon_2 - \epsilon_3) m \mid A \subseteq N_3} \leq
\p{|\hat{M}(A)| < (\delta - \epsilon_1 - \epsilon_2 - \epsilon_3) m}
\end{align}
Putting this all together, we find
\begin{align}
\p{|\hat{M}'| \leq (\delta - \epsilon_1 - \epsilon_2 - \epsilon_3) m} &\leq \p{|\hat{M}(A)| < (\delta - \epsilon_1 - \epsilon_2 - \epsilon_3) m \mid A \subseteq N_3} \\
&\leq \p{|\hat{M}(A)| < (\delta - \epsilon_1 - \epsilon_2 - \epsilon_3) m} \\
&= o(1)
\end{align}
(where the first line follows from definition of $M'$). Since $|\hat{M}'|$ is a lower bound for $|M'|$ with high probability, we can set $\epsilon = \epsilon_1 + \epsilon_2 + \epsilon_3$ and conclude the result.

\end{proof}
\subsection{Lower Bound on Refutation Length (Sketch)}
\label{subsec: lower bounds}
\label{subsection: examples (Lower Bounds)}
In this section we sketch the proof of the following theorem, which gives a lower bound that holds with high probability for the length of refutations of random 3-XOR games. Aside from the immediate implications of the theorem, this result is also significant because its proof uses a counting technique not found elsewhere in the paper.

\begin{repthm}{thm:sos-LB}
\ThmSoSLB
\end{repthm}
This significance of this result is twofold. Firstly, it gives a lower bound on refutation lengths which matches the length of refutations constructed using the methods of Section \ref{section: decidability algorithm} to a factor of $O(\log(\log(n)))$. This suggests that the algorithm described in Section \ref{section: decidability algorithm} is a near-optimal method for constructing refutations for symmetric XOR games.\footnote{Strictly speaking, this conclusion is motivated only for 3-XOR games. That being said, for larger $k$, Theorem \ref{thm:sos-LB} still gives a lower bound which is tight to a factor of $C_k \log(\log(n))$, and it is reasonable to expect that, with additional work, this lower bound could be tightened further.} Secondly, combining Theorem \ref{thm:sos-LB} with Lemma \ref{lem:no refutation of length l implies degree l pseudodistribution} show that an ncSoS proof that a random 3-XOR game has value $< 1$ requires going to level $\Omega(n \log(n) / \log(\log(n)))$ in the ncSoS hierarchy. This results in a runtime which is superexponential in $n$, and longer than the worst possible case for classical (commuting) SoS applied to XOR games (or boolean CSPs in general).

Theorem \ref{thm:sos-LB} is proved using a careful application of the first moment method. The full analysis is somewhat involved, and so we spend some time discussing the key ideas required for the proof. The proof hinges on enumerating possible refutations in a somewhat non-intuitive way. Rather than building up a refutation of length $\ell$ query by query, we will instead write down all possible sequences of $\ell$ queries, and consider all the ways those queries could cancel to form a refutation. The key definition required to make this counting work is that of a \emph{cancellation pattern}.

\begin{defn}
A length $\ell$ \textbf{one wire cancellation pattern} is a partition of $[\ell]$ into $\ell/2$ pairs of the form $\{(a_1,b_1), \ldots  , (a_{\ell/2}, b_{\ell/2}) \}$ with 
\begin{align}
a_i < b_i \text{ and } a_i < a_j \implies b_i > b_j 
\end{align}
$\all i,j \in [\ell/2]$ (no cancellation patterns exist for odd $\ell$). When discussing $k$-XOR games, a length $\ell$ \textbf{cancellation pattern} refers to an ordered list containing $k$ one wire cancellation patterns.
\end{defn}
\begin{defn}
Given a length $\ell$ cancellation pattern, the \textbf{locations} of that cancellation pattern are elements of $[\ell]$, corresponding to the positions at which queries can appear in the cancellation. The \textbf{sites} of the cancellation pattern are specified by coordinates $(\alpha, i) \in [k] \otimes [\ell]$, and represent the places where individual questions appear. Site $(\alpha_1, i_1)$ is said to \textbf{cancel} site $(\alpha_2, i_2)$ iff $\alpha_1 = \alpha_2$ and the pair $(i_1,i_2)$ is contained in the $\alpha_1$-th cancellation pattern. In this case, the pair of sites $((\alpha_1, i_1), (\alpha_2, i_2))$ is referred to as a \textbf{cancellation}.
\end{defn}
\begin{defn}
Using matrix notation to specify individual letters in a word, a cancellation pattern is \textbf{valid} on a word $W$ iff 
\begin{align}
w_{\alpha_1, i_1} = w_{\alpha_2, i_2} 
\end{align}
for all sites $(\alpha_1, i_1)$ and $(\alpha_2, i_2)$ which cancel one another.
\end{defn}

By definition, a word cancels to the identity iff there exists at least one cancellation pattern which is valid on the word. It is also straightforward to give a combinatorial bound on the number of possible length $\ell$ cancellation patterns.

\begin{claim}
The number of possible cancellation patterns on a single wire with $\ell$ locations is given by the $\ell/2$-th Catalan number, denoted by $\cC_{\ell/2}$.  The number of possible cancellation patterns on a length $\ell$ word formed from $k$-XOR queries is then given by 
\begin{align}
\left(\cC_{\ell/2}\right)^k \leq 2^{k\ell}.
\end{align}
\end{claim}

\begin{proof}
Direct from the definition of Catalan numbers, and standard bounds on their size. See \cite{Stanley-Catalan} for an extensive discussion.
\end{proof}

To illustrate the benefit of working in terms of cancellation patterns, we prove a simple theorem, regarding the existence of a restricted class of refutations.

\begin{thm} \label{theorem: necessary conditions for refutations which don't repeat clauses}
Let $m \in o(n^{k/2})$. Then, as $n \rightarrow \infty$,  a random $k$-XOR game with $m$ queries on an alphabet of size $n$ will contain a refutation in which every query is used at most once with probability at most $o(1)$.
\end{thm}

\begin{proof}
We apply the first moment method. There are $\ell! \binom{m}{\ell} $ ways of creating a word of length $\ell$ from the queries, and at most $2^{k\ell}$ cancellation patterns on the word. Since queries are all independent and randomly chosen, each cancellation pattern on a length $\ell$ word is valid with probability $\left(1/n\right)^{k \ell/2}$. Then the probability of a valid cancellation of any length is given by (using $(m)(m-1)...(m-2r) \leq m^{2r}$)
\begin{align}
\sum_{r=1}^{m/2} \left[(2r)!\binom{m}{2r}\left(\cC_r\right)^k (1/n)^{kr} \right] \leq \sum_{r=1}^{m/2} \left[2^k \left(m/n^{k/2}\right) \right]^{2r} \in o(1).
\end{align}
\end{proof}
We now take a small detour, and use techniques similar to the one above to reprove a result of Grigoriev~\cite{Gri01}. This is done to illustrate the power of these techniques, but also for completeness, as we will use Grigoriev's result in our proof of \thmref{sos-LB}. 
\begin{thm}[Originally proved in \cite{Gri01}] \label{theorem: Grigoriev}
Let $G$ be a random 3-XOR game on the set of queries $M$, with $|M| = m = Cn$ and alphabet size $n$. Define a classical refutation to be a subset of queries $T \subseteq M$ such that 
\begin{align}
| \{q \in T \mid q^{(\alpha)} = j \} | = 2m \all j \in [n], \alpha \in \{1,2,3\}
\end{align}
(if written as a word, $T$ would contain each $j \in [n]$ an even number of times on each wire). Then, with probability $1 - o(1)$ as $n \rightarrow \infty$ the shortest classical refutation contained in $m$ has length at least $en/C^2$. 
\end{thm}

\begin{proof}
We again use the first moment method, paralleling the argument used in the proof of Theorem \ref{theorem: necessary conditions for refutations which don't repeat clauses}. We find $\binom{Cn}{\ell}$ ways of choosing $\ell$ queries from $M$, and $((\ell-1)!!)^3$ ways of pairing up letters on all rows once $\ell$ queries have been chosen (if $\ell$ is even). As before, each pair of letters is equivalent independently with probability $(1/n)$ and so by the union bound the probability of a classical refutation of length less than $\ell$ is bounded by
\begin{align}
\sum_{r = 1}^{\ell/2} \left[\binom{Cn}{2r}((2r-1)!!)^3(1/n)^{3r} \right] &\leq \sum_{r = 1}^{\ell/2} \left[ (Cn)^{2r} \left(2^r r!\right) (1/n)^{3r}\right] \\
&= \sum_{r = 1}^{\ell/2} \left[ r! \left( 2C^2/n \right)^r \right] \\
&\leq \sum_{r = 1}^{\ell/2} \left[ e\sqrt{r} \left(2C^2r/(en) \right)^r \right].
\end{align}
Noting this sum is $o(1)$ provided $\ell C^2/en < 1$ completes the proof. 
\end{proof}

Returning to the informal proof of \thmref{sos-LB}, the natural approach is to try to generalize the proof of Theorem \ref{theorem: necessary conditions for refutations which don't repeat clauses} by allowing repeated queries and repeating the union bound analysis. Unfortunately, when queries are repeated not all cancellations are valid independent of one another, which makes it dramatically more difficult to compute the probability of a given cancellation pattern being valid. To accommodate this, we require additional terminology for discussing the different types of cancellations that can occur when a cancellation interacts with a word containing repeated queries. This is introduced below, along with a brief discussion of how these cancellations are accounted for in the full proof. 

\begin{defn}
Given a cancellation pattern on a word made up of queries from a random $k$-XOR game, define:
\begin{itemize}
\item The set of \textbf{independent cancellations} to be a maximal set of cancellations so that each cancellation is valid independent of all others in the set with probability $1/n$. 
\item The set of \textbf{dependent cancellations} to be the set of cancellations which are valid with probability 1 if all independent cancellations are valid. 
\item The set of \textbf{self cancellations} to be the set of all cancellations which are valid with probability 1 independent of all other cancellations (these occur when a query is canceled with itself).
\end{itemize}
A \textbf{full cancellation pattern} is a cancellation pattern where cancellations are specified to be independent, dependent or self ahead of time, and this full cancellation pattern is valid on a word iff the sets defined above are compatible with the way the cancellations are labeled ahead of time.
\end{defn}

Note there is some freedom in which cancellations are chosen as dependent vs.~independent. This ambiguity allows us to simplify the full proof, and  is left in intentionally.\footnote{Of course, it could also be removed by fixing a convention for the cancellations which are labeled independent (i.e.~choosing the lexicographically minimal set).}
Semi-formally, we can now give the proof of \thmref{sos-LB} as follows:

\begin{proof}[Proof (semi-formal)]
\renewcommand\qedsymbol{$\approx \! \square$} Our goal is to show that, under the conditions of \thmref{sos-LB}, any cancellation pattern on a word consisting of a small number of queries is valid with vanishing probability. We restrict our attention to minimum length refutations: refutations for which no subset of queries can be removed while leaving a valid refutation. 

We then attempt a union bound argument in which we identify the various ways queries can interact with cancellation patterns in the refutation.
We begin by segmenting the queries in the refutation into maximal strings of queries connected via dependent or self cancellations. We call these phrases. By definition, the phrases themselves must be connected by
independent cancellations.

We can bound the number of ways of building a
phrase of length $k$. The first query in a phrase can be a picked arbitrarily from a set of size $m$. After that, a query connected to a known query by a self-cancellation is fixed exactly, and concentration inequalities can be used to show that a query connected to a fixed query via a dependent cancellation is drawn from a set of size at most $m\log(n)/n$.\footnote{Proved in Lemma \ref{lem: structure of refutations}} Then the ways of choosing queries such that they
form the given phrase is bounded by
\begin{align}
m \paren{\frac{m\log(n)}{n}}^{k-1}. \label{eqn:phrase-bound}
\end{align}


We next place some restrictions on the number and type of phrases that
can occur in a refutation. By minimality, each phrase must contain at
least one site involved in an independent cancellation (otherwise the
phrase is ``redundant''); then by parity each phrase must contain two.
We also get a bound on the number of queries appearing an odd number
of times. Removing all queries that occur an even number of times,
and leaving only one copy of each query that occurs an odd number
will produce a classical refutation.
Theorem~\ref{theorem: Grigoriev} then tells us that with probability
$1-o(1)$, any valid refutation must have $en/C^2$ queries which
occur an odd number of times. 

We then use a result from the technical
proof: for $p$ phrases and $s$ sites with independent cancellations,
\begin{align}
s \geq 2p  + en/4C^2.
\end{align}
Using (\ref{eqn:phrase-bound}) to bound the
number of ways each phrase occurs,
and a factor of $1/\sqrt{n}$ per site in an independent cancellation
(making $1/n$ per independent cancellation) we find that any full
length-$\ell$ cancellation 
pattern is valid on some word of $\ell$ queries with probability at most
\begin{align}
m^{p} \paren{\frac{m\log(n)}{n}}^{\ell - p} \paren{\frac{1}{n}}^{s/2}
&\leq m^{p} \paren{\frac{m\log(n)}{n}}^{\ell - p}
	\paren{\frac{1}{n}}^{p + en/8C^2} \\
&\leq \paren{\frac{1}{\log(n)}}^{p} \paren{\frac{1}{n}}^{en/8C^2}
	\paren{C\log(n)}^{\ell}\\
&\leq \paren{\frac{1}{n}}^{en/8C^2}
	\paren{C\log(n)}^{\ell}.
\end{align}
Adding in a union bound over all possible length $\ell$ full cancellation patterns, we find the probability of a valid length $\ell$ cancellation pattern existing is at most
\begin{align}
\cC_{\ell/2}^3 3^{3\ell/2} \paren{\frac{1}{n}}^{en/8C^2}
\paren{C\log(n)}^{\ell}
&\leq 12^{3\ell/2} \paren{\frac{1}{n}}^{en/8C^2}
\paren{C\log(n)}^{\ell} \\
&\leq \paren{\frac{1}{n}}^{en/8C^2}
\paren{42 C \log(n)}^{\ell}.
\end{align}
Setting $m/n = C$ and following through the geometric series we find the probability of a refutation of length less than or equal to $\ell$ existing is at most
\begin{align}
\frac{42C\log(n)^{\ell+1}}{n^{en/8C^2}} + o(1)
\end{align}
where the $o(1)$ term comes from the use of results \ref{theorem: Grigoriev} and \ref{lem: structure of refutations} in our proof. If follows that the total probability of refutation is $o(1)$ unless 
\begin{align}
\ell \geq \frac{en\log(n)}{8C^2 \log(42C \log(n))} - 1 
\end{align}
completing the proof of \thmref{sos-LB}.
\end{proof}

While the proof above was hopefully convincing, it wasn't completely formal. A more careful proof that clearly discusses the various events the union bound is constructed over is given below. 

\subsection{Lower Bound on Refutation Length (Full Proof)}

For the most part, the key ideas used in the proof of \thmref{sos-LB} are well covered in Section~\ref{subsection: examples (Lower Bounds)}. The remaining details are primarily technical, but somewhat involved. We begin by formalizing the definition of a \textit{phrase}, used informally above.  
\begin{defn}
Consider a full cancellation pattern consisting of dependent, self and independent cancellations. Let $G$ be a graph with vertices corresponding to dependent or self cancellations in the cancellation pattern. Add an edge between vertices if the corresponding cancellations overlap at some location. The sets of cancellations corresponding to connected components in this graph are called \textbf{phrases}. 
\end{defn}

Our analysis will require language specific to the ways in which queries and phrases can occur in a refutation. That language is introduced below. 
\begin{defn}
Given a refutation, we define the following sets:  
\begin{itemize}
\item $L_r$ is the set of locations located at the leftmost point in some phrase. We call queries at these locations \textbf{roots}.
\item $L_c$ is the set of all locations in phrases which are not the leftmost point of a phrase. Queries at these locations are called \textbf{constrained queries}.
\item $P$ is the set of all phrases in the cancellation pattern.
\item $P_r \subseteq P$ is the set of all phrases for which every location in the phrase contains only self or dependent cancellations. Phrases in $P_r$ are called \textbf{redundant phrases}.
\item $S$ is the set of all sites in independent cancellations.
\end{itemize}
Redundant phrases are so named because removing all queries contained in them still leaves a valid refutation. For this reason \textbf{minimal length refutations} are defined to be refutations that do not contain any redundant phrases. 
\end{defn}

We now prove a few basic properties about the structure of refutations constructed from random queries.

\begin{lem} \label{lem: structure of refutations} Let $m = Cn$ for some constant $C$. Then, with probability $1-o(1)$, all refutations for a random 3-XOR game on $m$ queries with $n$ variables will satisfy
\begin{enumerate}
\item The refutation contains at least $en/C^2$ distinct queries occurring an odd number of times. \label{item: refutations have many distinct clauses}
\item The cancellations can be labeled so that $|S| \geq 2|P| + en/4C^2$. \label{item: refutations have many independent cancellations}
\item For all queries $q_i$: the refutation implies that $q_i$ cancels with at most $C \log(n)$ other queries on each wire. \label{item: no clause is repeated too frequently in refutations}
\end{enumerate}
\end{lem}

\begin{proof}
We prove \ref{item: refutations have many distinct clauses} by appealing to \cite{Gri01}. Note we can obtain a classical refutation from a quantum refutation by taking a single copy of each query repeated an odd number of times. Then, we know from \cite{Gri01} (alternately
Theorem~\ref{theorem: Grigoriev}) that there are $en/C^2$ distinct queries repeated an odd number of times in the quantum refutation.

To prove \ref{item: refutations have many independent cancellations},
we first note that every distinct query occurring an odd number of
times must be involved in at least three independent cancellations (one per wire) across all
locations where it appears, resulting in a total count of $3en/C^2$ cancellations. We will show that we can relabel independent and dependent cancellations such that at least $1/4$ of these independent cancellations are all contained in at most $en/4C^2$ phrases. 

To do so, we begin by making a list of all queries occurring an odd number of times in our refutation, and consider a cancellation pattern on which only the self-cancellations have been fixed. We refer to a phrase induced by these self cancellations as a \emph{subphrase}. We now extract a query from the list, pick an odd length subphrase involving that query (this subphrase may have length one), and mark three non-self cancellations coming from that subphrase as independent (one per wire). Next, we remove from our list any queries connected to this subphrase by the newly labeled independent cancellations. 
Removing the connected queries from the list ensures that any non-self cancellations involving elements remaining on our list will be independent from the cancellations we have labeled so far. 
We then repeat this process until we have exhausted all queries on our list, and then label the remaining cancellations in any valid manner. 

Over this process, we remove at most three additional queries from the list for every subphrase we identify, so when we have exhausted all queries on this list (but before we have labeled any dependent cancellations), we will have at $en/4C^2$ subphrases containing at least $3en/4C^2$ independent cancellations. Each of these subphrases is contained in a phrase (since all locations are connected via self-cancellations) and labeling the remaining cancellations cannot change the cancellations already labeled as independent, so we have found at least $3en/4C^2$ independent cancellations contained in at most $en/4C^2$ phrases.

From here the proof of \ref{item: refutations have many independent cancellations} is straightforward: by minimality, each phrase contains must contain at least one independent cancellation, and hence by parity each phrase must contain two. Furthermore, we have already identified a special set of at most $en/4C^2$ phrases which contain at least $3en/4C^2$ independent cancellations. Letting $p_1$ be the number of phrases identified so far, and $p_2$ be the number of phrases not contained in the set already identified, we see
\begin{align}
|S| \geq 2p_2 + \frac{3en}{4C^2} \geq 2(p_2 + p_1) + \frac{en}{4C^2} = 2|P|  + \frac{en}{4C^2}
\end{align}
as desired.

Finally, \ref{item: no clause is repeated too frequently in refutations} follows from concentration of measure. We define $y(j)$ to be the random variable counting the number of queries with letter $j$ on the top wire, so
\begin{align}
y(j) = \abs{\{i: q_i^{(1)} = j \}}.
\end{align}
It is then clear that 
\begin{align}
\E{y(j)} = m/n = C.
\end{align}
By a Chernoff bound 
\begin{align}
\p{y(j) \geq C \log(n)} &\leq e^{-(\log^2(n)-1)C/3n} \\
&\leq 1/n^{C \log(n)/3} = o(1)
\end{align}
and so a union bound argument gives the result for large $n$. 
\end{proof}
The proof of \thmref{sos-LB} will follow from our observations in Lemma \ref{lem: structure of refutations} and first moment arguments. To make clear the analysis, we first present a simple algorithm for generating minimal length refutations with length $\ell$ from a random set of queries $G$. 
\begin{alg}[Refutation generator] \label{alg: generating possible 3-XOR refutations} \item 
\emph{\textbf{input:}} A set of $m$ queries, with $m = C n$, and parameter $\ell$ \\
\emph{\textbf{output:}} A minimal refutation of length $\ell$, or \textit{failure}
\begin{enumerate}
\item Initialize $\ell$ locations where queries might be placed. 
\item Randomly generate a cancellation pattern consisting of self, dependent and independent cancellations on the $\ell$ locations. Identify the phrases in this cancellation pattern.
\item If there is any redundant phrase, return \emph{failure: not minimal}.
\item Randomly map queries to locations. \label{item: clause mapping step of algorithm}
\begin{enumerate}
\item If the independent cancellations require there to be more than $C \log(n) $ queries which agree on any wire, or if the cancellation pattern would imply $|S| \leq 2|P| + en/4C^2$, return \textit{failure: improbable cancellation}. \label{item: not too many choices for copied cancellation}
\item If self-cancellations occur between non-identical queries, or dependent cancellations are not implied by independent cancellations, return \textit{failure: improper labeling}.
\end{enumerate}
\item If any independent cancellations occur between queries which disagree on the wire where the cancellation is occurring, return \emph{failure: invalid cancellation}. \label{item: check independent cancellations}
\item Otherwise, return \textit{success} along with the cancellation pattern and query mapping. 
\end{enumerate}
\end{alg}
We prove \thmref{sos-LB} by proving two basic facts about Algorithm \ref{alg: generating possible 3-XOR refutations}. Firstly, we show that, with high probability\footnote{It should be stressed that this with high probability refers to the randomness associated with choosing the queries provided as input to the algorithm, not the randomness associated with the algorithm's run.}, there exists a random seed for which Algorithm~\ref{alg: generating possible 3-XOR refutations} finds a refutation provided one exists. Secondly, we show the expected number of paths on which Algorithm \ref{alg: generating possible 3-XOR refutations} returns success is small unless $\ell$ is sufficiently large. We will prove these claims separately. 
\begin{thm}[Correctness] \label{thm: correctness of refutation generating algorithm} 
Algorithm \ref{alg: generating possible 3-XOR refutations} only returns success when it finds a valid minimum length refutation. Furthermore, when the input queries are randomly selected, with probability $1-o(1)$ the algorithm has a positive probability of finding all valid length $\ell$ refutations. 
\end{thm}
\begin{proof}
The first claim is clear from inspection of the algorithm. The second follows from Lemma~\ref{lem: structure of refutations} and further inspection. In particular, the only refutations not found by the algorithm are those which require greater than $C \log(n)$ queries to agree on a wire, or those with a cancellation pattern for which 
\begin{align}
|S| \leq 2|P| + en/4C^2.
\end{align}
Lemma~\ref{lem: structure of refutations} tells us that these cases occur with probability $o(1)$ for randomly chosen queries.
\end{proof}
\begin{thm} \label{thm: bound on expected number of solutions to refutation generating algorithm}
Given a randomly chosen set of $m = C n$ queries as input, the expected number of minimal length $\ell$ refutations which can be found by Algorithm~\ref{alg: generating possible 3-XOR refutations} is upper bounded by
\begin{align}
(1/n)^{e/2C^2}(42C \log(n))^\ell .
\end{align}
In particular, we expect to find no refutations until 
\begin{align}
\ell = \Omega(n\log(n)/\log(\log(n)))
\end{align}
\end{thm}
\begin{proof}
We give an overcounting of the number of possible paths Algorithm \ref{alg: generating possible 3-XOR refutations} can take. We first note that a path can be completely specified by a choice of cancellation pattern and mapping of queries to locations.

Using our rough bound on the Catalan numbers,
there are at most $\cC_{\ell/2}^{3} \leq 4^{3\ell/2}$ different
ways of pairing up all sites for cancellations. Since each
cancellation can be one of three types, we find a total of 
\begin{align}
3^{3\ell/2} 4^{3\ell/2} \leq 42^{\ell}
\end{align}
possible cancellation patterns.

We next give a rough (over)counting of the number of ways queries can be mapped to locations such that the cancellation pattern is not rejected in step \ref{item: clause mapping step of algorithm} of the algorithm. 

In particular, we allow arbitrary queries to be mapped to locations in $L_r$. After this mapping, we note all remaining locations are in $L_c$. Assuming the cancellation pattern was not rejected in step \ref{item: not too many choices for copied cancellation}, a location connected to a fixed query by a self cancellation can only have a single query mapped to it, and a location connected to a fixed query by a dependent cancellation can have at most $C\log(n)$ queries mapped. In total then, we find 
\begin{align}
m^{|L_r|}\left(C\log(n)\right)^{|L_c|} = m^{|P|}\left(C\log(n)\right)^{|L_c|}
\end{align}
possible mappings from queries to locations. 

Finally, we bound the probability that our given query assignment doesn't fail in step \ref{item: check independent cancellations} of the algorithm. Noting independent cancellations are, by definition, independent we find the probability of failure is given by
\begin{align}
\left(\frac{1}{n}\right)^{|S|/2}.
\end{align}
Since our cancellation pattern doesn't contain any redundant phrases, and was not rejected as improbable by the algorithm we also have 
\begin{align}
|S| \geq 2|P| + en/4C^2.
\end{align}
The overall expected number of successes for a given cancellation pattern can then be bounded by:
\begin{align}
m^{|P|}(C\log(n))^{|L_c|} \paren{\frac{1}{n}}^{|P| + en/8C^2}
&= m^{|P|}\paren{\frac{m\log(n)}{n}}^{\ell - |P|}
	\paren{\frac{1}{n}}^{|P| + en/8C^2}\\
&\leq \paren{\frac{1}{\log(n)}}^{|P|}
	\paren{\frac{1}{n}}^{en/8C^2} \paren{C \log(n)}^{\ell}\\
&\leq \paren{\frac{1}{n}}^{en/8C^2} (C \log(n))^\ell
\end{align}
resulting in an overall bound on the expected number of successes for any length $\ell$ of 
\begin{align}
\left(\frac{1}{n} \right)^{en/8C^2} (42C \log(n))^\ell.
\end{align}
Summing the geometric series, the expected total number of refutations of length less than $\ell$ can then be bounded above by 
\begin{align}
\left(\frac{1}{n} \right)^{en/8C^2} \frac{(42C \log(n))^{\ell+1}-1}{(42C \log(n)) - 1}.
\end{align}
We see this is $o(1)$\footnote{\text{As a word of caution: it should be noted the $\left(e n \log(n)\right)/\left(8C^2\log(\log(n))\right)$} only dominates when $C$ is taken to be a constant with respect to $n$. When $C$ scales with $n$ the above analysis will still work, but requires more care in computing the final bound.}  unless 
\begin{align}
\ell \geq \frac{e n \log(n)}{8C^2\log(\log(n))} - o\left(\frac{n \log(n)}{\log\log(n)}\right),
\end{align} 
and the desired result follows from Markov's inequality.
\end{proof}
To close this section, we note Theorem~\ref{thm:sos-LB} is immediate from Theorems~\ref{thm:
  correctness of refutation generating algorithm} and \ref{thm: bound
  on expected number of solutions to refutation generating algorithm}.

\label{subsec: Random Game Thresholds}

\pagebreak
\appendix
\section{Explicit Refutation for Capped GHZ}
\label{subsec: Explicit refn Capped GHZ}
While the results of Section~\ref{section: decidability algorithm} prove
that a refutation must exist for every symmetric game with a \PREF{} specification,
which includes all games in the Capped GHZ family, it is illuminating to follow
a concrete demonstration of the methods of that section. We focus on the
3rd order Capped GHZ game and explicitly use the \PREF{} specification
and shuffling technology to construct a refutation.

The 3rd order Capped GHZ is defined by
\begin{equation}
CG_3 := \left\{ \bbm 1 \\ 1 \\ 1 \\ -1 \ebm,
\bbm 1 \\ 2 \\ 2 \\ +1 \ebm, \bbm 2 \\ 1 \\ 2 \\ +1 \ebm, \bbm 2 \\ 2 \\ 1 \\ +1 \ebm,
\bbm 2 \\ 3 \\ 3 \\ +1 \ebm, \bbm 3 \\ 2 \\ 3 \\ +1 \ebm, \bbm 3 \\ 3 \\ 2 \\ +1 \ebm,
\bbm 3 \\ 3 \\ 3 \\ +1 \ebm  \right\}.
\end{equation}

By the arguments of the previous section, $z = \begin{pmatrix} -1 & 1 & 1 & 1 & -2 & -2 & -2 & 4 \end{pmatrix}^\transp$ is a \PREF{} specification for this game.
This \PREF{} specification allows us to write down a word
$W \ppsim I$ by filling
all even positions with $|z_i|$ copies of queries that have $z_i > 0$
and filling
all odd positions with $|z_i|$ copies of queries that have $z_i < 0$,
\begin{equation}
W := \begin{array}{r c c | c c | c c | c c | c c | c c | c c@{\,}l}
\ldelim[{3}{0pt} & 1 & 1 & 2 & 2 & 3 & 2 & 3 & 3 & 2 & 3 & 3 & 3 & 3 & 3 &
\rdelim]{3}{0pt}
\\
& 1 & 2 & 3 & 1 & 2 & 2 & 3 & 3 & 3 & 3 & 2 & 3 & 3 & 3 &\\
& 1 & 2 & 3 & 2 & 3 & 1 & 2 & 3 & 3 & 3 & 3 & 3 & 2 & 3 & \\
& \multicolumn{2}{c}{(1)} &
\multicolumn{2}{c}{(2)} &
\multicolumn{2}{c}{(3)} &
\multicolumn{2}{c}{(4)} &
\multicolumn{2}{c}{(5)} &
\multicolumn{2}{c}{(6)} &
\multicolumn{2}{c}{(7)} &
\end{array}.
\end{equation}

Now, following the procedure of Section~\ref{subsec: Algorithm},
we determine the parity-preserving
permutation that causes the first two rows to cancel. The first row even letters
can be mapped to canceling odd letters by
the bijection $f_1: \cE \rightarrow \cO$ defined by
\begin{align}
\pi_1 &:= (3 \, 5) \\
f_1(\cE_i) &:= \cO_{\pi_1(i)}
\end{align}
while the second row even letters can likewise be mapped
via the bijection $f_2: \cE \rightarrow \cO$
defined by
\begin{align}
\pi_2 &:= (1 \, 6 \, 2) \\
f_2(\cE_i) &:= \cO_{\pi_2(i)}.
\end{align}
The combined bijection $f: \cE \cup \cO \rightarrow \cE \cup \cO$
given by
\begin{align}
f(\cE_i) &:= \cO_{\pi_1(i)} \\
f(\cO_i) &:= \cE_{\pi^{-1}_2(i)}
\end{align}
then gives rise to a permutation on the queries via the Foata correspondence.
Writing $f$ as a product of cycles gives
\begin{equation}
f = \paren{\cO_1 \cE_2 \cO_2 \cE_6 \cO_6 \cE_1} \paren{\cO_3 \cE_3 \cO_5 \cE_5} \paren{\cO_4 \cE_4} \paren{\cO_7 \cE_7}
\end{equation}
and thus the query permutation written in two-line notation
\begin{equation}
\pi := \begin{pmatrix}
\cO_1 & \cE_1 & \cO_2 & \cE_2 & \cO_3 & \cE_3 & \cO_4 & \cE_4 &
\cO_5 & \cE_5 & \cO_6 & \cE_6 & \cO_7 & \cE_7 \\
\cO_1 & \cE_2 & \cO_2 & \cE_6 & \cO_6 & \cE_1 & \cO_3 & \cE_3 &
\cO_5 & \cE_5 & \cO_4 & \cE_4 & \cO_7 & \cE_7
\end{pmatrix}.
\end{equation}
Applying this permutation to $W$ produces the desired form:
\begin{align}
\pi(W) &= \begin{bmatrix}
1 & 2 & 2 & 3 & 3 & 1 & 3 & 2 & 2 & 3 & 3 & 3 & 3 & 3 \\
1 & 1 & 3 & 3 & 2 & 2 & 2 & 2 & 3 & 3 & 3 & 3 & 3 & 3 \\
1 & 2 & 3 & 3 & 3 & 2 & 3 & 1 & 3 & 3 & 2 & 3 & 2 & 3
\end{bmatrix} \\
&\sim \left[ \begin{array}{c c | c c | c c | c c | c c | c c | c c}
&&&&&&&&&&&&& \\
&&&&&&&&&&&&& \\
1 & 2 & 3 & 3 & 3 & 2 & 3 & 1 & 3 & 3 & 2 & 3 & 2 & 3
\end{array} \right].
\end{align}

We must now use shuffling technology to concatenate to $\pi(W)$
additional queries that rearrange the bottom wire such that it also cancels.
We begin by
identifying the needed shuffle functions, then describe their construction
in terms of shift gadgets. Finally, we expand out the resulting refutation
to demonstrate that the shuffle technology indeed works as expected.

The third wire of $\pi(W)$ may be canceled by the pairwise permutation $\pi' = (1 \, 3 \, 4)$. This can be produced by the composition of two shuffle functions:
\begin{align}
\pi' =& f_b f_a \\
f_a :=& (1 \, 2 \, 4 \, 3) \\
f_b :=& (1 \, 2).
\end{align}
Both shuffle functions can be implemented via a set of shift
gadgets that ``save'' pairs from the 3rd wire onto the 1st and 2nd
followed by a set of shift gadgets that ``load'' those pairs
back onto the 3rd wire shuffled by the relevant shuffle function.
Explicitly, we can perform the first shuffle by enacting the following
saves and loads:
\begin{enumerate}
\item Save $(2\, 3) \rightarrow $ wire $1$.
\item Save $(2\, 3) \rightarrow $ wire $1$.
\item Save $(3\, 3) \rightarrow $ wire $1$.
\item Save $(3\, 1) \rightarrow $ wire $2$.
\item Save $(3\, 2) \rightarrow $ wire $1$.
\item Save $(3\, 3) \rightarrow $ wire $2$.
\item Save $(1\, 2) \rightarrow $ wire $1$.
\item Load from wires $1$ and $2$ in the order $2, 2, 1, 1, 1, 1, 1$.
\end{enumerate}
After the saves, the 1st wire is
$\sim h(23) h(23) h(33) h(32) h(12)$ and the 2nd wire is
$\sim h(31) h(33)$, where $h(ab)$ indicates the ``saved'' form
of the pair $(a\,b)$, while the 3rd wire is $\sim I$. Loading in the
specified order then reinstates the property that the 1st and 2nd wires
cancel to $I$ while the 3rd wire becomes $\sim 3\, 3 \, | \, 3 \, 1 \, | \, 1 \, 2 \, | \, 3 \, 2 \, | \, 3 \, 3 \, | \, 2 \, 3 \, | \, 2 \, 3$. We may then perform
the second shuffle by enacting the following saves and loads:
\begin{enumerate}
\item Save $(2\, 3) \rightarrow $ wire $1$.
\item Save $(2\, 3) \rightarrow $ wire $1$.
\item Save $(3\, 3) \rightarrow $ wire $1$.
\item Save $(3\, 2) \rightarrow $ wire $1$.
\item Save $(1\, 2) \rightarrow $ wire $1$.
\item Save $(3\, 1) \rightarrow $ wire $1$.
\item Save $(3\, 3) \rightarrow $ wire $2$.
\item Load from wires $1$ and $2$ in the order $1, 2, 1, 1, 1, 1, 1$.
\end{enumerate}
After the saves, the 1st wire is
$\sim h(23) h(23) h(33) h(32) h(12) h(31)$ and the 2nd wire is
$\sim h(33)$, while the 3rd wire is $\sim I$. Loading in the specified
order then reinstates the property that the 1st and 2nd wires cancel to $I$
while the 3rd wire becomes $\sim 3\,1\,|\,3\,3\,|\,1\,2\,|\,3\,2\,|\,3\,3\,|\,2\,3\,|\,2\,3 \sim I$.

\newcommand{\saveone}[1]{S^{3 \rightarrow 1}_{(#1)}}
\newcommand{\savetwo}[1]{S^{3 \rightarrow 2}_{(#1)}}
\newcommand{\loadone}[1]{S^{3 \leftarrow 1}_{(#1)}}
\newcommand{\loadtwo}[1]{S^{3 \leftarrow 2}_{(#1)}}
\newcommand{\savedot}[1]{S^{3 \rightarrow \cdot}_{(#1)}}
Combining the construction of these two shuffle functions finally produces
the desired refutation:
\begin{equation}
\begin{aligned}
R := &\pi(W) & (\text{first two wires canceled}) \\
& \saveone{2\,3} \saveone{2\,3} \saveone{3\,3} \savetwo{3\,1}
	\saveone{3\,2} \savetwo{3\,3} \saveone{1\,2} & (\text{saves for }f_a) \\
& \loadtwo{3\,3} \loadtwo{3\,1} \loadone{1\,2} \loadone{3\,2}
	\loadone{3\,3} \loadone{2\,3} \loadone{2\,3} & (\text{loads for }f_a) \\
& \saveone{2\,3} \saveone{2\,3} \saveone{3\,3} \saveone{3\,2}
	\saveone{1\,2} \saveone{3\,1} \savetwo{3\,3} & (\text{saves for }f_b) \\
& \loadone{3\,1} \loadtwo{3\,3} \loadone{1\,2} \loadone{3\,2}
	\loadone{3\,3} \loadone{2\,3} \loadone{2\,3}. & (\text{loads for }f_b)
\end{aligned}
\end{equation}
We simplify by canceling neighboring inverses and using the fact that
the shuffle
gadget $\savedot{3\,3}$ and its inverse act on an identical pair,
and thus can be chosen to be the identity to find
\begin{equation}
\label{eqn:R simplified}
R = \pi(W) \saveone{2\,3} \saveone{2\,3} \savetwo{3\,1} \saveone{3\,2}
\saveone{1\,2} \loadtwo{3\,1} \loadone{1\,2} \loadone{3\,2} \loadone{2\,3} \loadone{2\,3}.
\end{equation}
We next construct the needed shuffle gadgets using the neighboring queries in $W$
giving rise to each pairing $(2\,3)$, $(3\,1)$, etc. For example,
$\savetwo{3\,1}$ can be constructed using permutations of the neighboring
queries
\begin{equation*}
\begin{aligned}
\begin{bmatrix} 3 \\ 2 \\ 3 \end{bmatrix}
\end{aligned}
\quad\text{and}\quad
\begin{aligned}
\begin{bmatrix} 2 \\ 2 \\ 1 \end{bmatrix}
\end{aligned}
\end{equation*}
associated with the pair $(3\,1)$ on the third wire of $\pi(W)$:
\begin{align}
\savetwo{3\,1} :=& \begin{bmatrix}
2 & 2 & 3 & 3 \\
2 & 1 & 3 & 2 \\
1 & 2 & 2 & 3
\end{bmatrix} \\
\sim & \begin{bmatrix} &&&\\ 2 & 1 & 3 & 2 \\ 1&&&3 \end{bmatrix}.
\end{align}
This gadget serves to cancel the pair $(3\,1)$ on the third
wire and ``save'' it in the form $h(31) = (2\,1\,3\,2)$ on the second wire.
Recall that reversing the queries of this gadget gives the ``load''
form $\loadtwo{3\,1}$.
Crucially, these constructions are only possible because neighboring
pairs in $\pi(W)$ are guaranteed to have a cancellation in the second wire,
which, exploiting the symmetric nature of the game, allows
us to simplify the third wire of the shuffle gadget.

Expanding (\ref{eqn:R simplified}) allows us to conclude $R \sim I$ by inspection.
\begin{equation}
\begin{aligned}
R =& \begin{bmatrix}
1 & 2 & 2 & 3 & 3 & 1 & 3 & 2 & 2 & 3 & 3 & 3 & 3 & 3 \\
1 & 1 & 3 & 3 & 2 & 2 & 2 & 2 & 3 & 3 & 3 & 3 & 3 & 3 \\
1 & 2 & 3 & 3 & 3 & 2 & 3 & 1 & 3 & 3 & 2 & 3 & 2 & 3
\end{bmatrix} \\
&
\begin{bmatrix}3&3&2&3\\3&3&3&3\\3&3&3&2\end{bmatrix}
\begin{bmatrix}3&3&2&3\\3&3&3&3\\3&3&3&2\end{bmatrix}
\begin{bmatrix}2&2&3&3\\2&1&3&2\\1&2&2&3\end{bmatrix}
\begin{bmatrix}2&2&3&2\\1&1&3&3\\2&2&2&3\end{bmatrix}
\begin{bmatrix}1&2&1&1\\2&2&1&1\\2&1&1&1\end{bmatrix} \\
&
\begin{bmatrix}3&3&2&2\\2&3&1&2\\3&2&2&1\end{bmatrix}
\begin{bmatrix}1&1&2&1\\1&1&2&2\\1&1&1&2\end{bmatrix}
\begin{bmatrix}2&3&2&2\\3&3&1&1\\3&2&2&2\end{bmatrix}
\begin{bmatrix}3&2&3&3\\3&3&3&3\\2&3&3&3\end{bmatrix}
\begin{bmatrix}3&2&3&3\\3&3&3&3\\2&3&3&3\end{bmatrix} \sim I.
\end{aligned}
\end{equation}
The only query with parity bit $s_i = -1$ is
$x_0 = \begin{bmatrix}1&1&1\end{bmatrix}^\transp$, the all-zeros query.
This query appears an odd number of times in $W$ because $W$ was
constructed using the \PREF{} specification $z$. By definition,
the shuffle gadgets have parity bits that multiply to $1$. Thus,
by construction, the associated product of parity bits for
refutation $R$ is $-1$, and therefore $\omega^*(CG_3) < 1$.


\begin{thebibliography}{10}

\bibitem{BHK-QKD05}
J.~Barrett, L.~Hardy, and A.~Kent.
\newblock No signaling and quantum key distribution.
\newblock {\em Physical review letters}, 95(1):010503, 2005,
  \href{http://arxiv.org/abs/quant-ph/0405101}{{\ttfamily
  arXiv:quant-ph/0405101}}.

\bibitem{bayer1992trailing}
D.~Bayer and P.~Diaconis.
\newblock Trailing the dovetail shuffle to its lair.
\newblock {\em The Annals of Applied Probability}, 2(2):294--313, 1992.

\bibitem{BBLV09}
J.~Bri\"{e}t, H.~Buhrman, T.~Lee, and T.~Vidick.
\newblock Multipartite entanglement in {XOR} games.
\newblock {\em Quantum Info. Comput.}, 13(3-4):334--360, Mar. 2013,
  \href{http://arxiv.org/abs/0911.4007}{{\ttfamily arXiv:0911.4007}}.

\bibitem{BrietV13}
J.~Bri{\"e}t and T.~Vidick.
\newblock Explicit lower and upper bounds on the entangled value of multiplayer
  {XOR} games.
\newblock {\em Communications in Mathematical Physics}, 321(1):181--207, 2013,
  \href{http://arxiv.org/abs/1108.5647}{{\ttfamily arXiv:1108.5647}}.

\bibitem{cirel1980quantum}
B.~S. Cirel'son.
\newblock Quantum generalizations of {B}ell's inequality.
\newblock {\em Letters in Mathematical Physics}, 4(2):93--100, 1980.

\bibitem{chsh1969}
J.~F. Clauser, M.~A. Horne, A.~Shimony, and R.~A. Holt.
\newblock Proposed experiment to test local hidden-variable theories.
\newblock {\em Phys. Rev. Lett.}, 23:880--884, Oct 1969.

\bibitem{CHTW04}
R.~Cleve, P.~Hoyer, B.~Toner, and J.~Watrous.
\newblock Consequences and limits of nonlocal strategies.
\newblock In {\em CCC '04}, pages 236--249, 2004,
  \href{http://arxiv.org/abs/quant-ph/0404076}{{\ttfamily
  arXiv:quant-ph/0404076}}.

\bibitem{cleve2014characterization}
R.~Cleve and R.~Mittal.
\newblock Characterization of binary constraint system games.
\newblock In {\em International Colloquium on Automata, Languages, and
  Programming}, pages 320--331. Springer, 2014.

\bibitem{Col06}
R.~Colbeck.
\newblock {\em Quantum And Relativistic Protocols For Secure Multi-Party
  Computation}.
\newblock PhD thesis, University of Cambridge, 2006,
  \href{http://arxiv.org/abs/0911.3814}{{\ttfamily arXiv:0911.3814}}.

\bibitem{DLTW08}
A.~C. Doherty, Y.-C. Liang, B.~Toner, and S.~Wehner.
\newblock The quantum moment problem and bounds on entangled multi-prover
  games.
\newblock In {\em CCC '08}, pages 199--210, 2008,
  \href{http://arxiv.org/abs/0803.4373}{{\ttfamily arXiv:0803.4373}}.

\bibitem{DM02}
O.~Dubois and J.~Mandler.
\newblock The 3-{XORSAT} threshold.
\newblock {\em Comptes Rendus Math{\'e}matique}, 335(11):963--966, 2002.

\bibitem{E91}
A.~K. Ekert.
\newblock Quantum cryptography based on {B}ell's theorem.
\newblock {\em Physical review letters}, 67(6):661, 1991.

\bibitem{FritzNT14}
T.~Fritz, T.~Netzer, and A.~Thom.
\newblock Can you compute the operator norm?
\newblock {\em Proceedings of the American Mathematical Society},
  142(12):4265--4276, 2014,  \href{http://arxiv.org/abs/1207.0975}{{\ttfamily
  arXiv:1207.0975}}.

\bibitem{Gao15}
J.~Gao.
\newblock Quantum union bounds for sequential projective measurements.
\newblock {\em Physical Review A}, 92(5):052331, 2015,
  \href{http://arxiv.org/abs/1410.5688}{{\ttfamily arXiv:1410.5688}}.

\bibitem{greenberger1990bell}
D.~M. Greenberger, M.~A. Horne, A.~Shimony, and A.~Zeilinger.
\newblock Bell's theorem without inequalities.
\newblock {\em American Journal of Physics}, 58(12):1131--1143, 1990.

\bibitem{Gri01}
D.~Grigoriev.
\newblock Linear lower bound on degrees of {P}ositivstellensatz calculus proofs
  for the parity.
\newblock {\em Theoretical Computer Science}, 259:613--622, 2001.

\bibitem{HNW16}
A.~Harrow, A.~Natarajan, and X.~Wu.
\newblock Limitations of semidefinite programs for separable states and
  entangled games.
\newblock 2016,  \href{http://arxiv.org/abs/1612.09306}{{\ttfamily
  arXiv:1612.09306}}.

\bibitem{haastad2001some}
J.~H{\aa}stad.
\newblock Some optimal inapproximability results.
\newblock {\em Journal of the ACM (JACM)}, 48(4):798--859, 2001.

\bibitem{ji2015classical}
Z.~Ji.
\newblock Classical verification of quantum proofs.
\newblock In {\em Proceedings of the 48th Annual ACM SIGACT Symposium on Theory
  of Computing}, pages 885--898. ACM, 2016.

\bibitem{johansson2012giant}
T.~Johansson.
\newblock The giant component of the random bipartite graph.
\newblock Master's thesis, Chalmers University of Technology, 2012.

\bibitem{KRT10}
J.~Kempe, O.~Regev, and B.~Toner.
\newblock Unique games with entangled provers are easy.
\newblock {\em SIAM Journal on Computing}, 39(7):3207--3229, 2010,
  \href{http://arxiv.org/abs/0710.0655}{{\ttfamily arXiv:0710.0655}}.

\bibitem{NPA08}
M.~Navascu{\'e}s, S.~Pironio, and A.~Ac{\'i}n.
\newblock A convergent hierarchy of semidefinite programs characterizing the
  set of quantum correlations.
\newblock {\em New J. Phys.}, 10(7):073013, 2008,
  \href{http://arxiv.org/abs/0803.4290}{{\ttfamily arXiv:0803.4290}}.

\bibitem{OV16}
D.~Ostrev and T.~Vidick.
\newblock Entanglement of approximate quantum strategies in {XOR} games, 2016,
  \href{http://arxiv.org/abs/1609.01652}{{\ttfamily arXiv:1609.01652}}.

\bibitem{PWPVJ08}
D.~P{\'e}rez-Garc{\'\i}a, M.~M. Wolf, C.~Palazuelos, I.~Villanueva, and
  M.~Junge.
\newblock Unbounded violation of tripartite {B}ell inequalities.
\newblock {\em Communications in Mathematical Physics}, 279(2):455--486, 2008,
  \href{http://arxiv.org/abs/quant-ph/0702189}{{\ttfamily
  arXiv:quant-ph/0702189}}.

\bibitem{PS16}
B.~Pittel and G.~B. Sorkin.
\newblock The satisfiability threshold for k-xorsat.
\newblock {\em Combinatorics, Probability and Computing}, 25(2):236--268, 2016,
   \href{http://arxiv.org/abs/1212.1905}{{\ttfamily arXiv:1212.1905}}.

\bibitem{RUV13}
B.~W. Reichardt, F.~Unger, and U.~Vazirani.
\newblock Classical command of quantum systems.
\newblock {\em Nature}, 496(7446):456--460, 2013.

\bibitem{schrijver86}
A.~Schrijver.
\newblock {\em Theory of Linear and Integer Programming}.
\newblock Wiley, Chichester, 1986.

\bibitem{Slofstra16}
W.~Slofstra.
\newblock {T}sirelson's problem and an embedding theorem for groups arising
  from non-local games, 2016,
  \href{http://arxiv.org/abs/1606.03140}{{\ttfamily arXiv:1606.03140}}.

\bibitem{Slofstra17}
W.~Slofstra.
\newblock {T}he set of quantum correlations is not closed, 2017,
  \href{http://arxiv.org/abs/1703.08618}{{\ttfamily arXiv:1703.08618}}.

\bibitem{Stanley-Catalan}
R.~P. Stanley.
\newblock {\em Enumerative Combinatorics, vol. 2}.
\newblock Cambridge University Press, 1999.
\newblock Exercise 6.36 and references therein.

\bibitem{trevisan2012khot}
L.~Trevisan.
\newblock On {K}hot's unique games conjecture.
\newblock {\em Bulletin (New Series) of the American Mathematical Society},
  49(1), 2012.

\bibitem{tsirel1987quantum}
B.~S. Tsirel'son.
\newblock Quantum analogues of the {Bell} inequalities. {The} case of two
  spatially separated domains.
\newblock {\em Journal of Mathematical Sciences}, 36(4):557--570, 1987.

\bibitem{VV14}
U.~Vazirani and T.~Vidick.
\newblock Fully device-independent quantum key distribution.
\newblock {\em Physical review letters}, 113(14):140501, 2014.

\bibitem{Vid13}
T.~Vidick.
\newblock Three-player entangled {XOR} games are {NP}-hard to approximate.
\newblock In {\em Proceedings of the 2013 IEEE 54th Annual Symposium on
  Foundations of Computer Science}, FOCS '13, pages 766--775. IEEE Computer
  Society, 2013,  \href{http://arxiv.org/abs/1302.1242}{{\ttfamily
  arXiv:1302.1242}}.

\end{thebibliography}
\end{document}